\pdfoutput=1
\RequirePackage{ifpdf}
\ifpdf 
\documentclass[pdftex]{sigma}
\else
\documentclass{sigma}
\fi

\usepackage{mathrsfs, mathtools}
\usepackage{bm, stmaryrd}
\usepackage[mathcal]{euscript}
\usepackage[all]{xy}

\usepackage{tikz}
\usetikzlibrary{cd}
\usetikzlibrary{matrix,calc}
\usetikzlibrary{decorations.markings,shapes.geometric,shapes.misc}
\usetikzlibrary{decorations.pathreplacing,positioning}
\usetikzlibrary{arrows}

\tikzset{cross/.style={cross out, draw=black, minimum size=2*(#1-\pgflinewidth), inner sep=0pt, outer sep=0pt}, cross/.default={1pt}}

\tikzset{cross/.style={cross out, draw=black, minimum size=2*(#1-\pgflinewidth), inner sep=0pt, outer sep=0pt}, cross/.default={1pt}}

\newcommand{\tensor}[1]{{\mathfrak{#1}}}
\newcommand{\up}[1]{[#1]}
\newcommand{\uz}[1]{(#1)}

\DeclareMathOperator{\res}{res}

\DeclareMathOperator{\End}{End}

\DeclareMathOperator{\im}{im}

\DeclareMathOperator{\Ad}{Ad}
\DeclareMathOperator{\ad}{ad}

\newcommand{\longhookrightarrow}{\lhook\joinrel\relbar\joinrel\rightarrow}

\newcommand{\g}{\mathfrak{g}}
\newcommand{\C}{\mathcal{C}}
\newcommand{\CP}{\mathbb{P}^1}

\newcommand{\Oc}{\mathcal{O}}
\newcommand{\Ec}{\mathcal{E}}
\newcommand{\da}{\mathfrak{d}}

\newcommand{\tp}{\null^t}

\newcommand{\Pc}{\mathcal{P}}
\newcommand{\Pb}{\overline{\mathcal{P}}}

\newcommand{\diag}{\text{diag}}

\newcommand{\e}[2]{e^{#1}_{{\color{white}#1}#2}}

\def\d{\mathfrak{d}}
\def\fa{\mathfrak{f}}
\def\g{\mathfrak{g}}

\def\id{\textup{id}}

\def\k{\mathfrak{k}}

\def\D{\mathcal{D}}

\def\L{\mathcal{L}}
\def\T{\mathcal{T}}

\def\C{\mathcal{C}}
\def\J{\mathcal{J}}

\def\CS{\mathrm{CS}}

\def\jb{{\bm j}}
\def\jbt{\bm{\tilde{\bm \jmath}}}
\def\bp{{\bm p}}

\newcommand{\bsb}[1]{\bm{[} #1 \bm{]}}
\newcommand{\brb}[1]{\bm{(} #1 \bm{)}}

\def\bzeta{{\bm \zeta}}
\def\bz{{\bm z}}
\def\bZ{\bm{\mathcal{Z}}}
\def\gzeta{\g^{\brb{\bzeta}}}
\def\gz{\g^{\bsb{\bz}}}
\def\Gz{G^{\bsb{\bz}}}
\def\gzp{\g^{\bsb{\bz'}}}
\def\Gzp{G^{\bsb{\bz'}}}

\newcommand{\lau}[1]{(\kern-.2em( #1 )\kern-.2em)}

\newcommand{\ms}[1]{\mathsf{#1}}

\def\CC{\mathbb{C}}
\def\CP{\mathbb{C}P^1}

\def\RR{\mathbb{R}}

\def\ZZ{\mathbb{Z}}

\def\ii{{\rm i}}

\def\1{\tensor{1}}
\def\2{\tensor{2}}
\def\3{\tensor{3}}
\def\4{\tensor{4}}

\numberwithin{equation}{section}
\newtheorem{Theorem}{Theorem}[section]
\newtheorem{Corollary}[Theorem]{Corollary}
\newtheorem{Lemma}[Theorem]{Lemma}
\newtheorem{Proposition}[Theorem]{Proposition}
 { \theoremstyle{definition}

\newtheorem{Remark}[Theorem]{Remark} }

\begin{document}

\newcommand{\arXivNumber}{2011.13809}

\renewcommand{\PaperNumber}{058}

\FirstPageHeading

\ShortArticleName{Integrable $\mathcal{E}$-models from 4d Chern--Simons Theory}

\ArticleName{Integrable $\boldsymbol{\mathcal{E}}$-Models, 4d Chern--Simons Theory\\
and Affine Gaudin Models. I.~Lagrangian Aspects}

\Author{Sylvain LACROIX~$^{\rm ab}$ and Beno\^it VICEDO~$^{\rm c}$}
\AuthorNameForHeading{S.~Lacroix and B.~Vicedo}

\Address{$^{\rm a)}$~II. Institut f\"ur Theoretische Physik, Universit\"at Hamburg,\\
\hphantom{$^{\rm a)}$}~Luruper Chaussee 149, 22761 Hamburg, Germany}
\EmailD{\href{mailto:sylvain.lacroix@desy.de}{sylvain.lacroix@desy.de}}

\Address{$^{\rm b)}$~Zentrum f\"ur Mathematische Physik, Universit\"at Hamburg,\\
\hphantom{$^{\rm b)}$}~Bundesstrasse 55, 20146 Hamburg, Germany}

\Address{$^{\rm c)}$~Department of Mathematics, University of York, York YO10 5DD, UK}
\EmailD{\href{mailto:benoit.vicedo@gmail.com}{benoit.vicedo@gmail.com}}

\ArticleDates{Received December 07, 2020, in final form May 31, 2021; Published online June 10, 2021}

\Abstract{We construct the actions of a very broad family of 2d integrable $\sigma$-models. Our starting point is a universal 2d action obtained in [arXiv:2008.01829] using the framework of Costello and Yamazaki based on 4d Chern--Simons theory. This 2d action depends on a~pair of 2d fields~$h$ and~$\mathcal{L}$, with $\mathcal{L}$ depending rationally on an auxiliary complex parameter, which are tied together by a constraint. When the latter can be solved for $\mathcal{L}$ in terms of~$h$ this produces a 2d integrable field theory for the 2d field $h$ whose Lax connection is given by~$\mathcal{L}(h)$. We~construct a general class of solutions to this constraint and show that the resulting 2d integrable field theories can all naturally be described as $\mathcal{E}$-models.}


\Classification{17B80; 37K05; 37K10}

\vspace{-4mm}

\section{Introduction}

The $\Ec$-model, introduced by Klim\v{c}\'{\i}k and \v{S}evera~\cite{Klimcik:1995ux, Klimcik:1996nq, Klimcik:1995dy}, makes manifest the duality between pairs of $\sigma$-models related by Poisson--Lie $T$-duality.
Let $D$ be an even dimensional real Lie group whose Lie algebra $\da$ is equipped with a non-degenerate symmetric invariant bilinear form $\langle\!\langle \cdot, \cdot \rangle\!\rangle_{\da}$, i.e., $(\da, \langle\!\langle \cdot, \cdot \rangle\!\rangle_{\da})$ is a quadratic Lie algebra. The $\Ec$-model describes the dynamics of a $D$-valued field $l \in C^\infty(\Sigma, D)$ on a $2$d worldsheet which we take here to be $\Sigma = \RR^2$.
The key ingredient entering the action of the $\Ec$-model, and which gives the model its name, is an invertible linear operator $\Ec\colon \da \to \da$ which is symmetric with respect to the bilinear form $\langle\!\langle \cdot, \cdot \rangle\!\rangle_{\da}$. One often also assumes that $\Ec$ is an involution, i.e., that $\Ec^2 = \id$, which is related to the relativistic invariance of the $\sigma$-models.

Given a maximal isotropic subalgebra $\k \subset \da$ with respect to $\langle\!\langle \cdot, \cdot \rangle\!\rangle_{\da}$ (throughout the paper will refer to maximal isotropic subalgebras as \emph{Lagrangian} subalgebras) with corresponding connected Lie subgroup $K \subset D$ one can associate with the $\Ec$-model on $D$ a $\sigma$-model on the left coset $K \backslash D$. In~particular, if $\k, \tilde{\k} \subset \da$ is a pair of complementary Lagrangian subalgebras with corresponding connected Lie subgroups $K, \tilde K \subset D$ then the associated $\sigma$-models on $K \backslash D$ and $\tilde K \backslash D$ are Poisson--Lie $T$-dual. The situation is summarised in the diagram
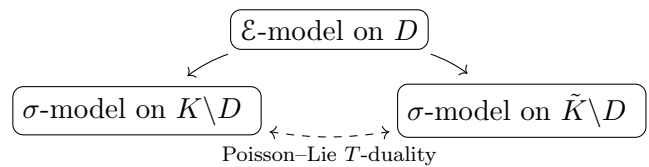
\begin{equation*}
\hspace{75pt}
\begin{tikzpicture}[scale=0.9]
 \node (a) at (0,0) [align=center, draw, rounded corners] {$\Ec$-model on $D$};
 \node (b) at (-1,-1.2) [left, align=center, draw, rounded corners]
 {
 $\sigma$-model on $K \backslash D$
 };
 \node (c) at (1,-1.2) [right, align=center, draw, rounded corners]
 {
 $\sigma$-model on $\tilde{K} \backslash D$
 };

 \draw[->, shorten >=6pt, shorten <=5pt] (a) to [bend right=10] (b);
 \draw[->, shorten >=6pt, shorten <=5pt] (a) to [bend left=10] (c);
 \draw[<->, dashed, shorten >=3pt, shorten <=3pt] (b) to [bend right=9] node[below]{\scriptsize Poisson--Lie $T$-duality} (c);
\end{tikzpicture}
\end{equation*}
in which each arrow represents a canonical transformation relating the phase spaces and Hamiltonians of the respective theories~\cite{Klimcik:1995ux, Klimcik:1995dy, Sfetsos:1996xj, Sfetsos:1997pi}. It~is in this sense that the $\sigma$-models on $K \backslash D$ and $\tilde K \backslash D$ are often referred to as $\Ec$-models themselves.

Even though the $\Ec$-model was devised as a means of understanding Poisson--Lie $T$-duality, it turns out that many of the known integrable $\sigma$-models can be described as $\Ec$-models~\cite{Hoare:2020mpv, Klimcik:2015gba, Klimcik:2017ken, Klimcik:2020fhs} or in terms of their close relatives called dressing cosets or degenerate $\Ec$-models~\cite{Klimcik:1996np}, as in~\cite{Klimcik:2019kkf}.\footnote{See also~\cite{Demulder:2017zhz, Demulder:2018lmj, Demulder:2019bha, Sfetsos:2015nya} for related works on integrable aspects of $\Ec$-models.} In fact, the prototypical class of integrable deformations, given by the Yang--Baxter $\sigma$-model, was originally conceived in~\cite{Klimcik:2002zj} as an example of a $\sigma$-model exhibiting Poisson--Lie symmetry which allows it to be $T$-dualised. It~was shown only some years later that it was integrable in~\cite{Klimcik:2008eq}. Another important class of integrable deformations, given by the $\lambda$-model and constructed in~\cite{Sfetsos:2013wia}, is related to the Poisson--Lie $T$-dual of the Yang--Baxter $\sigma$-model by analytic continuation~\cite{Hoare:2015gda, Klimcik:2015gba, Vicedo:2015pna}.

In light of the above observations, it is natural to ask under which conditions a given \mbox{$\Ec$-model} is also integrable or, conversely, under which conditions an integrable $\sigma$-model can be recast as an $\Ec$-model. A natural starting point is to recall that the equations of motion of the \mbox{$\Ec$-model} are equivalent to the flatness of a $\d$-valued current~$\J$. It~was observed in~\cite{Severa:2017kcs} that if one can find a~linear map $p_z\colon \da \to \g^\CC$, with $\g^\CC$ a complex Lie algebra, depending rationally on a complex para\-me\-ter $z$ and satisfying a certain algebraic property reviewed in Section~\ref{sec: inverse jz}, then it can be used to lift the on-shell flat connection~$\J$ to an on-shell flat meromorphic $\g^\CC$-valued connection~$p_z(\J)$ thus defining a Lax connection for the model. In~particular, this observation was used to rederive the integrability of the Yang--Baxter $\sigma$-model and the $\lambda$-model in the framework of~$\Ec$-models. However, the problem of constructing more general families of suitable maps $p_z$ so as to produce new examples of integrable $\sigma$-models from $\Ec$-models was left open in~\cite{Severa:2017kcs}. One upshot of the present work is a systematic construction of such maps~$p_z$.

The purpose of this paper is to construct a very general family of $2$d integrable field theories, or more precisely integrable $\sigma$-models, and show that they all naturally admit descriptions as $\Ec$-models. In~fact, the identification of a suitable symmetric invertible linear operator $\Ec \colon \da \to \da$ on~the relevant quadratic Lie algebra $(\da, \langle\!\langle \cdot, \cdot \rangle\!\rangle_{\da})$ is a key part of our construction of these integrable $\sigma$-models.

In order to construct this broad family of integrable $\Ec$-models, we shall put to full use two general frameworks for describing $2$d classical integrable field theories which have emerged over the last couple of years. The first, initially proposed in~\cite{Vicedo:2017cge} and then further developed in~\cite{Delduc:2019bcl, Lacroix:2019xeh}, is based on classical dihedral affine Gaudin models. The second, proposed by Costello and Yamazaki~\cite{Costello:2019tri}, is based instead on $4$d Chern--Simons theory which was originally developed for describing integrable spin chains in~\cite{Costello:2013zra, Costello:2013sla, Costello:2017dso, Costello:2018gyb, Witten:2016spx}. See also~\cite{Ashwinkumar:2020krt, Bittleston:2020hfv, Bykov1, Bykov2, Bykov3, Costello:2020lpi, Delduc:2019whp, Fukushima:2020kta, Fukushima:2020dcp, Penna:2020uky, Schmidtt:2019otc, Tian:2020ryu, Tian:2020meg} for further recent developments in the field theory setting. Although the frameworks of~\cite{Vicedo:2017cge} and~\cite{Costello:2019tri} are very different in flavour, they are in fact intimately related~\cite{Vicedo:2019dej}: they are based on the Hamiltonian and Lagrangian formalisms respectively. In~this paper we will construct the actions of integrable $\Ec$-models starting from the action of $4$d Chern--Simons theory\footnote{Note that a relation between 4d Chern--Simons theory and $\Ec$-models was already described in~\cite{Delduc:2019bcl}. However, this relation is very different in nature from the one described in the present article since it applies to all $\Ec$-models, regardless of whether they are integrable or not.} but using also input from the theory of affine Gaudin models. The Hamiltonian analysis of these actions and their relation to classical dihedral affine Gaudin models will be considered elsewhere~\cite{InPrep}.

The Lagrangian of $4$d Chern--Simons theory is proportional to $\omega \wedge \CS(A)$, where $\CS(A)$ is the Chern--Simons $3$-form for a $\g^\CC$-valued $1$-form $A$ on $\Sigma \times \CP$ and $\omega$ is a meromorphic $1$-form on $\CP$. In~the setup of~\cite{Costello:2019tri}, $2$d integrable field theories are described by introducing surface defects along $\Sigma$ at the poles of $\omega$ on $\CP$.
When $\omega$ has at most double poles, this approach was used in~\cite{Delduc:2019whp} to construct a unifying $2$d action for many known $2$d integrable $\sigma$-models. The generalisation of this $2$d action for an arbitrary meromorphic $1$-form $\omega$ was obtained in~\cite{Benini:2020skc}, where the passage from $4$d Chern--Simons theory to $2$d integrable $\sigma$-models was streamlined and put on a firm mathematical footing using methods from homotopical algebra.

More precisely, the $2$d actions derived in~\cite{Delduc:2019whp} and~\cite{Benini:2020skc} are both actions for a certain group valued field $h$ living on $\Sigma$ but which also depend on a $1$-form $\L$ on $\Sigma$ that depends meromorphically on $\CP$. In~order to obtain a $2$d action for the field $h$ alone one still needs to solve a certain boundary condition, or constraint, relating $\L$ to $h$ and depending on a choice of Lagrangian subalgebra $\k$ of a certain quadratic Lie algebra $\da$ determined by $\omega$. Given any solution $\L = \L(h)$ of this constraint, one obtains a $2$d integrable field theory for the $2$d field $h$. The connection $\L(h)$ then plays the role of the Lax connection of this $2$d integrable field theory. The main purpose of this paper is to solve the boundary condition relating $\L$ and $h$ in the general setting of~\cite{Benini:2020skc}. In~doing so, we are naturally led to introduce a linear operator $\Ec$ on the Lie algebra $\da$, with all the properties required to define an $\Ec$-model. In~fact, this linear operator has a very natural origin from the point of view of affine Gaudin models and our construction of this operator is motivated by~\cite{Delduc:2019bcl}. We~find that, upon solving the constraint, the $2$d action of~\cite{Benini:2020skc} coincides with that of the $\sigma$-model on the coset $K \backslash D$ associated with an $\Ec$-model.

At this point it is useful to highlight the various levels of generality of our setup.

Firstly, all the integrable $\sigma$-models which have so far been constructed using either the framework of $4$d Chern--Simons theory or that of dihedral affine Gaudin models, with the exception of an example considered recently in~\cite{Bittleston:2020hfv}, start from a choice of meromorphic $1$-form $\omega$ which has at most double poles. By exploiting the results of~\cite{Benini:2020skc}, in the present work we build integrable $\sigma$-models starting from a completely general $1$-form $\omega$ (although for technical reasons to be discussed in the main text, we require $\omega$ to have one double pole at infinity). To illustrate the effect of higher order poles in $\omega$ we give an explicit example in which $\omega$ has a fourth order pole.

Secondly, the constraint on $\L$ and $h$, which depends on a Lagrangian subalgebra $\k$ of $\da$, has so far only been solved for a limited number of concrete examples. In~the present article, we solve this constraint for arbitrary Lagrangian subalgebras $\k$. Even in the case when $\omega$ has at most double poles, this generality on $\k$ allows us, for instance, to obtain the non-abelian $T$-dual of the principal chiral model from $4$d Chern--Simons theory, as anticipated in~\cite{Delduc:2019whp}.

To end this introduction, we will illustrate in a simple case the general family of integrable $\Ec$-models constructed in the main text.

Let $\g$ be a real Lie algebra equipped with a non-degenerate symmetric invariant bilinear form $\langle \cdot, \cdot \rangle \colon \g \times \g \to \RR$. Consider the meromorphic $1$-form
\begin{gather} \label{omega intro}
\omega = - \ell^\infty_1 \frac{\prod_{i=1}^N (z - \zeta_i)}{\prod_{i=1}^N (z - z_i)} {\rm d}z,
\end{gather}
with distinct real poles and zeroes $z_i, \zeta_i \in \RR$ for $i = 1, \dots, N$. Note, in particular, that $\omega$ has a~double pole at infinity. In~the main text we shall take $\omega$ to have an arbitrary number of finite poles of arbitrary order, but also restrict to the case of a double pole at infinity.

First, we define the quadratic Lie algebra $(\da, \langle\!\langle \cdot, \cdot \rangle\!\rangle_{\da})$. As we shall see, since the $1$-form~$\omega$ in~\eqref{omega intro} has only simple poles along the real axis, the associated Lie algebra $\da$ is given in~this case by the direct sum of Lie algebras $\da = \g^{\oplus N}$, which is the Lie algebra of $D = G^{\times N}$. More\-over,~$\da$ comes equipped with a natural non-degenerate symmetric invariant bilinear form defined in~terms of $\omega$ by
\begin{gather} \label{bilinear intro}
\langle\!\langle \cdot, \cdot \rangle\!\rangle_{\da} \colon\quad \da \times \da \longrightarrow \RR, \qquad
\big\langle{\mkern-4mu}\big\langle (\ms u^i)_{i = 1}^N, (\ms v^j)_{j=1}^N \big\rangle{\mkern-4mu}\big\rangle_{\da} = \sum_{i=1}^N (\res_{z_i} \omega) \langle \ms u^i, \ms v^i \rangle
\end{gather}
for any pair of elements $(\ms u^i)_{i = 1}^N, (\ms v^j)_{j=1}^N \in \da$. More generally, if $\omega$ has higher order poles then the corresponding copies of $\g$ in $\da$ are replaced by truncated loop algebras over $\g$ (or its complexification if the pole is not real) and the bilinear form~\eqref{bilinear intro} is replaced by one involving all coefficients in the partial fraction decomposition of $\omega$.

Next, we define the linear map $\Ec \colon \da \to \da$. We~associate to each zero $\zeta_i$ of $\omega$ an $\epsilon_i \in \RR \setminus \{ 0 \}$. We~then define
\begin{gather} \label{E intro}
\Ec (\ms u^i)_{i = 1}^N = \bigg( \sum_{j,k=1}^N \frac{\prod_{r \neq j} (\zeta_r - z_k) \prod_{r \neq i} (z_r - \zeta_j)}{\prod_{r \neq k} (z_r - z_k) \prod_{r \neq j} (\zeta_r - \zeta_j)} \epsilon_j \ms u^k \bigg)_{i = 1}^N,
\end{gather}
for every $(\ms u^i)_{i = 1}^N \in \da$. This operator is invertible because $\epsilon_i \neq 0$ and it is symmetric with respect to~\eqref{bilinear intro}. Moreover, if $\g$ is compact and the sign of each $\epsilon_i$ is chosen to coincide with the sign of~$- \varphi'(\zeta_i)$, where $\varphi(z)$ is the twist function defined by $\omega = \varphi(z) {\rm d}z$, then $\Ec$ is positive with respect to~\eqref{bilinear intro}. This ensures that the Hamiltonian is positive. Furthermore, if $\epsilon_i^2 =1$ for $i = 1, \dots, N$ then $\Ec^2 = \id$ which ensures that the model is relativistic invariant.

Finally, let $\k$ be a Lagrangian subalgebra of $\da$ and $K$ the associated connected Lie subgroup of~$D$. (We also require that $\k$ satisfies a technical condition together with the operator $\Ec$, which can be easily ensured for instance if $\Ec$ is positive; we refer to the main text for details.) We~can then construct the $\sigma$-model for a field $l \in C^\infty(\Sigma, D)$ with a gauge symmetry by $K$, from the $\Ec$-model associated with the above data.
By our construction, this $\sigma$-model on $K \backslash D$ is integrable and its Lax connection, depending rationally on the spectral parameter $z$, is given explicitly in~the present case by
\begin{gather} \label{Lax intro}
\L(z) = \sum_{i,j=1}^N \frac{\prod_{r \neq i} (\zeta_r - z_j) \prod_r (z_r - \zeta_i)}{\prod_{r \neq j} (z_r - z_j) \prod_{r \neq i} (\zeta_r - \zeta_i)} \frac{\J^j}{z-\zeta_i}
\end{gather}
for a certain $\g$-valued 1-form $\J^j$ on $\Sigma$ for each $j = 1, \dots, N$ depending on the $D$-valued field~$l$. We~refer to Section~\ref{sec: intro example} for details. The equations of motion of the $\Ec$-model, given in this case by the flatness of $\L(z_i) = \mathcal J^i$ for each $i = 1, \dots, N$, is equivalent to the flatness of the above Lax connection for all $z$.

The above example contains the Yang--Baxter $\sigma$-model, the $\lambda$-model and more generally the family of integrable $\sigma$-models constructed in~\cite{Bassi:2019aaf} which couple together $N_1 \in \ZZ_{\geq 0}$ copies of the Yang--Baxter $\sigma$-model and $N_2 \in \ZZ_{\geq 0}$ copies of the $\lambda$-model, where $2N_1 + 2N_2 = N$.

The plan of the paper is as follows. We~begin in Section~\ref{sec: E-models} by reviewing the definition of the $\Ec$-model and the construction of the associated $\sigma$-model on $K \backslash D$. In~particular, we discuss the properties of $\Ec$ and various projectors that will be relevant for our purposes. In~Section~\ref{sec: 2d from 4d} we review the general $2$d action constructed in~\cite{Benini:2020skc} and in Section~\ref{sec: removing inf} we derive from it the $2$d action that serves as the starting point for our analysis. In~Section~\ref{sec: E-model from 4d} we construct a solution to the constraint from~\cite{Benini:2020skc} and relate the resulting $2$d integrable field theory to an $\Ec$-model. Finally, we give a number of simple examples in Section~\ref{sec: examples} to illustrate the general construction before concluding in Section~\ref{sec: outlook}.

\subsubsection*{List of notations}
For the reader's convenience we gather here a list of notations used throughout the paper.
\begin{itemize}\itemsep=0pt
 \item $\g$, $G$ \--- real finite-dimensional Lie algebra and corresponding Lie group,
 \item[] $\langle \cdot, \cdot \rangle$ \--- non-degenerate invariant symmetric bilinear form on $\g$,
 \item $\bz$ \--- set of independent poles of $\omega$,
 \item[] $\bsb{\bz}$ \--- independent poles of $\omega$ counting multiplicities,
 \item[] $\gz$, $\Gz$ \--- associated defect Lie algebra and Lie group,
 \item[] $\langle\!\langle \cdot, \cdot \rangle\!\rangle_{\gz}$ \--- non-degenerate invariant symmetric bilinear form on $\gz$,
 \item[] $\mathfrak{f}$, $F$ \--- Lagrangian subalgebra of $\gz$ and corresponding Lie subgroup of $\Gz$,
 \item $\bz'$ \--- set of independent finite poles of $\omega$,
 \item[] $\bsb{\bz'}$ \--- independent finite poles of $\omega$ counting multiplicities,
 \item[] $\d$, $D$ \--- associated defect Lie algebra and Lie group,
 \item[] $\langle\!\langle \cdot, \cdot \rangle\!\rangle_{\d}$ \--- non-degenerate invariant symmetric bilinear form on $\d$,
 \item[] $\k$, $K$ \--- Lagrangian subalgebra of $\d$ and corresponding Lie subgroup of $D$,
 \item $\bzeta$ \--- set of independent zeroes of $\omega$,
 \item[] $\brb{\bzeta}$ \--- independent zeroes of $\omega$ counting multiplicities,
 \item[] $\gzeta$ \--- associated vector space,
 \item[] $\langle\!\langle \cdot, \cdot \rangle\!\rangle_{\gzeta}$ \--- non-degenerate symmetric bilinear form on $\gzeta$.
\end{itemize}

\section[Background on the E-model] {Background on the $\boldsymbol{\Ec}$-model}\label{sec: E-models}

Throughout this section we let $\da$ denote an arbitrary real even dimensional Lie algebra equipped with a non-degenerate ad-invariant symmetric bilinear form
\begin{gather*}
\langle\!\langle \cdot, \cdot \rangle\!\rangle_\da \colon\ \da \times \da \longrightarrow \RR.
\end{gather*}
We let $D$ be a real Lie group with Lie algebra $\da$. In~Section~\ref{sec: removing inf} below we shall introduce a~specific real even dimensional Lie algebra $\da$ and corresponding Lie group $D$, to which the results of the present section will apply verbatim.

For any linear operator $\Oc \in \End \da$ we denote by $\tp\Oc$ its transpose with respect to $\langle\!\langle \cdot, \cdot \rangle\!\rangle_\da$, namely such that
$\langle\!\langle \ms U, \Oc \ms V \rangle\!\rangle_\da = \langle\!\langle \tp\Oc \ms U, \ms V \rangle\!\rangle_\da$,
for any $\ms U, \ms V \in \da$. We~also let $\Ad_l$ denote the adjoint action $\Ad_l \ms U \coloneqq l \ms U l^{-1}$ of $l \in D$ on $\da$.

\subsection[The operators E and Pl]
{The operators $\boldsymbol{\Ec}$ and $\boldsymbol{\Pc_l}$}\label{sec: operators E Pl}

We fix an invertible operator $\Ec \in \End \da$, symmetric with respect to $\langle\!\langle \cdot, \cdot \rangle\!\rangle_\da$, i.e., $\tp\Ec=\Ec$, and let
\begin{gather}\label{Eq:FE}
\langle\!\langle \ms U, \ms V \rangle\!\rangle_{\da, \Ec} \coloneqq \langle\!\langle \ms U, \Ec^{-1} \ms V \rangle\!\rangle_\da
\end{gather}
for every $\ms U, \ms V \in \da$. This defines another non-degenerate bilinear form on $\da$.

We suppose that $\da$ admits a Lagrangian subalgebra $\k$ with respect to $\langle\!\langle \cdot, \cdot \rangle\!\rangle_\da$, which thus satisfies $\dim\k = \frac{1}{2}\dim\da$.

We shall need to make another important assumption on the operator $\Ec$ and the Lagrangian subalgebra $\k \subset \da$. Namely, we will suppose that for any $l\in D$ we have
\begin{gather}\label{Eq:Assump}
\Ad_l^{-1} \k \, \cap \, \Ec \Ad_l^{-1} \k = \lbrace 0 \rbrace.
\end{gather}
Before exploring the consequences of this assumption, we give a sufficient condition on the operator $\Ec$ for the condition~\eqref{Eq:Assump} to hold.

\begin{Remark}
As we shall see later in Section~\ref{sec: EM tensor E-model}, in particular Remark~\ref{rem: rel inv E-model}, this sufficient condition is actually quite natural in the study of $\sigma$-models, as it will be related to the property of the Hamiltonian being bounded below in these models.
\end{Remark}

\begin{Lemma}\label{Lem:Pos}
If $\Ec$ is such that $\langle\!\langle \cdot, \cdot \rangle\!\rangle_{\da, \Ec}$ is positive-definite, any Lagrangian subalgebra $\k \subset \da$ satisfies~\eqref{Eq:Assump}.
\begin{proof}
Let us fix $l \in D$. By applying successively the definition~\eqref{Eq:FE} of $\langle\!\langle \cdot, \cdot \rangle\!\rangle_{\da, \Ec}$, the ad-in\-va\-riance of $\langle\!\langle \cdot, \cdot \rangle\!\rangle_\da$ and the isotropy of $\k$, we have
\begin{gather*}
\langle\!\langle \Ad_l^{-1} \k, \Ec\Ad_l^{-1} \k \rangle\!\rangle_{\da, \Ec} = \langle\!\langle \Ad_l^{-1} \k, \Ad_l^{-1} \k \rangle\!\rangle_\da = \langle\!\langle \k, \k \rangle\!\rangle_\da = 0.
\end{gather*}
Hence the subspaces $\Ad_l^{-1}\k$ and $\Ec\Ad_l^{-1}\k$ are orthogonal with respect to $\langle\!\langle \cdot, \cdot \rangle\!\rangle_{\da, \Ec}$. If the latter is positive-definite, these subspaces then have trivial intersection.
\end{proof}
\end{Lemma}

Since the operators $\Ec$ and $\Ad_l^{-1}$ are both invertible, we have
\begin{gather*}
\dim \Ad_l^{-1} \k = \dim \Ec \Ad_l^{-1} \k = \dim \k = \frac{1}{2} \dim\da.
\end{gather*}
By the assumption~\eqref{Eq:Assump}, we thus have the vector space direct sum decomposition (here we explicitly use the assumption that $\k \subset \d$ is Lagrangian)
\begin{gather}\label{Eq:Direct}
\da = \Ad_l^{-1} \k \, \dotplus \, \Ec \Ad_l^{-1} \k.
\end{gather}
As observed in the proof of Lemma~\ref{Lem:Pos}, this direct sum is orthogonal with respect to $\langle\!\langle \cdot, \cdot \rangle\!\rangle_{\da, \Ec}$. We~define the projector $\Pc_l$ relative to~\eqref{Eq:Direct} with kernel and image
\begin{gather}\label{Eq:KerIm}
\ker \Pc_l = \Ad_l^{-1} \k \qquad \text{and} \qquad \im \Pc_l = \Ec \Ad_l^{-1} \k.
\end{gather}
It will also be convenient to introduce the operator
\begin{gather}\label{Eq:Pb}
\Pb_l \coloneqq \id - \tp\Pc_l,
\end{gather}
which is easily seen to define another projector with kernel and image
\begin{gather}\label{Eq:KerImbarPb}
\ker \Pb_l = \Ad_l^{-1} \k \qquad \text{and} \qquad \im \Pb_l = \Ec^{-1} \Ad_l^{-1} \k.
\end{gather}
To see the first equation, observe that $\ker\Pb_l=\ker(\id-\tp\Pc_l)=\im \tp\Pc_l$ which is equal to the subspace orthogonal to $\ker \Pc_l = \Ad_k^{-1} \k$ with respect to $\langle\!\langle \cdot, \cdot \rangle\!\rangle_\da$ and therefore to $\Ad_l^{-1}\k$ itself since it is Lagrangian. Similarly, one finds that $\im \Pb_l$ is the subspace orthogonal to $\im \Pc_l = \Ec \Ad_l^{-1} \k$ which one checks is given by $\Ec^{-1} \Ad_l^{-1} \k$.

We shall need the following technical properties of the projectors $\Pc_l$ and $\Pb_l$, whose proof we give in Appendix~\ref{sec: PropE}.

\begin{Proposition}\label{PropE}
The projectors $\Pc_l$ and $\Pb_l$ have the following properties:
\begin{itemize}\itemsep=0pt
\item[$(i)$] $\tp \Pc_l = \Ec^{-1} \Pc_l \Ec$,
\item[$(ii)$] $\Pc_l \Ec + \Ec \Pb_l = \Ec$ and $\Ec^{-1} \Pc_l + \Pb_l \Ec^{-1} = \Ec^{-1}$,
\item[$(iii)$] $\tp \Pc_l \Pb_l = \tp \Pb_l \Pc_l = 0$,
\item[$(iv)$] $\Pc_l - \Pb_l=\tp \Pc_l \Pc_l = -\tp \Pb_l \Pb_l$,
\item[$(v)$] $\Pc_l\tp\Pc_l=\Pb_l \tp\Pb_l = 0$,
\item[$(vi)$] $\Pb_l = \Pc_l$ if $\Ec^2 = \id$.
\end{itemize}
\end{Proposition}

\subsection[The E-model action] {The $\boldsymbol{\Ec}$-model action}\label{sec: Emodel action}

Define $\Sigma\coloneqq\mathbb{R}^2$ on which we fix coordinates $(\tau,\sigma)$ for convenience. The $\Ec$-model describes the dynamics of a $D$-valued field $\ell \in C^\infty(\Sigma, D)$ on $\Sigma$ with the first order action~\cite{Klimcik:1995ux, Klimcik:1996nq, Klimcik:1995dy}
\begin{gather} \label{E-model action 0}
S_{\Ec}(\ell) = \frac{1}{2} \int_\Sigma \bigl( \langle\!\langle \ell^{-1} \partial_\tau \ell, \ell^{-1}\partial_\sigma \ell \rangle\!\rangle_\da - \langle\!\langle \ell^{-1} \partial_\sigma \ell, \Ec\, \ell^{-1}\partial_\sigma \ell \rangle\!\rangle_\da \bigr) {\rm d}\sigma \wedge {\rm d}\tau - \frac{1}{2} I^{\rm WZ}_{\da}[\ell].
\end{gather}
Here we introduce the standard WZ-term for $\ell$ as\footnote{Here, we follow the conventions of~\cite{Delduc:2019bcl,Delduc:2019whp} for the definition of WZ-terms.}
\begin{gather} \label{WZ term}
I^{\rm WZ}_{\da}[\ell] \coloneqq -\frac{1}{6} \int_{\Sigma \times I} \big\langle{\mkern-4mu}\big\langle \widehat{\ell}^{-1} {\rm d} \widehat{\ell}, \big[\widehat{\ell}^{-1} {\rm d} \widehat{\ell}, \widehat{\ell}^{-1} {\rm d} \widehat{\ell}\big] \big\rangle{\mkern-4mu}\big\rangle_{\da},
\end{gather}
where $I \coloneqq [0, 1]$ and $\widehat{\ell} \in C^\infty(\Sigma \times I, D)$ is any smooth extension of $\ell$ to $\Sigma \times I$ with the property that $\widehat{\ell} = \ell$ near $\Sigma \times \{ 0 \} \subset \Sigma \times I$ and $\widehat{\ell} = \id$ near $\Sigma \times \{ 1 \} \subset \Sigma \times I$. The WZ-term $I^{\rm WZ}_{\da}[\ell]$ is independent of the choice of extension $\widehat{\ell}$; see, e.g.,~\cite{Benini:2020skc}.

Let $K$ be the connected Lie subgroup of $D$ corresponding to the Lie subalgebra $\k \subset \da$ from Section~\ref{sec: operators E Pl}. In~this section we recall the derivation of the action for the $\sigma$-model on the left coset $K \backslash D$ starting from the $\Ec$-model action~\eqref{E-model action 0}.

We begin by introducing a new $D$-valued field $l \in C^\infty(\Sigma, D)$ and a $K$-valued field $b\in C^\infty(\Sigma, K)$, then define the action $S'_{\Ec, \k}(l,b) \coloneqq S_{\Ec}(bl)$. Of course, the latter is invariant under the gauge transformation
\begin{gather}\label{Eq:GaugeE}
l \longmapsto kl, \qquad b \longmapsto bk^{-1},
\end{gather}
with local parameter $k \in C^\infty(\Sigma, K)$, and fixing this gauge invariance by imposing the gauge condition $b = \id$ we recover the original action~\eqref{E-model action 0} for the field $l$. However, since we would like to keep the gauge invariance~\eqref{Eq:GaugeE}, so as to obtain a model on $K \backslash D$, we will eliminate $b$ in a~different way.

To compute the action $S'_{\Ec, \k}(l,b)$ explicitly, we make use of the Polyakov--Wiegmann identity~\cite{Polyakov-Wiegmann}
\begin{gather}\label{PW}
I^{\rm WZ}_{\da}[bl] = I^{\rm WZ}_{\da}[b] + I^{\rm WZ}_{\da}[l] + \int_\Sigma \bigl( \langle\!\langle b^{-1} \partial_\tau b, \partial_\sigma l l^{-1} \rangle\!\rangle_\da - \langle\!\langle b^{-1} \partial_\sigma b, \partial_\tau l l^{-1} \rangle\!\rangle_\da \bigr) {\rm d}\sigma \wedge {\rm d}\tau
\end{gather}
and of the fact that
\begin{gather*}
(bl)^{-1} \partial_\mu (bl) = \Ad_l^{-1} b^{-1}\partial_\mu b + l^{-1} \partial_\mu l
\end{gather*}
for $\mu=\tau,\sigma$. By the isotropy of the subalgebra $\k \subset \da$, the WZ-term for the $K$-valued field $b$ and $\langle\!\langle \Ad_l^{-1} b^{-1} \partial_\tau b, \Ad_l^{-1} b^{-1} \partial_\sigma b \rangle\!\rangle_\da = \langle\!\langle b^{-1} \partial_\tau b, b^{-1} \partial_\sigma b \rangle\!\rangle_\da$ both vanish. After a few manipulations, one observes that $S'_{\Ec, \k}(l,b)$ depends on the field $b$ only through $Y=b^{-1} \partial_\sigma b$. More precisely, we have
\begin{gather}
S'_{\Ec, \k}(l,b) = \frac{1}{2} \int_\Sigma \bigl( \langle\!\langle l^{-1} \partial_\tau l, l^{-1}\partial_\sigma l \rangle\!\rangle_\da - \langle\!\langle l^{-1} \partial_\sigma l, \Ec\, l^{-1}\partial_\sigma l \rangle\!\rangle_\da \bigr) {\rm d}\sigma \wedge {\rm d}\tau - \frac{1}{2} I^{\rm WZ}_{\da}[l] \notag
\\ \hphantom{S'_{\Ec, \k}(l,b) =}
{} + \int_{\Sigma} \left( \langle\!\langle Y, \Ad_l ( l^{-1}\partial_\tau l - \Ec \, l^{-1}\partial_\sigma l ) \rangle\!\rangle_\da - \frac{1}{2} \langle\!\langle Y, \Ad_l \Ec \Ad_l^{-1}\, Y \rangle\!\rangle_\da \right) {\rm d}\sigma \wedge {\rm d}\tau, \label{Eq:ActionEInt}
\end{gather}
which is quadratic and algebraic in $Y$. We~can therefore integrate out the degrees of freedom in~the field $b$, or equivalently in $Y$. For~that, we first determine its equation of motion by~com\-pu\-ting the variation of the action under an infinitesimal variation $\delta Y \in C^\infty(\Sigma, \k)$ of $Y$, which reads
\begin{gather*}
\delta S'_{\Ec, \k}(l,b) = \int_{\Sigma} \big( \langle\!\langle \delta Y, \Ad_l \big( l^{-1}\partial_\tau l - \Ec \, l^{-1}\partial_\sigma l - \Ec \Ad_l^{-1}\, Y \big) \rangle\!\rangle_\da \big) {\rm d}\sigma \wedge {\rm d}\tau.
\end{gather*}
The vanishing of the above variation for any $\k$-valued field $\delta Y$ requires that
\begin{gather} \label{Z def}
Z \coloneqq \Ad_l \bigl( l^{-1}\partial_\tau l - \Ec \, l^{-1}\partial_\sigma l - \Ec \Ad_l^{-1}\, Y \bigr)
\end{gather}
belongs to the subspace orthogonal to $\k$ with respect to $\langle\!\langle \cdot,\cdot \rangle\!\rangle_\da$, which coincides with $\k$ itself since~$\k$ is Lagrangian.

To solve the equation of motion $Z\in\k$, we first rewrite~\eqref{Z def} as
\begin{gather}\label{YZ}
\Ec\Ad_l^{-1} Y + \Ad_l^{-1} Z = l^{-1}\partial_\tau l - \Ec \, l^{-1}\partial_\sigma l.
\end{gather}
Recall the projector $\Pc_l$ introduced in Section~\ref{sec: operators E Pl} through its kernel and image~\eqref{Eq:KerIm}. By definition of $Y$, the quantity $\Ec\Ad_l^{-1} Y$ is valued in $\im \Pc_l = \Ec \Ad_l^{-1} \k$. Moreover, the equation of motion $Z\in\k$ is equivalent to $\Ad_l^{-1} Z$ belonging to $\ker \Pc_l = \Ad_l^{-1} \k$. Applying $\Pc_l$ to equation~\eqref{YZ}, we~thus get $\Ec\Ad_l^{-1} Y = \Pc_l \bigl( l^{-1}\partial_\tau l - \Ec \, l^{-1}\partial_\sigma l \bigr)$, hence
\begin{gather*}
Y = \Ad_l \Ec^{-1}\Pc_l \bigl( l^{-1}\partial_\tau l - \Ec \, l^{-1}\partial_\sigma l \bigr).
\end{gather*}
Note that in the above derivation, we have used the existence of the projector $\Pc_l$ and thus the assumption~\eqref{Eq:Assump} made in Section~\ref{sec: operators E Pl}. Reinserting the above expression for $Y$ in the action~\eqref{Eq:ActionEInt}, we finally arrive at an action for the field $l$ alone. After a few manipulations, using part $(ii)$ of~Proposition~\ref{PropE}, we find
\begin{gather}
S_{\Ec,\k}(l) \coloneqq \frac{1}{2} \int_\Sigma \big( \langle\!\langle l^{-1} \partial_\tau l, \Ec^{-1} \Pc_l \big(l^{-1} \partial_\tau l\big) \rangle\!\rangle_\da - \langle\!\langle l^{-1} \partial_\sigma l, \Ec \overline{\Pc}_l \big(l^{-1} \partial_\sigma l\big) \rangle\!\rangle_\da \notag
\\ \hphantom{S_{\Ec,\k}(l) \coloneqq \frac{1}{2} \int_\Sigma \big(}
{} + \langle\!\langle l^{-1} \partial_\tau l, \big(\overline{\Pc}_l - \null^t \Pc_l\big) \big(l^{-1} \partial_\sigma l\big) \rangle\!\rangle_\da \big) {\rm d}\sigma \wedge {\rm d}\tau - \frac{1}{2} I^{\rm WZ}_{\da}[l].
\label{Eq:ActionE}
\end{gather}
Written in this form, the model is relativistic if and only if $\Ec^{-1}\Pc_l = \Ec\Pb_l$ and $\Pb_l - \null^t\Pc_l$ is skew-symmetric. By part $(vi)$ of Proposition~\ref{PropE} we deduce that these conditions are trivially satisfied if $\Ec^2=\id$. In~the latter case the action~\eqref{Eq:ActionE} can be rewritten, making multiple use of properties from Proposition~\ref{PropE}, in the more familiar form
\begin{gather*}
S(f) \coloneqq S_{\Ec,\k}(f^{-1}) = - \frac{1}{2} \int_\Sigma \big\langle{\mkern-4mu}\big\langle f^{-1} \partial_- f, \big(\id - 2 P_f(\Ec)\big) \big(f^{-1} \partial_+ f\big) \big\rangle{\mkern-4mu}\big\rangle_\da {\rm d}\sigma^+ \wedge {\rm d}\sigma^- + \frac{1}{2} I^{\rm WZ}_{\da}[f],
\end{gather*}
where $\partial_\pm \coloneqq \partial_\tau \pm \partial_\sigma$ and $\sigma^\pm = \frac 12 (\tau \pm \sigma)$. Here we have introduced the new $D$-valued field $f \coloneqq l^{-1}$, to match the notation of~\cite{Klimcik:2015gba}, and defined
\begin{gather} \label{Klimcik proj}
P_f(\Ec) \coloneqq \Ad_f^{-1} \Ec \Pc_{f^{-1}} (\Ec - \id) \Ad_f
\end{gather}
which is easily seen to be a projector. Moreover, its kernel and image can be deduced from those of $\Pc_{f^{-1}}$ in~\eqref{Eq:KerIm} to be given by
\begin{gather*}
\ker P_f(\Ec) = \big( \id + \Ad_f^{-1} \Ec \Ad_f \big) \da \qquad \text{and} \qquad \im P_f(\Ec) = \k,
\end{gather*}
which coincide with those given, for instance, in~\cite{Klimcik:2015gba}. Note that it follows from the decom\-po\-sition~\eqref{Eq:Direct} that $\big(\id + \Ad_f^{-1} \Ec \Ad_f\!\big) \da = \big(\id + \Ad_f^{-1} \Ec \Ad_f\!\big) \k$. Indeed,~\eqref{Eq:Direct} implies that $\da = \k \dotplus \big(\Ad_f^{-1} \Ec \Ad_f - \id\big) \k$. The result then follows by acting on both sides with $\big(\id + \Ad_f^{-1} \Ec \Ad_f\!\big)$ and noting that $(\Ec + \id) (\Ec - \id) = 0$.

In what follows we shall only use the alternative form~\eqref{Eq:ActionE} of the action (note that this was previously used in~\cite{Hoare:2020mpv}).

\subsection{Gauge invariance} \label{sec: gauge inv E-model}

The action $S_{\Ec,\k}(l)$ was obtained in Section~\ref{sec: Emodel action} by integrating out the field $b$ from the action~$S'_{\Ec, \k}(l,b)$, which is invariant under the gauge transformation~\eqref{Eq:GaugeE}. Therefore, by construction, $S_{\Ec,\k}(l)$ should be invariant under the residual gauge transformation $l\mapsto kl$ with $k\in C^\infty(\Sigma, K)$. In~this section we check this statement explicitly using the expression~\eqref{Eq:ActionE} of~$S_{\Ec,\k}(l)$.

We will need the gauge transformation of the projectors $\Pc_l$ and $\Pb_l$. By equation~\eqref{Eq:KerIm}, the kernel of $\Pc_{kl}$ is given by
\begin{gather*}
\ker \Pc_{kl} = \Ad_{kl}^{-1} \k = \Ad_l^{-1} \Ad_k^{-1} \k = \Ad_l^{-1} \k = \ker \Pc_l.
\end{gather*}
Similarly, one finds $\im \Pc_{kl} = \im \Pc_l$. We~thus have $\Pc_{kl}=\Pc_l$, hence also $\Pb_{kl}=\Pb_l$, i.e., the projectors $\Pc_l$ and $\Pb_l$ are gauge invariant. Moreover, we have
\begin{gather*}
\Pc_l \bigl( (kl)^{-1}\partial_\tau (kl) \bigr) = \Pc_l \bigl( l^{-1}\partial_\tau l + \Ad_l^{-1} k^{-1} \partial_\tau k \bigr) = \Pc_l \big(l^{-1} \partial_\tau l \big),
\end{gather*}
where we have used the fact that $\Ad_l^{-1} k^{-1} \partial_\tau k$ belongs to $\Ad_l^{-1}\k = \ker\Pc_l$. Therefore we deduce $\langle\!\langle l^{-1} \partial_\tau l, \Ec^{-1} \Pc_l (l^{-1} \partial_\tau l) \rangle\!\rangle_\da $ is gauge invariant, using the fact that $\tp \bigl( \Ec^{-1} \Pc_l \bigr) = \Ec^{-1} \Pc_l$ by part $(i)$ in~Proposition~\ref{PropE}.

Similarly, one finds that
\begin{gather*}
\Pb_l \bigl( (kl)^{-1}\partial_\sigma (kl) \bigr) = \Pb_l \bigl( l^{-1}\partial_\sigma l + \Ad_l^{-1} k^{-1} \partial_\sigma k \bigr) = \Pb_l (l^{-1} \partial_\sigma l ),
\end{gather*}
using $\ker\Pb_l=\Ad_l^{-1}\k$ from~\eqref{Eq:KerImbarPb}. The second term in the action~\eqref{Eq:ActionE} is thus also gauge invariant (using the symmetry of $\Ec \overline{\Pc}_l$).

Rewriting the third term in~\eqref{Eq:ActionE} as $ \langle\!\langle l^{-1} \partial_\tau l, \overline{\Pc}_l (l^{-1} \partial_\sigma l) \rangle\!\rangle_\da - \langle\!\langle \Pc_l (l^{-1} \partial_\tau l), l^{-1} \partial_\sigma l \rangle\!\rangle_\da$, using the gauge invariance of $\Pc_l( l^{-1} \partial_\tau l)$ and $\Pb_l ( l^{-1} \partial_\sigma l)$ derived above and the fact that
$(kl)^{-1} \partial_\mu (kl) = l^{-1} \partial_\mu l + \Ad_l^{-1} k^{-1}\partial_\mu k$,
one deduces that
\begin{gather*}
S_{\Ec,\k}(kl) = S_{\Ec,\k}(l) - \frac{1}{2} I^{\rm WZ}_{\da}[kl] + \frac{1}{2} I^{\rm WZ}_{\da}[l]
\\ \hphantom{S_{\Ec,\k}(kl) =}
{}+ \frac{1}{2} \int_\Sigma \left( \langle\!\langle \Ad_l^{-1} k^{-1}\partial_\tau k, \overline{\Pc}_l \big(l^{-1} \partial_\sigma l\big) \rangle\!\rangle_\da - \langle\!\langle \Pc_l \big(l^{-1} \partial_\tau l\big), \Ad_l^{-1} k^{-1}\partial_\sigma k \rangle\!\rangle_\da \right) {\rm d}\sigma \wedge {\rm d}\tau.
\end{gather*}
Using the Polyakov--Wiegmann identity~\eqref{PW} (with $b$ replaced by $k$ and noting that $I^{\rm WZ}_{\da}[k]$ vanishes as $\k$ is isotropic) and the facts that $\tp\Pb_l=\id-\Pc_l$ and $\tp\Pc_l=\id-\Pb_l$, one rewrites the above equation as
\begin{gather*}
S_{\Ec,\k}(kl) = S_{\Ec,\k}(l) - \frac{1}{2} \int_\Sigma \left( \langle\!\langle \Pc_l \Ad_l^{-1} k^{-1}\partial_\tau k, l^{-1} \partial_\sigma l \rangle\!\rangle_\da \right) {\rm d}\sigma \wedge {\rm d}\tau
\\ \hphantom{S_{\Ec,\k}(kl) =}
{}+ \frac{1}{2} \int_\Sigma \left( \langle\!\langle l^{-1} \partial_\tau l, \Pb_l \Ad_l^{-1} k^{-1}\partial_\sigma k \rangle\!\rangle_\da \right) {\rm d}\sigma \wedge {\rm d}\tau.
\end{gather*}
Finally, since $\Ad_l^{-1} k^{-1}\partial_\mu k$ belongs to $\Ad_l^{-1} \k = \ker\Pc_l = \ker\Pb_l$, we simply obtain
$S_{\Ec,\k}(kl) = S_{\Ec,\k}(l)$,
as expected.

\subsection{Equations of motion} \label{sec: eom E-model}

In this subsection, we derive the equations of motion of the field $l$ coming from the $\Ec$-model action~\eqref{Eq:ActionE}. It~will be useful to introduce a $\da$-valued 1-form $\J \coloneqq \J_\sigma{\rm d}\sigma + \J_\tau{\rm d}\tau$ with com-\linebreak ponents
\begin{subequations}\label{Eq:J}
\begin{gather}
\label{Eq:J a}
\J_\sigma \coloneqq \overline{\Pc}_l\big(l^{-1} \partial_\sigma l\big) + \Ec^{-1} \Pc_l\big(l^{-1} \partial_\tau l\big), \\
\J_\tau \coloneqq \Ec \overline{\Pc}_l \big(l^{-1} \partial_\sigma l\big) + \Pc_l \big(l^{-1} \partial_\tau l\big),
\end{gather}
\end{subequations}
such that $\J_\tau = \Ec\J_\sigma$. It~follows from the computations in Section~\ref{sec: gauge inv E-model} that these expressions are invariant under the local symmetry $l \mapsto k l$ for arbitrary $k \in C^\infty(\Sigma, K)$. One also checks directly from the definitions of the projectors $\Pc_l$ and $\Pb_l$ in Section~\ref{sec: operators E Pl} that the gauge transformation of~$\J$ by $l$ is valued in the subalgebra $\k$, namely
\begin{gather}\label{Eq:JGauged}
\null^l \J \coloneqq - {\rm d}l l^{-1} + \Ad_l \J \in \k.
\end{gather}
In terms of $\J$, the action~\eqref{Eq:ActionE} can be rewritten in the simple form
\begin{gather}\label{Eq:ActionJ}
S_{\Ec,\k}(l) = - \frac{1}{2} \int_\Sigma \langle\!\langle l^{-1} {\rm d} l, \J \rangle\!\rangle_\da - \frac{1}{2} I^{\rm WZ}_{\da}[l].
\end{gather}
Here, and in the rest of the paper, we extend bilinear pairings such as $\langle\!\langle \cdot, \cdot \rangle\!\rangle_\da$ to Lie algebra valued forms using the exterior product. We~will then derive the equations of motion of this action by varying the field $l$ by an infinitesimal right multiplication. We~will need the following lemma, which describes the transformation of $\J$ under this transformation.

\begin{Lemma}\label{Lem:VarJ}
Under an infinitesimal multiplication $\delta l = l\epsilon$ of $l$, where $\epsilon \in C^\infty(\Sigma, \da)$, the variation of $\J_\sigma$ is given by
\begin{gather*}
\delta\J_\sigma = \Pb_l\bigl( \partial_\sigma \epsilon + [\J_\sigma,\epsilon] \bigr) + \Ec^{-1}\Pc_l \bigl(\partial_\tau \epsilon + [\J_\tau,\epsilon] \bigr).
\end{gather*}
Moreover, the variation of $\J_\tau$ is given by $\delta\J_\tau=\Ec\delta\J_\sigma$.
\end{Lemma}

\begin{proof}
Let us consider the gauge transformation~\eqref{Eq:JGauged} of $\J$ by $l$. Its variation under an infi\-nitesimal multiplication $\delta l = l\epsilon$ is given in terms of the variation of $\J$ by
\begin{gather*}
\delta (\null^l\J) = \Ad_l \big( \delta\J + [\epsilon, \J] - d\epsilon \big).
\end{gather*}
Rewriting this equation in components and acting with $\Ad_l^{-1}$, we get
\begin{gather*}
\delta\J_\sigma = \Ad_l^{-1} \delta \big(\null^l \J_\sigma\big) + \partial_\sigma \epsilon + [\J_\sigma,\epsilon] \qquad \text{and} \qquad \Ec\, \delta \J_\sigma = \Ad_l^{-1} \delta \big(\null^l \J_\tau\big) + \partial_\tau \epsilon + [\J_\tau,\epsilon],
\end{gather*}
where we have used the fact that $\J_\tau=\Ec\J_\sigma$. Recall from~\eqref{Eq:JGauged} that $\null^l \J_\sigma$ and $\null^l \J_\tau$, and thus also their variations $\delta (\null^l \J_\sigma)$ and $\delta (\null^l \J_\tau)$, are valued in $\k$. The first terms in the right-hand sides of the above equations are thus valued in $\Ad_l^{-1} \k = \ker\Pc_l = \ker\Pb_l$. Applying $\Pb_l$ to the first equation and $\Pc_l$ to the second one, we thus get:
\begin{gather*}
\Pb_l(\delta\J_\sigma) = \Pb_l\bigl( \partial_\sigma \epsilon + [\J_\sigma,\epsilon] \bigr) \qquad \text{and} \qquad \Pc_l\Ec(\delta\J_\sigma) = \Pc_l \bigl(\partial_\tau \epsilon + [\J_\tau,\epsilon] \bigr).
\end{gather*}
Using the fact that $\Pc_l \Ec = \Ec - \Ec\Pb_l$ (see part $(ii)$ of Proposition~\ref{PropE}) and taking the sum of the first equation above and the action of $\Ec^{-1}$ on the second, we thus get
\begin{gather*}
\delta\J_\sigma = \Pb_l\bigl( \partial_\sigma \epsilon + [\J_\sigma,\epsilon] \bigr) + \Ec^{-1}\Pc_l \bigl(\partial_\tau \epsilon + [\J_\tau,\epsilon] \bigr).
\end{gather*}
This ends the proof of the lemma (noting that the variation of $\J_\tau=\Ec\J_\sigma$ directly follows from the variation of $\J_\sigma$).
\end{proof}

Using Lemma~\ref{Lem:VarJ}, one can compute the variation of the action~\eqref{Eq:ActionJ} and derive the equations of motion of the model.

\begin{Proposition} \label{prop: E-model eom}
The equations of motion of the action~\eqref{Eq:ActionJ} take the form of the zero curvature equation:
\begin{gather*}
{\rm d}\J + \frac{1}{2} [\J,\J] = 0.
\end{gather*}
\end{Proposition}

\begin{proof}
Combining Lemma~\ref{Lem:VarJ} with the facts that $\tp\Pb_l=\id-\Pc_l$ and $\tp(\Ec^{-1}\Pc_l) = \Ec^{-1}\Pc_l$ (see Section~\ref{sec: operators E Pl}), one has
\begin{gather*}
\langle\!\langle l^{-1} \partial_\tau l, \delta \J_\sigma \rangle\!\rangle_\da = \langle\!\langle (\id-\Pc_l) l^{-1} \partial_\tau l, \partial_\sigma \epsilon + [\J_\sigma,\epsilon] \rangle\!\rangle_\da + \langle\!\langle \Ec^{-1}\Pc_l\, l^{-1}\partial_\tau l, \partial_\tau \epsilon + [\J_\tau,\epsilon] \rangle\!\rangle_\da.
\end{gather*}
Similarly, using $\tp\Pc_l=\id-\Pb_l$ and $\tp(\Ec\Pb_l) = \Ec\Pb_l$ we get
\begin{gather*}
\langle\!\langle l^{-1} \partial_\sigma l, \delta \J_\tau \rangle\!\rangle_\da = \langle\!\langle \Ec\Pb_l\, l^{-1} \partial_\sigma l, \partial_\sigma \epsilon + [\J_\sigma,\epsilon] \rangle\!\rangle_\da + \langle\!\langle (\id-\Pb_l)l^{-1}\partial_\sigma l, \partial_\tau \epsilon + [\J_\tau,\epsilon] \rangle\!\rangle_\da.
\end{gather*}
Taking the difference of the above two equations, we obtain
\begin{gather*}
\langle\!\langle l^{-1} \partial_\sigma l, \delta \J_\tau \rangle\!\rangle_\da - \langle\!\langle l^{-1} \partial_\tau l, \delta \J_\sigma \rangle\!\rangle_\da =\langle\!\langle (\id-\Pb_l)l^{-1}\partial_\sigma l - \Ec^{-1}\Pc_l\, l^{-1}\partial_\tau l, \partial_\tau \epsilon + [\J_\tau,\epsilon] \rangle\!\rangle_\da
\\ \hphantom{\langle\!\langle l^{-1} \partial_\sigma l, \delta \J_\tau \rangle\!\rangle_\da - \langle\!\langle l^{-1} \partial_\tau l, \delta \J_\sigma \rangle\!\rangle_\da =}
{} - \langle\!\langle (\id-\Pc_l) l^{-1} \partial_\tau l - \Ec\Pb_l\, l^{-1} \partial_\sigma l, \partial_\sigma \epsilon + [\J_\sigma,\epsilon] \rangle\!\rangle_\da
\\ \hphantom{\langle\!\langle l^{-1} \partial_\sigma l, \delta \J_\tau \rangle\!\rangle_\da - \langle\!\langle l^{-1} \partial_\tau l, \delta \J_\sigma \rangle\!\rangle_\da}
{}= \langle\!\langle l^{-1}\partial_\sigma l - \J_\sigma, \partial_\tau \epsilon + [\J_\tau,\epsilon]  \rangle\!\rangle_\da
\\ \hphantom{\langle\!\langle l^{-1} \partial_\sigma l, \delta \J_\tau \rangle\!\rangle_\da - \langle\!\langle l^{-1} \partial_\tau l, \delta \J_\sigma \rangle\!\rangle_\da =}
{} - \langle\!\langle l^{-1} \partial_\tau l - \J_\tau, \partial_\sigma \epsilon + [\J_\sigma,\epsilon] \rangle\!\rangle_\da ,
\end{gather*}
where in the last equality we have used the definition~\eqref{Eq:J} of $\J_\sigma$ and $\J_\tau$. In~terms of forms, we can rewrite the above equation as
\begin{gather*}
\langle\!\langle l^{-1} {\rm d} l, \delta \J \rangle\!\rangle_\da = \langle\!\langle {\rm d} \epsilon + [\J,\epsilon], \J - l^{-1} {\rm d} l \rangle\!\rangle_\da.
\end{gather*}
Under the infinitesimal multiplication $\delta l = l\epsilon$, the variation of the 1-form $l^{-1}{\rm d}l$ is given by
\begin{gather*}
\delta\big(l^{-1}{\rm d}l\big) = {\rm d}\epsilon + \big[l^{-1}{\rm d}l,\epsilon\big],
\end{gather*}
while the variation of the WZ-term $I^{\rm WZ}_{\da}[l]$ follows from the Polyakov--Wiegmann identity~\cite{Polyakov-Wiegmann} and reads
\begin{gather*}
\delta I^{\rm WZ}_{\da}[l] = \int_\Sigma \langle\!\langle {\rm d}\epsilon, l^{-1} {\rm d}l \rangle\!\rangle_\da.
\end{gather*}
Combining all the above, we then determine the variation of the action~\eqref{Eq:ActionJ} to be
\begin{align*}
\delta S_{\Ec,\k}(l) &= - \frac{1}{2} \int_\Sigma \big( \langle\!\langle \delta(l^{-1} {\rm d} l), \J  \rangle\!\rangle_\da + \langle\!\langle l^{-1} {\rm d} l, \delta\J  \rangle\!\rangle_\da \big) - \frac{1}{2} \delta I^{\rm WZ}_{\da}[l] \\
&= - \int_\Sigma \bigg( \langle\!\langle {\rm d}\epsilon, \J \rangle\!\rangle_\da + \frac{1}{2} \langle\!\langle [\J,\epsilon], \J \rangle\!\rangle_\da + \frac{1}{2} \langle\!\langle [l^{-1}{\rm d}l,\epsilon], \J \rangle\!\rangle_\da + \frac{1}{2} \langle\!\langle l^{-1}{\rm d}l, [\J,\epsilon] \rangle\!\rangle_\da \bigg).
\end{align*}
Using the ad-invariance of $\langle\!\langle \cdot, \cdot \rangle\!\rangle_\da$ and integration by part, we finally get
\begin{gather*}
\delta S_{\Ec,\k}(l) = \int_\Sigma \bigg\langle\!\!\!\bigg\langle \epsilon, d\J + \frac{1}{2} [\J, \J] \bigg\rangle\!\!\!\bigg\rangle_\da.
\end{gather*}
The result now follows by requiring that $\delta S_{\Ec,\k}(l)=0$ for every $\epsilon$.
\end{proof}

\begin{Remark}
In the relativistic case when $\Ec^2 = \id$, see Section~\ref{sec: EM tensor E-model} below and in particular Remark~\ref{rem: rel inv E-model}, one can rewrite~\eqref{Eq:J a} in terms of the projector~\eqref{Klimcik proj} as
\begin{gather*}
\J_\sigma = - \partial_\sigma f f^{-1} + \frac 12 \Ad_f \big( P_f(\Ec) \big(f^{-1} \partial_+ f\big) - P_f(-\Ec) \big(f^{-1} \partial_- f\big) \big).
\end{gather*}
This coincides with \cite[equation~(17)]{Klimcik:2015gba} up to an overall sign, which is due to a difference in conventions.
Indeed, the equations of motion from Proposition~\ref{prop: E-model eom} can be written in components~as
\begin{gather*}
\partial_\tau \J_\sigma - \partial_\sigma (\Ec \J_\sigma) + [\Ec \J_\sigma,\J_\sigma] = 0
\end{gather*}
which are to be compared with the equations of motion in \cite[equation~(9)]{Klimcik:2015gba}.
\end{Remark}

\subsection{Energy-momentum tensor} \label{sec: EM tensor E-model}

The following proposition will be useful in the discussion of Section~\ref{sec: EM tensor}. We~give its proof in Appendix~\ref{App:EMT}.

\begin{Proposition} \label{prop: EM tensor}
 The components of the energy-momentum tensor of the $\Ec$-model~\eqref{Eq:ActionE} are given by
\begin{gather*}
T^\tau_{\;\;\,\tau} = -T^\sigma_{\;\;\,\sigma} = \frac{1}{2} \langle\!\langle \J_\sigma , \Ec \J_\sigma \rangle\!\rangle_\da, \qquad T^\tau_{\;\;\,\sigma} = \frac{1}{2} \langle\!\langle \J_\sigma , \J_\sigma \rangle\!\rangle_\da, \qquad T^\sigma_{\;\;\,\tau} = - \frac{1}{2} \langle\!\langle \J_\sigma , \Ec^2 \J_\sigma \rangle\!\rangle_\da,
\end{gather*}
where $\J_\sigma$ is the $\da$-valued field defined in~\eqref{Eq:J}.
\end{Proposition}
\begin{Remark} \label{rem: rel inv E-model}
The relativistic invariance of the $\Ec$-model can be deduced immediately from Proposition~\ref{prop: EM tensor}.
Defining the $2$d Minkowski metric $\eta_{\mu\nu}$ by $\eta_{\tau\tau}=-\eta_{\sigma\sigma}=1$ and $\eta_{\tau\sigma}=-\eta_{\sigma\tau}=0$, we can lower the indices of the energy-momentum tensor and define $T_{\mu\nu} \coloneqq \eta_{\mu\rho}T^\rho_{\;\;\,\nu}$. In~particular, we get
\begin{gather*}
T_{\tau\tau} = T_{\sigma\sigma} = \frac{1}{2} \langle\!\langle \J_\sigma , \Ec \J_\sigma \rangle\!\rangle_\da, \qquad T_{\tau\sigma} = \frac{1}{2} \langle\!\langle \J_\sigma , \J_\sigma \rangle\!\rangle_\da, \qquad T_{\sigma\tau} = \frac{1}{2} \langle\!\langle \J_\sigma , \Ec^2 \J_\sigma \rangle\!\rangle_\da.
\end{gather*}
It follows that if $\Ec^2 = \id$ then the energy-momentum tensor $T_{\mu \nu}$ is symmetric, which implies the relativistic invariance of the model.

Moreover, if $\Ec$ is positive with respect to the bilinear form $\langle\!\langle \cdot, \cdot \rangle\!\rangle_\da$ then the Hamiltonian $\int_\RR {\rm d}\sigma \; T^\tau_{\;\; \tau}$ is positive.
\end{Remark}

\section{2d integrable field theories from 4d Chern--Simons theory} \label{sec: 2d from 4d}

The general $2$d action obtained in~\cite{Benini:2020skc} will serve as the starting point of our analysis in Section~\ref{sec: E-model from 4d} below, so in this section we begin by reviewing the results of~\cite{Benini:2020skc}, referring the reader to the latter for details and proofs.

In order to ensure that the $2$d action is real, we will also impose reality conditions following~\cite{Delduc:2019whp}, see also~\cite{Delduc:2019bcl, Vicedo:2017cge} in the context of affine Gaudin models. Although this was not directly considered in~\cite{Benini:2020skc}, the analysis there readily applies to the real setting.

\subsection{Surface defects} \label{sec: defects}

Let $\omega$ be a meromorphic $1$-form on $\CP$. We~denote its set of poles by $\bZ \subset \CP$ and denote by $n_x \in \ZZ_{\geq 1}$ the order of a pole $x \in \bZ$. Although this was not necessary in the analysis of~\cite{Benini:2020skc}, we~shall assume here that $\omega$ has a double pole at infinity, namely $\infty \in \bZ$ with $n_\infty = 2$. Let us fix a coordinate $z$ on $\CC \subset \CP$ so that $\omega$ can be written explicitly as
\begin{gather}\label{Eq:Omega}
\omega = \bigg( \sum_{x \in \bZ'} \sum_{p=0}^{n_x-1} \frac{\ell^x_p}{(z-x)^{p+1}} - \ell^\infty_1 \bigg) {\rm d}z \eqqcolon \varphi(z) {\rm d}z,
\end{gather}
where $\bZ' \coloneqq \bZ \setminus \{ \infty \}$, for some $\ell^x_p \in \CC$ which we refer to as the \emph{levels}. We~also define $\ell_0^\infty \coloneqq \res_\infty \omega = - \sum_{x \in \bZ'} \ell^x_0$.

We impose reality conditions on each pole $x \in \bZ$ and its corresponding set of levels $\ell^x_p$, $p = 0, \dots, n_x - 1$ by requiring that $\overline{\varphi(z)} = \varphi(\bar z)$. In~particular, we define the subset of real poles $\bz_{\rm r} \coloneqq \bz'_{\rm r} \sqcup \{ \infty \}$, where $\bz'_{\rm r} \coloneqq \bZ' \cap \RR$. By the above assumption on $\varphi$ the associated levels are real, i.e., $\ell^x_p \in \RR$. The remaining poles come in complex conjugate pairs and we define $\bz_{\rm c} \coloneqq \{ x \in \bZ \mid \Im x > 0 \}$ so that $\bZ = \bz_{\rm r} \sqcup \bz_{\rm c} \sqcup \bar{\bz}_{\rm c}$. For~every $x \in \bz_{\rm c} \sqcup \bar{\bz}_{\rm c}$ we have $n_{\bar x} = n_x$ and $\overline{\ell^x_p} = \ell^{\bar x}_p$ for $p = 0, \dots, n_x - 1$. It~is convenient to introduce the set $\bz \coloneqq \bz_{\rm r} \sqcup \bz_{\rm c}$ of independent poles. We~also introduce the subset $\bz' \coloneqq \bz'_{\rm r} \sqcup \bz_{\rm c} \subset \bz$ of finite independent poles in $\bz$.

The set of zeroes of $\omega$ can be similarly decomposed as $\bzeta_{\rm r} \sqcup \bzeta_{\rm c} \sqcup \bar{\bzeta}_{\rm c}$ with $\bzeta_{\rm r} \subset \RR$ the subset of real zeroes and $\bzeta_{\rm c} \subset \{ z \in \CC \mid \Im z > 0\}$ the subset of complex zeroes. We~introduce the set $\bzeta \coloneqq \bzeta_{\rm r} \sqcup \bzeta_{\rm c}$ of independent zeroes and let $m_y \in \ZZ_{\geq 1}$ denote the order of the zero $y \in \bm \zeta$. For~$y \in \bzeta_{\rm c}$, $\omega$ also has a zero of order $m_{\bar y} \coloneqq m_y$ at $\bar y \in \bar{\bzeta}_{\rm c}$.

It will be convenient to introduce the group $\Pi = \{ \id, \ms t \} \cong \ZZ_2$ which acts on $\CP$ by letting~$\ms t$ act by complex conjugation $\mu_{\ms t} \colon z \mapsto \bar z$. Note that we can then write $\bZ = \Pi \bz$ and we have $\Pi \bzeta = \bzeta_{\rm r} \sqcup \bzeta_{\rm c} \sqcup \bar \bzeta_{\rm c}$. Let $\Pi_x \subset \Pi$ denote the stabiliser subgroup of a point $x \in \CC$, so that $\Pi_x = \{ \id \}$ is the trivial group for $x \in \bz_{\rm c} \sqcup \bar{\bz}_{\rm c}$ and $\Pi_x = \Pi$ for $x \in \bz_{\rm r}$. In~particular, $|\Pi_x| = 2$ if~$x \in \bz_{\rm r}$ and $|\Pi_x| = 1$ if~$x \in \bz_{\rm c}$. The analogous statements hold for the stabilisers $\Pi_y$ of zeroes $y \in \bzeta_{\rm r} \sqcup \bzeta_{\rm c} \sqcup \bar \bzeta_{\rm c}$ of $\omega$.

Let $C \coloneqq \CP \setminus \bm \zeta$, $\Sigma \coloneqq \RR^2$ and $X \coloneqq \Sigma \times C$. We~will always think of $\omega$ as defining a $1$-form on $X$ with singularities along the disjoint union
\begin{gather*}
\D \coloneqq \bigsqcup_{x \in \bm z} \Sigma_x
\end{gather*}
of \emph{surface defects} $\Sigma_x \coloneqq \Sigma \times \{ x \} \subset X$, each trivially homeomorphic to $\Sigma$, i.e., $\Sigma_x \cong \Sigma$. We~denote the embedding of the individual surface defects by
\begin{gather*}
\iota_x \colon\ \Sigma_x \longhookrightarrow X.
\end{gather*}

To account for the fact that poles and zeroes of $\omega$ may not be simple, it will be convenient to let $\bsb{\bz}$ denote the set of pairs $\up{x, p}$ with $x \in \bm z$ and $p = 0, \dots, n_x - 1$, and likewise, let $\brb{\bzeta}$ be the set of pairs $\uz{y, q}$ with $y \in \bzeta$ and $q = 0, \dots, m_y - 1$. We~think of the collection $\up{x, p}$ for~$p = 0, \dots, n_x - 1$ (resp.~$\uz{y, q}$ for $q = 0, \dots, m_y - 1$) as an infinitesimal ``thickening'' of the pole~$x$ (resp.~the zero $y$).

\subsection{Defect Lie algebra} \label{sec: defect Lie alg}

Let $G$ be a real simply connected Lie group. We~suppose that its Lie algebra $\g$ is equipped with a non-degenerate invariant symmetric bilinear form $\langle \cdot, \cdot \rangle \colon \g \times \g \to \RR$.

Let $\g^\CC \coloneqq \g \otimes_\RR \CC$ denote the complexification of $\g$, which comes equipped with an anti-linear involution $\tau \colon \g^\CC \to \g^\CC$ given by complex conjugation in the second tensor factor. We~extend the bilinear form on $\g$ to a bilinear form $\langle \cdot, \cdot \rangle \colon \g^\CC \times \g^\CC \to \CC$ by complex linearity, so that $\langle \tau \ms u, \tau \ms v \rangle = \overline{\langle \ms u, \ms v \rangle}$ for any $\ms u, \ms v \in \g^\CC$.
Let $\tau \colon G^\CC \to G^\CC$ be the lift of $\tau \colon \g^\CC \to \g^\CC$ to an involutive automorphism of $G^\CC$. The real Lie group $G$ can then be identified as the subgroup of fixed points of $\tau$.

Let $\T^{n_x}_x \coloneqq \RR[\varepsilon_x]/ (\varepsilon_x^{n_x})$ for real poles $x \in \bz_{\rm r}$ and $\T^{n_x}_x \coloneqq \CC[\varepsilon_x]/ (\varepsilon_x^{n_x})$ for complex poles $x \in \bz_{\rm c}$. We~define the \emph{defect Lie algebra} as the real Lie algebra
\begin{gather} \label{defect alg}
\gz \coloneqq \bigoplus_{x \in \bz_{\rm r}} \g \otimes_\RR \T^{n_x}_x \oplus \bigoplus_{x \in \bz_{\rm c}} \big(\g^\CC \otimes_\CC \T^{n_x}_x\big)_\RR,
\end{gather}
where $(\g^\CC \otimes_\CC \T^{n_x}_x)_\RR$ is the realification of the complex Lie algebra $\g^\CC \otimes_\CC \T^{n_x}_x$, i.e., $\g^\CC \otimes_\CC \T^{n_x}_x$ regarded as a Lie algebra over $\RR$. We~use the notation $\ms u^{\up{x, p}} \coloneqq \ms u \otimes \varepsilon_x^p \in \gz$ for any $\ms u \in \g$ and~$\up{x, p} \in \bsb{\bz_{\rm r}}$ or $\ms u \in \g^\CC$ and $\up{x, p} \in \bsb{\bz_{\rm c}}$. The Lie algebra relations of $\gz$ are given explicitly in~terms of this basis as
\begin{gather*}
\big[ \ms u^{\up{x, p}}, \ms v^{\up{y, q}} \big] = \delta_{xy} [\ms u, \ms v]^{\up{x, p+q}}.
\end{gather*}
Note that this is zero if $p + q \geq n_x$. We~equip $\gz$ with a non-degenerate invariant symmetric bilinear form defined by
\begin{gather} \label{form on gz}
\langle\!\langle \cdot, \cdot \rangle\!\rangle_{\gz} \colon\ \gz \times \gz \longrightarrow \RR, \qquad
\big\langle{\mkern-4mu}\big\langle \ms u^{\up{x, p}}, \ms v^{\up{y, q}} \big\rangle{\mkern-4mu}\big\rangle_{\gz} = \delta_{xy} \, \frac{2}{|\Pi_x|} \Re \big( \ell^x_{p+q} \langle \ms u, \ms v \rangle \big).
\end{gather}
Here we define $\ell^x_p = 0$ for all $p \geq n_x$. Note that for $x \in \bz_{\rm r}$ we have
\begin{gather*}
\big\langle{\mkern-4mu}\big\langle \ms u^{\up{x, p}}, \ms v^{\up{y, q}} \big\rangle{\mkern-4mu}\big\rangle_{\gz} = \delta_{xy} \, \ell^x_{p+q} \langle \ms u, \ms v \rangle,
\end{gather*}
while for $x \in \bz_{\rm c}$ we have
\begin{gather*}
\big\langle{\mkern-4mu}\big\langle \ms u^{\up{x, p}}, \ms v^{\up{y, q}} \big\rangle{\mkern-4mu}\big\rangle_{\gz} = \delta_{xy} \, \big( \ell^x_{p+q} \langle \ms u, \ms v \rangle + \ell^{\bar x}_{p+q} \langle \tau \ms u, \tau \ms v \rangle \big).
\end{gather*}
One can also introduce a real Lie group with Lie algebra $\gz$ which we will call the \emph{defect group} and denote by $\Gz$.

From now on we will assume that the real Lie algebra $\gz$ is even dimensional. In~other words, either $\g$ itself is even dimensional or the number of poles of $\omega$ counting multiplicities, namely $\sum_{x \in \bm z} n_x$, is even. Note that since we are assuming $n_\infty = 2$ it follows that $\sum_{x \in \bm z'} n_x$ is also even.

\subsection{The map \texorpdfstring{$\bm j^\ast$}{jast}}

Let $\Omega^1\big(X, \g^\CC\big)$ denote the complex vector space of smooth $\g^\CC$-valued $1$-forms on $X$. We~can define two actions of the group $\Pi$ on $\Omega^1\big(X, \g^\CC\big)$: we can let $\ms t \in \Pi$ act as the pullback by complex conjugation $\mu_{\ms t} \colon z \mapsto \bar z$ or we can let it act as $\tau$ on $\g^\CC$. We~let $\Omega^1\big(X, \g^\CC\big)^\Pi$ denote the real vector space consisting of $1$-forms $\eta \in \Omega^1\big(X, \g^\CC\big)$ on which these two actions agree, namely such that $\mu_{\ms t}^\ast \eta = \tau \eta$.

The relationship between the defect Lie algebra $\gz$ and the surface defect $\D$ can be understood through the following linear map of real vector spaces
\begin{gather} \label{j pull}
\jb^\ast \colon\ \Omega^1\big(X, \g^\CC\big)^\Pi \longrightarrow \Omega^1\big(\Sigma, \gz\big), \qquad
\eta \longmapsto \bigg(\sum_{p=0}^{n_x - 1} \frac{1}{p!} \iota_x^\ast (\partial^p_z \eta) \otimes \varepsilon_x^p \bigg)_{x \in \bm z}.
\end{gather}
In words, this map takes a smooth equivariant $1$-form $\eta \in \Omega^1\big(X, \g^\CC\big)^\Pi$ and returns the first~$n_x$ terms in the holomorphic part of its Taylor expansion at points on the surface defect $\D$, keeping only the two components of the $1$-form along $\Sigma$. Note that for $x \in \bz_{\rm r}$ the corresponding component of~\eqref{j pull} is indeed in $\g \otimes_\RR \T^{n_x}_x$ since
\begin{gather*}
\tau \big( \iota_x^\ast (\partial^p_z \eta) \big) = \iota_x^\ast (\partial^p_{\bar z} (\tau \eta)) = \iota_x^\ast (\partial^p_{\bar z} (\mu_{\ms t}^\ast \eta)) = \iota_x^\ast \mu_{\ms t}^\ast (\partial^p_z \eta) = \iota_x^\ast (\partial^p_z \eta),
\end{gather*}
where in the first equality we used the anti-linearity of $\tau$, in the second step the equivariance of~$\eta$ and in the final step the fact that $\mu_{\ms t} \circ \iota_x = \iota_x$ since $x \in \bz_{\rm r}$.

In the simplest case when $n_x = 1$ for all $x \in \bm z$, the map $\jb^\ast$ is simply the pullback by the embedding $\jb \colon \D \hookrightarrow X$ since we have the canonical identification
\begin{gather} \label{iso surf defect}
\Omega^1\big(\D, \g^\CC\big)^\Pi \cong \bigg( \bigoplus_{x \in \bm z} \Omega^1\big(\Sigma, \g^\CC\big) \bigg)^\Pi \cong \Omega^1\bigg( \Sigma, \bigg( \bigoplus_{x \in \bm z} \g^\CC \bigg)^\Pi \bigg) = \Omega^1\big( \Sigma, \gz \big).
\end{gather}
In this case, the map~\eqref{j pull} is given simply by $\jb^\ast \eta = ( \iota_x^\ast \eta )_{x \in \bm z}$, namely the collection of pullbacks of $\eta$ to each surface defect $\Sigma_x$.

\subsection{4d Chern--Simons theory with edge modes} \label{sec: edge modes}

The action of $4$-dimensional Chern--Simons theory~\cite{Costello:2019tri} for an equivariant $\g^\CC$-valued $1$-form $A \in \Omega^1\big(X, \g^\CC\big)^\Pi$ is given by integrating $\frac{\ii}{4 \pi} \omega \wedge \CS(A) \in \Omega^4(X)$ over $X$, where $\CS(A) \coloneqq \big\langle A, {\rm d}A + \frac 13 [A,A] \big\rangle$ is the Chern--Simons 3-form. By \cite[Lemma 2.4]{Delduc:2019whp} this action is real.

Strictly speaking, the $4$-form $\omega \wedge \CS(A)$ is not integrable in the neighbourhood of a surface defect $\Sigma_x$ corresponding to a higher order pole $x \in \bm z$ with $n_x > 1$. For~this reason, one needs to~introduce a suitable regularisation of the action~\cite{Benini:2020skc}, which we denote by $S_{\rm 4d}(A)$. The proof of~\cite[Lemma 2.4]{Delduc:2019whp} generalises to this regularised action, showing that it is also real. The behaviour of $S_{\rm 4d}(A)$ under gauge transformations
\begin{gather*}
A \longmapsto \null^g A \coloneqq - {\rm d}g g^{-1} + g A g^{-1}
\end{gather*}
by $g \in C^\infty(X, G^\CC)^\Pi$ was studied in~\cite{Benini:2020skc}, where it was shown that gauge invariance can be achieved in two separate but equivalent ways: either by imposing boundary conditions on the field $A \in \Omega^1\big(X, \g^\CC\big)^\Pi$, or by coupling $A$ to a new field localised in the formal neighbourhood of the surface defect $\D$ which amounts to a field $h \in C^\infty\big(\Sigma, \Gz\big)$.

To describe a general class of boundary conditions on $A$, note that by applying the map~\eqref{j pull} to $A \in \Omega^1\big(X, \g^\CC\big)^\Pi$ we obtain a $1$-form $\jb^\ast A \in \Omega^1\big(\Sigma, \gz\big)$.
Let $\fa \subset \gz$ be a Lagrangian subalgebra of the defect Lie algebra $\gz$ and let $F$ be the corresponding connected real Lie subgroup of~$\Gz$. It~was shown in~\cite{Benini:2020skc} that the regularised action $S_{\rm 4d}(A)$ becomes gauge invariant if we restrict attention to fields $A \in \Omega^1\big(X, \g^\CC\big)^\Pi$ for which
$\jb^\ast A \in \Omega^1(\Sigma, \fa)$.
Correspondingly, gauge transformation parameters $g \in C^\infty\big(X, G^\CC\big)^\Pi$ are restricted to be such that their ``pullback'' to the formal neighbourhood of $\D$ is $F$-valued.

An alternative way of ensuring gauge-invariance of $4$d Chern--Simons theory, which provides a direct route to the action of the $2$d integrable field theory~\cite{Benini:2020skc}, requires introducing a new field $h \in C^\infty\big(\Sigma, \Gz\big)$ called the \emph{edge mode}. In~the simplest case when $\omega$ has only simple poles, i.e., $n_x = 1$ for all $x \in \bm z$, we have a canonical isomorphism $C^\infty\big(\D, G^\CC\big)^\Pi \cong C^\infty\big(\Sigma, \Gz\big)$ by the same line of reasoning as in~\eqref{iso surf defect}, allowing us to view the edge mode in this case as a~$G^\CC$-valued function on $\D$. In~the presence of higher order poles of $\omega$, one can think of the edge mode as a~$G^\CC$-valued field localised in a formal neighbourhood of the defect $\D$. Its role is to witness the boundary condition on $A$. Specifically, rather than imposing boundary conditions on $A$ \emph{strictly}, as in the previous paragraph, we require $A$ to satisfy these boundary conditions only up to a~gauge transformation by the edge mode, namely
\begin{gather} \label{A h constraint}
\null^h (\jb^\ast A) \in \Omega^1(\Sigma, \fa).
\end{gather}

\begin{Remark}
Here we depart slightly from the conventions used in~\cite{Benini:2020skc}, where the condition~\eqref{A h constraint} was written as $\null^{h^{-1}} (\jb^\ast A) \in \Omega^1(\Sigma, \fa)$. Effectively, our edge mode coincides with the inverse of the edge mode in~\cite{Benini:2020skc}.
\end{Remark}

We can now ensure gauge invariance of $4$d Chern--Simons theory by coupling the bulk field $A$ to the edge mode $h$, through its ``pullback'' $\jb^\ast A$. Explicitly, we introduce the extended action~\cite{Benini:2020skc}
\begin{gather} \label{ext action}
S^{\rm ext}_{\rm 4d}(A, h) = S_{\rm 4d}(A) - \frac{1}{2} \int_\Sigma \langle\!\langle h^{-1} {\rm d} h, \jb^\ast A \rangle\!\rangle_{\gz} - \frac 12 I^{\rm WZ}_{\gz}[h],
\end{gather}
where we use the standard WZ-term defined as in~\eqref{WZ term} but with the group $\Gz$ replacing the~role of~$D$.
The action~\eqref{ext action} and the constraint~\eqref{A h constraint} are invariant under the gauge transformation
\begin{subequations} \label{gauge tr A g}
\begin{gather} \label{gauge tr A g 1}
A \longmapsto \null^g A, \qquad
h \longmapsto h (\jb^\ast g)^{-1}
\end{gather}
for any $g \in C^\infty\big(X, G^\CC\big)^\Pi$. There is also a further gauge transformation acting on the edge mode alone as
\begin{gather} \label{gauge tr A g 2}
h \longmapsto f h
\end{gather}
\end{subequations}
for any $f \in C^\infty(\Sigma, F)$. The invariance of the action~\eqref{ext action} under these gauge transformations follows using the Polyakov--Wiegmann identity on the WZ-term.

\subsection{Reduction to 2d integrable field theories} \label{sec: red to 2d}

Having introduced edge modes in the extended action~\eqref{ext action}, the passage to $2$d integrable field theories is now fairly direct. Indeed, the edge mode $h \in C^\infty\big(\Sigma, \Gz\big)$ will ultimately play the role of the collection of fields of the $2$d integrable field theory. The gauge field $A \in \Omega^1(X, \g)$, on the other hand, will become the Lax connection $\L$ of the integrable field theory. For~this to happen, however, we have to restrict attention to $1$-forms $A$ which only have components along $\Sigma \subset X$ and which depend holomorphically on the complex direction $C \subset X$. More precisely, this can be done by focusing on a certain class of solutions to part of the equations of motion for the extended action~\eqref{ext action}, which we now describe.

For a complex vector space $V$ we let $R^\infty_{\Pi\bzeta}(V)$ denote the space of $V$-valued rational functions with poles at each $y \in \Pi \bzeta$ of order at most $m_y$, the order of the zero $y$ of $\omega$. If $V$ is equipped with an anti-linear involution $\tau \colon V \to V$ then we can define an action of $\Pi$ on $V$ by letting $\ms t \in \Pi$ act as $\tau$. This then also lifts to an action of $\Pi$ on~$R^\infty_{\Pi\bzeta}(V)$. We~can also define an action of $\Pi$ on~$R^\infty_{\Pi\bzeta}(V)$ by letting $\ms t \in \Pi$ act as the pullback by complex conjugation $\mu_{\ms t} \colon z \mapsto \bar z$. We~let $R^\infty_{\Pi\bzeta}(V)^\Pi$ denote the real vector space of rational functions in $R^\infty_{\Pi\bzeta}(V)$ on which these two actions coincide. In~what follows we will either take $V = \g^\CC$ or $V = C^\infty\big(\Sigma, \g^\CC\big)$, where the action of $\Pi$ on the latter is induced from the action of $\Pi$ on $\g^\CC$.

Following~\cite{Benini:2020skc}, we will restrict attention to \emph{admissible} $1$-forms $\L \in \Omega^1\big(X, \g^\CC\big)^\Pi$ with the following properties:
\begin{itemize}\itemsep=0pt
 \item[$(a)$] We have $\L = \L_\sigma {\rm d}\sigma + \L_\tau {\rm d}\tau$ with both components $\L_\sigma, \L_\tau \in R^\infty_{\Pi\bzeta}\big(C^\infty\big(\Sigma, \g^\CC\big)\big)^\Pi$. Explicitly, this means that we can write, for $\mu = \sigma, \tau$,
\begin{gather*}
\L_\mu = \L_{\rm c, \mu} + \sum_{\uz{y, q} \in \brb{\Pi \bzeta}} \frac{\L_\mu^{\uz{y, q}}}{(z-y)^{q+1}},
\end{gather*}
for some $\L_{\rm c, \mu} \in C^\infty(\Sigma, \g)$ and $\L_\mu^{\uz{y, q}} \in C^\infty\big(\Sigma, \g^\CC\big)$. In~the case when $y \in \bzeta_{\rm r}$ we have $\L_\mu^{\uz{y, q}} \in C^\infty(\Sigma, \g)$ and for $y \in \bzeta_{\rm c} \sqcup \bar \bzeta_{\rm c}$ we have $\tau \L_\mu^{\uz{y, q}} = \L_\mu^{\uz{\bar y, q}}$ for all $q = 0, \dots, m_y = m_{\bar y}$.
 \item[$(b)$] The single component of the curvature ${\rm d} \L + \frac 12 [\L, \L] = F(\L)_{\sigma\tau} {\rm d}\sigma \wedge {\rm d}\tau$ is also such that $F(\L)_{\sigma\tau} \in R^\infty_{\Pi\bzeta}\big(C^\infty\big(\Sigma, \g^\CC\big)\big)^\Pi$. Given property $(a)$, this is equivalent to the commutator term in $F(\L)_{\sigma\tau}$ having no poles of order greater than $m_y$ at each $y \in \bm \zeta$. Explicitly, we can write this as
\begin{gather*}
\sum_{q = p - m_y + 1}^{m_y - 1} \big[ \L_\sigma^{\uz{y, q}}, \L_\tau^{\uz{y, p - q}} \big] = 0
\end{gather*}
for every $y \in \bm \zeta$ and every $p = m_y - 1, \dots, 2 m_y - 2$.
\end{itemize}

Let us now \emph{suppose} that for every $h \in C^\infty\big(\Sigma, \Gz\big)$ there exists an admissible $1$-form $\L = \L(h) \in \Omega^1\big(X, \g^\CC\big)^\Pi$ such that the condition~\eqref{A h constraint} holds, namely
\begin{gather} \label{L h constraint}
\null^h (\jb^\ast \L) \in \Omega^1(\Sigma, \fa).
\end{gather}
Moreover, we require that the collection of solutions $\L(h)$ for every $h \in C^\infty\big(\Sigma, \Gz\big)$ is equivariant under those gauge transformations of the form~\eqref{gauge tr A g} which preserve the class of admissible $1$-forms.
Specifically, for every $g \in C^\infty(\Sigma, G)$ and $f \in C^\infty(\Sigma, F)$ we should have
\begin{gather} \label{L h gauge inv}
\null^{\Delta(g)^{-1}} \big( \jb^\ast \L\big(fh\Delta(g)^{-1}\big) \big) = \jb^\ast \L(h),
\end{gather}
with $\Delta \colon G \to G^{\times |\bz|} \subset \Gz$ the diagonal map. (At the Lie algebra level, the latter is given explicitly by the diagonal embedding $\g \to \g^{\oplus |\bz|} \subset \gz$, $\ms u \to (\ms u^{\up{x,0}})_{x \in \bz}$.)
Note that~\eqref{L h gauge inv} is compatible with the constraint~\eqref{L h constraint}.

In the terminology of \cite[Remark 5.10]{Benini:2020skc} this amounts to specifying a section of a certain surjective map. This section can then be used to pull back the action of $4$d Chern--Simons theory in the presence of edge modes~\eqref{ext action} to the action of a $2$d integrable field theory with Lax connection $\L(h)$. More explicitly, given an admissible $1$-form $\L(h)$ with the properties described above, if we substitute $A = \L(h)$ in the action~\eqref{ext action} then the first term $S_{\rm 4d}(\L(h))$ vanishes using both admissibility properties $(a)$ and $(b)$ and we are left with the action
\begin{gather} \label{ext action L}
S_{\rm 2d}(h) = - \frac{1}{2} \int_\Sigma \langle\!\langle h^{-1} {\rm d} h, \jb^\ast \L(h) \rangle\!\rangle_{\gz} - \frac 12 I^{\rm WZ}_{\gz}[h]
\end{gather}
for the field $h \in C^\infty\big(\Sigma, \Gz\big)$. In~particular, computing the variation of the action~\eqref{ext action L} with respect to the fields $h$ and $\L(h)$, taking into account the constraint~\eqref{L h constraint} relating these two fields, see \cite[equations~(5.2)--(5.5)]{Benini:2020skc} for details, we find that the equations of motion take the form
\begin{gather*}
{\rm d} \jb^\ast \L(h) + \frac 12 [\jb^\ast \L(h), \jb^\ast \L(h)] = 0.
\end{gather*}
By the admissibility of the $1$-form $\L(h)$ it then follows from \cite[Proposition 5.6]{Benini:2020skc}, see also the related discussion in Section~\ref{sec: inverse jz} below, that the above equation of motion is equivalent to the zero-curvature equation for $\L(h)$ itself, namely
\begin{gather*}
{\rm d} \L(h) + \frac 12 [\L(h), \L(h)] = 0.
\end{gather*}
Furthermore, because of the behaviour~\eqref{L h gauge inv} of $\L(h)$ under gauge transformations, the action~\eqref{ext action L} is invariant under
\begin{gather} \label{gauge tr}
h \longmapsto f h \Delta(g)^{-1}
\end{gather}
for any $g \in C^\infty(\Sigma, G)$ and $f \in C^\infty(\Sigma, F)$.

\subsection{Removing the edge mode at infinity} \label{sec: removing inf}

In order to obtain the $\Ec$-model from the action~\eqref{ext action L} we will need to make one further reduction. Specifically, we shall partially fix the gauge invariance~\eqref{gauge tr} by setting the component of the edge mode $h \in C^\infty\big(\Sigma, \Gz\big)$ at infinity to the identity.

Consider the real Lie subalgebra of the defect Lie algebra $\gz$ defined by
\begin{gather} \label{defect grp alg inf}
\da \coloneqq \bigoplus_{x \in \bm z'_{\rm r}} \g \otimes_\RR \T^{n_x}_x \oplus \bigoplus_{x \in \bm z_{\rm c}} \big(\g^\CC \otimes_\CC \T^{n_x}_x\big)_\RR.
\end{gather}
Notice that in comparing this definition with that of $\gz$ in~\eqref{defect alg} we have simply removed the factor $\g \otimes_\RR \T^2_\infty$ corresponding to the pole at infinifty. We~let $D$ denote the corresponding connected Lie subgroup of the defect group $\Gz$.
Recall from Section~\ref{sec: defect Lie alg} that we are assuming $\dim \gz$ is even, meaning that either $\g$ is even dimensional or $\sum_{x \in \bm z'} n_x$ is even. Hence $\dim \d$ is also even.

The non-degenerate bilinear form on $\gz$ defined in~\eqref{form on gz} restricts to the subalgebra $\da \subset \gz$. We~denote this restriction by
\begin{gather} \label{ip da}
\langle\!\langle \cdot, \cdot \rangle\!\rangle_\da \colon\ \da \times \da \longrightarrow \RR.
\end{gather}

\begin{Remark}
A more natural notation for the Lie group and Lie algebra in~\eqref{defect grp alg inf} would be~$\Gzp$ and $\gzp$. We~could also keep calling the induced bilinear form~\eqref{ip da} as $\langle\!\langle \cdot, \cdot \rangle\!\rangle_{\gz}$. The reason for using the above notation is that these will correspond to the standard notation for the Lie group on which the $\Ec$-model is defined.
\end{Remark}

Recalling from Section~\ref{sec: defects} that we are assuming $n_\infty = 2$, there is an obvious Lagrangian subalgebra of the factor $\g \otimes_\RR \T^2_\infty$ of $\gz$ at infinity given by the abelian subalgebra
\begin{gather*}
\g_{\rm ab} \coloneqq \g \otimes_\RR \varepsilon_\infty \RR[\varepsilon_\infty] / \big(\varepsilon_\infty^2\big).
\end{gather*}
For any choice of Lagrangian subalgebra $\k \subset \da$ we can then take $\fa = \g_{\rm ab} \oplus \k \subset \gz$ for the Lagrangian subalgebra used in Section~\ref{sec: edge modes}.

We denote by $G\big(\T^2_\infty\big)$ the factor of the Lie group $\Gz$ corresponding to the point at infinity. Concretely it is given by the tangent bundle $TG$ and as a Lie group it is isomorphic to $G \ltimes \g$. By a slight abuse of notation we will still denote by $G$ the Lie subgroup of $G\big(\T^2_\infty\big)$ identified with the subgroup $G \times \{ 0 \}$ of $G \ltimes \g$. Letting $G_{\rm ab} \subset G\big(\T^2_\infty\big)$, identified as $\{ \id \} \times \g \subset G \ltimes \g$, and $K \subset D$ denote the connected Lie subgroups corresponding to the Lagrangian subalgebras $\g_{\rm ab} \subset \g \otimes_\RR \T^2_\infty$ and $\k \subset \da$, then we also have the corresponding Lie subgroup $F = G_{\rm ab} \times K \subset \Gz = G\big(\T^2_\infty\big) \times D$.

Let $h_\infty \in C^\infty\big(\Sigma, G\big(\T^2_\infty\big)\big)$ be the component of the edge mode $h \in C^\infty\big(\Sigma, \Gz\big)$ at infinity. It~can be factorised uniquely as $h_\infty = v u$ for some $u \in C^\infty(\Sigma, G)$ and $v \in C^\infty(\Sigma, G_{\rm ab})$ relative to the global decomposition $G\big(\T^2_\infty\big) = G_{\rm ab} G$. Using the transformation~\eqref{gauge tr} with $f = \big(v^{-1}, \id_K\big)$ and $g = u$, we can then bring $h_\infty$ to the identity element.

Let $l \in C^\infty(\Sigma, D)$ denote the remaining components of the edge mode in $D$, so that we can write $h = \big(\id_{G(\T^2_\infty)}, l\big)$. The component of the condition~\eqref{L h constraint} at infinity then says that
$(\jb^\ast \L)|_{\g \otimes_\RR \T^2_\infty} \in \Omega^1(\Sigma, \g_{\rm ab})$,
which is equivalent to saying that the $1$-form $\L$ vanishes at infinity. In~terms of the notation introduced in the admissibility condition $(a)$ we therefore have $\L_{\rm c, \mu} = 0$ for $\mu = \sigma, \tau$. In~other words, having fixed the component of the edge mode at infinity to the identity, we will now focus on admissible $1$-forms $\L \in \Omega^1\big(X, \g^\CC\big)^\Pi$ of the form
\begin{gather} \label{admissible a inf}
\L_\mu = \sum_{\uz{y, q} \in \brb{\Pi \bzeta}} \frac{\L_\mu^{\uz{y, q}}}{(z-y)^{q+1}}.
\end{gather}
The remaining components of the constraint~\eqref{L h constraint} read
\begin{gather} \label{L l constraint}
\null^l (\jbt^\ast \L) \in \Omega^1(\Sigma, \k),
\end{gather}
where the map $\jbt^\ast$ is defined as in~\eqref{j pull} but with infinity removed, namely
\begin{gather} \label{j pull inf}
\jbt^\ast \colon\quad \Omega^1\big(X, \g^\CC\big)^\Pi \longrightarrow \Omega^1(\Sigma, \da), \qquad
\eta \longmapsto \bigg( \sum_{p=0}^{n_x - 1} \frac{1}{p!} \iota_x^\ast (\partial^p_z \eta) \otimes \varepsilon_x^p \bigg)_{x \in \bm z'}.
\end{gather}

Recall that in Section~\ref{sec: red to 2d} we assumed the existence of a collection of admissible $1$-forms $\L(h) \in \Omega^1\big(X, \g^\CC\big)^\Pi$ for each $h \in C^\infty\big(\Sigma, \Gz\big)$ satisfying the constraint~\eqref{L h constraint}. This allowed us to obtain the action~\eqref{ext action L} of a $2$d integrable field theory for the field $h \in C^\infty\big(\Sigma, \Gz\big)$ with associated Lax connection $\L(h)$. Moreover, we supposed that the 1-forms $\L(h)$ behave as~\eqref{L h gauge inv} under the gauge transformations~\eqref{gauge tr} of $h$, ensuring that these transformations define local symmetries of the action. In~the present subsection, we used part of these gauge symmetries to fix $h=(\id_{G(\T^2_\infty)}, l)$, with $l \in C^\infty(\Sigma, D)$.
Through this gauge fixing, finding admissible solutions~$\L(h)$ of the constraint~\eqref{L h constraint} for each $h \in C^\infty\big(\Sigma, \Gz\big)$, behaving under gauge transformations as in~\eqref{L h gauge inv}, then becomes equivalent to finding admissible $1$-forms $\L(l) \in \Omega^1\big(X, \g^\CC\big)^\Pi$ of the form~\eqref{admissible a inf} for each $l \in C^\infty(\Sigma, D)$, which solve the contraint equation~\eqref{L l constraint} and with the property that
\begin{gather} \label{constraint sol K-inv}
\jbt^\ast \L(kl) = \jbt^\ast \L(l)
\end{gather}
for all $k \in C^\infty(\Sigma, K)$. This last property follows from~\eqref{L h gauge inv} and describes the behaviour of~the collection of $1$-forms $\L(l)$ under what remains of the gauge symmetries~\eqref{gauge tr}, namely the transformations $l \mapsto k l$ for $k \in C^\infty(\Sigma, K)$. Performing the gauge fixing $h = (\id_{G(\T^2_\infty)}, l)$ in the action~\eqref{ext action L}, we then obtain
\begin{gather} \label{2d action}
S_{\rm 2d}(l) = - \frac{1}{2} \int_\Sigma \langle\!\langle l^{-1} {\rm d} l, \jbt^\ast \L(l) \rangle\!\rangle_\da - \frac 12 I^{\rm WZ}_{\da}[l],
\end{gather}
where the WZ-term $I^{\rm WZ}_{\da}[l]$ is defined in the same way as in~\eqref{WZ term}. By construction, the action~\eqref{2d action} is invariant under the residual gauge symmetry $l \mapsto k l$ with $k \in C^\infty(\Sigma, K)$, and defines an integrable field theory with Lax connection $\L(l)$.

\section[Integrable E-models from 4d Chern--Simons theory]
{Integrable $\boldsymbol{\Ec}$-models from 4d Chern--Simons theory}\label{sec: E-model from 4d}

In Section~\ref{sec: 2d from 4d} we reviewed the results of~\cite{Benini:2020skc} and arrived at the final expression~\eqref{2d action} for the action of a $2$d integrable field theory in the case when $\omega$ has a second order pole at infinity, i.e., $\omega$ is of~the form~\eqref{Eq:Omega}. Comparing the form of the action~\eqref{2d action} with that of the $\Ec$-model written in the form~\eqref{Eq:ActionJ} strongly suggests that the $2$d integrable field theories described by~\eqref{2d action} correspond to integrable $\Ec$-models.

Recall, however, that the derivation of the action~\eqref{2d action} hinges on the assumption made in~Section~\ref{sec: removing inf} that the constraint~\eqref{L l constraint} admits a solution $\L = \L(l) \in \Omega^1\big(X, \g^\CC\big)^\Pi$ in the subspace of admissible $\g^\CC$-valued $1$-forms, for every $l \in C^\infty(\Sigma, D)$, with the property~\eqref{constraint sol K-inv}. In~order to complete the description of the $2$d integrable field theory, it therefore remains to verify this assumption and explicitly construct solutions of the constraint equation~\eqref{L l constraint} within the admissible class of $1$-forms.

In order to construct a general class of solutions to the constraint~\eqref{L l constraint} in Section~\ref{sec: recover E-model}, we will see that we are naturally led to introduce an operator $\Ec \colon \da \to \da$ in Section~\ref{sec: admissible for E-model} which will correspond to the operator of the same name in the $\Ec$-model. The relationship between the actions~\eqref{Eq:ActionJ} and~\eqref{2d action} will then be made explicit in Section~\ref{sec: recover E-model}.

\subsection{The maps \texorpdfstring{$\bm j_{\bm z'}$}{jz} and \texorpdfstring{$\bm \pi_{\bm \zeta}$}{pizeta}} \label{sec: maps j pi}

Since the admissibility conditions $(a)$ and $(b)$ from Section~\ref{sec: red to 2d} are formulated in terms of the components $\L_\sigma$ and $\L_\tau$, it will be more convenient to express the constraint~\eqref{L l constraint} in terms of~these components as well.

Recall from Section~\ref{sec: red to 2d} the definition of the real vector space of $\Pi$-equivariant $V$-valued rational functions $R^\infty_{\Pi\bzeta}(V)^\Pi$ for any complex vector space $V$ equipped with an action of $\Pi$. Having removed the component of the edge mode at infinity in Section~\ref{sec: removing inf}, we are now working with admissible $1$-forms $\L \in \Omega^1\big(X, \g^\CC\big)^\Pi$ of the form~\eqref{admissible a inf}. It~is therefore convenient to introduce the subspace $R_{\Pi\bzeta}(V) \subset R^\infty_{\Pi\bzeta}(V)$ of $V$-valued rational functions which vanish at infinity.
This subspace is clearly stable under the action of $\Pi$ so that we may form the real vector space of $\Pi$-equivariants $R_{\Pi\bzeta}(V)^\Pi \subset R^\infty_{\Pi\bzeta}(V)^\Pi$. In~terms of this notation, we are therefore focusing on~the class of admissible $1$-form with components $\L_\sigma, \L_\tau \in R_{\Pi\bzeta}\big(C^\infty\big(\Sigma, \g^\CC\big)\big)^\Pi$.

We define, cf.~\eqref{j pull inf},
\begin{gather} \label{j z def}
\jb_{\bm z'} \colon\quad R_{\Pi\bzeta}\big(\g^\CC\big)^\Pi \longrightarrow \da, \qquad
f \longmapsto \bigg(\sum_{p=0}^{n_x - 1} \frac{1}{p!} (\partial^p_z f)|_x \otimes \varepsilon_x^p \bigg)_{x \in \bm z'},
\end{gather}
which returns the first $n_x$ terms in the Taylor expansion of the rational function at the set of~finite poles $\bm z'$ of $\omega$.
This extends component-wise to a morphism
\begin{gather*} 
\jb_{\bm z'} \colon\ R_{\Pi\bzeta}\big( C^\infty\big(\Sigma, \g^\CC\big) \big)^\Pi \longrightarrow C^\infty(\Sigma, \da).
\end{gather*}
Note that for an admissible $1$-form $\L = \L_\sigma {\rm d}\sigma + \L_\tau {\rm d}\tau \in \Omega^1\big(X, \g^\CC\big)$ we have the relation
$\jbt^\ast \L = \jb_{\bm z'} \L_\sigma {\rm d}\sigma + \jb_{\bm z'} \L_\tau {\rm d}\tau$ with the map $\jbt^\ast$ in~\eqref{j pull inf}. We~can then rewrite~\eqref{L l constraint} equivalently in~components as
\begin{gather} \label{L l constraint 2}
- \partial_\sigma l l^{-1} + \Ad_l (\jb_{\bm z'} \L_\sigma) \in C^\infty(\Sigma, \k), \qquad
- \partial_\tau l l^{-1} + \Ad_l (\jb_{\bm z'} \L_\tau) \in C^\infty(\Sigma, \k).
\end{gather}

We associate with the zeroes of $\omega$ the real vector space
\begin{gather} \label{gzeta def}
\gzeta \coloneqq \bigoplus_{\uz{y, q} \in \brb{\bzeta_{\rm r}}} \g \oplus \bigoplus_{\uz{y, q} \in \brb{\bzeta_{\rm c}}} \g^\CC,
\end{gather}
where $\g^\CC$ is regarded as a real vector space. Recall that $\brb{\bzeta}$ is the set of pairs $\uz{y, q}$ with $y \in \bm \zeta$ and $q = 0, \dots, m_y - 1$. Notice that the definition of $\gzeta$ in~\eqref{gzeta def} is very similar to that of $\gz$ in~\eqref{defect alg}. However, it is important to note that the former is only a vector space while the latter is a Lie algebra.

We shall also make use of the isomorphism
\begin{gather} \label{pi zeta def}
\bm \pi_{\bm \zeta} \colon\quad R_{\Pi\bzeta}\big(\g^\CC\big)^\Pi \overset{\cong}\longrightarrow \gzeta, \qquad
\sum_{\uz{y, q} \in \brb{\Pi \bzeta}} \frac{\ms u^{\uz{y,q}}}{(z-y)^{q+1}} \longmapsto \big(\ms u^{\uz{y, q}} \big)_{\uz{y, q} \in \brb{\bzeta}}
\end{gather}
which, as in the case of~\eqref{j z def}, extends component-wise to an isomorphism
\begin{gather} \label{pi zeta R}
\bm \pi_{\bm \zeta} \colon\ R_{\Pi\bzeta}\big(C^\infty\big(\Sigma, \g^\CC\big) \big)^\Pi \overset{\cong}\longrightarrow C^\infty\big(\Sigma, \gzeta\big).
\end{gather}
Applied explicitly to the components of the Lax connection in~\eqref{admissible a inf} this gives
\begin{gather*}
\bm \pi_{\bm \zeta} \L_\mu = \big( \L_\mu^{\uz{y, q}} \big)_{\uz{y, q} \in \bzeta}.
\end{gather*}
In particular, since~\eqref{pi zeta R} is an isomorphism, the field content of the Lax connection $\L$ in~\eqref{admissible a inf} is completely encoded in the collection of coefficients $\bm \pi_{\bm \zeta} \L_\mu \in C^\infty\big(\Sigma, \gzeta\big)$ in the partial fraction decomposition of its components at the set of zeroes $\bzeta$ of $\omega$.

\subsection[Admissible 1-forms for the E-model]
{Admissible 1-forms for the $\boldsymbol{\Ec}$-model}\label{sec: admissible for E-model}

The admissibility condition $(b)$ can easily be solved by choosing an $\epsilon_y \in \RR \setminus \{ 0 \}$ for each $y \in \bzeta_{\rm r}$ and an $\epsilon_y \in \CC \setminus \{ 0 \}$ for each $y \in \bzeta_{\rm c}$, and requiring that
\begin{gather}\label{Eq:Lz}
\L_\tau^{\uz{y, q}} = \epsilon_y \L_\sigma^{\uz{y, q}}
\end{gather}
for all $\uz{y, q} \in \brb{\Pi \bzeta}$, where we define $\epsilon_{\bar y} \coloneqq \overline{\epsilon_y}$ for $\bar y \in \bar\bzeta_{\rm c}$. This condition, which also appeared in~\cite{Benini:2020skc, Costello:2019tri, Delduc:2019whp, Vicedo:2019dej}, is motivated by the expression for the Lax connection in affine Gaudin models. In~the case when all the zeroes of $\omega$ are simple it takes the form \cite[equations~(2.39)--(2.40)]{Delduc:2019bcl}, which is to be compared with the expressions in the admissibility condition~\eqref{admissible a inf} for $m_y = 1$, combined with~\eqref{Eq:Lz}. In~order to make use of the condition~\eqref{Eq:Lz} to solve the constraint~\eqref{L l constraint 2}, it will be convenient to first reformulate it as a relation between $\jb_{\bm z'} \L_\sigma$ and $\jb_{\bm z'} \L_\tau$.

Consider the linear isomorphism
\begin{gather} \label{Etilde def}
\widetilde{\Ec} \colon\ \gzeta \overset{\cong}\longrightarrow \gzeta, \qquad
\big( \ms u^{\uz{y,q}} \big)_{\uz{y,q} \in \brb{\bzeta}} \longmapsto \big( \epsilon_y \ms u^{\uz{y,q}} \big)_{\uz{y,q} \in \brb{\bzeta}}.
\end{gather}
We may then rewrite~\eqref{Eq:Lz} as
\begin{gather} \label{admissible E def}
\bm \pi_{\bm \zeta} \L_\tau = \widetilde{\Ec}(\bm \pi_{\bm \zeta} \L_\sigma).
\end{gather}

\begin{Lemma} \label{lem: Cauchy}
We have an isomorphism of real vector spaces
\begin{gather*}
\C \coloneqq \jb_{\bm z'} \circ \bm \pi_{\bm \zeta}^{-1} \colon\ \gzeta \overset{\cong}\longrightarrow \da.
\end{gather*}
In particular, the linear map $\jb_{\bm z'}$ defined in~\eqref{j z def} is also an isomorphism.
\begin{proof}
Let $\ms U = \big(\ms u^{\uz{y,q}}\big)_{\uz{y,q} \in \brb{\bzeta}} \in \gzeta$. The components of $\jb_{\bm z'} \big(\bm \pi_{\bzeta}^{-1} \ms U\big) \in \da$ are obtained by taking the first $n_x$ terms in the Taylor expansion of the rational function $\bm \pi_{\bzeta}^{-1} \ms U$ at each $x \in \bz'$. For~the purpose of this proof, it is necessary to also consider separately the Taylor expansions at the conjugate poles $\bar x \in \bar \bz_{\rm c}$ for each $x \in \bz_{\rm c}$, even though these are related to the Taylor expansions at $x$ by the automorphism $\tau$. Explicitly, the coefficients of the expansions at all the poles $x \in \bZ'$ are given by
\begin{gather*}
\frac{1}{p!} \big( \partial_z^p (\bm \pi_{\bm \zeta}^{-1} \ms U) \big)\big|_x = \sum_{\uz{y,q} \in \brb{\Pi \bzeta}} C^{\up{x, p}}_{\quad\;\; \uz{y,q}} \ms u^{\uz{y,q}}
\end{gather*}
for all $\up{x,p} \in \bsb{\bZ'}$, where we have introduced the notation $\ms u^{\uz{\bar y,q}} \coloneqq \tau \big(\ms u^{\uz{y,q}}\big)$ for any $y \in \bzeta_{\rm c}$ and
\begin{gather} \label{Cauchy mat}
C^{\up{x, p}}_{\quad\;\; \uz{y,q}} \coloneqq \binom{p+q}{p} \frac{(-1)^p}{(x-y)^{p+q+1}}
\end{gather}
for all $\up{x,p} \in \bsb{\bZ'}$ and $\uz{y,q} \in \brb{\Pi \bzeta}$.
The expressions in~\eqref{Cauchy mat} are the components of what is known as a confluent Cauchy matrix, see for instance \cite[equation~(13)]{Vavrin}.

The map $\C \colon \gzeta \to \da$ is then defined for any
\[
\ms U = \big(\ms u^{\uz{y,q}}\big)_{\uz{y,q} \in \brb{\bzeta}} \in \gzeta$ by $\C(\ms U) = \big( \C (\ms U)^{x, p} \otimes \varepsilon_x^p \big)_{\up{x, p} \in \bsb{\bz'}},
\]
where
\begin{gather} \label{map C explicit}
\C (\ms U)^{x, p} = \sum_{\uz{y,q} \in \brb{\bzeta_{\rm r}}} C^{\up{x, p}}_{\quad\;\; \uz{y,q}} \ms u^{\uz{y,q}} + \sum_{\uz{y,q} \in \brb{\bzeta_{\rm c}}} \Big( C^{\up{x, p}}_{\quad\;\; \uz{y,q}} \ms u^{\uz{y,q}} + C^{\up{x, p}}_{\quad\;\; \uz{\bar y,q}} \tau \big(\ms u^{\uz{y,q}}\big) \Big),
\end{gather}
in terms of the Cauchy matrix~\eqref{Cauchy mat}.

On the other hand, since $\omega$ is a meromorphic differential with zeroes at each $y \in \Pi \bzeta$ of order~$m_y$ and poles at each $x \in \bZ$ of order $n_x$, we have
\begin{gather*}
\sum_{y \in \Pi \bzeta} m_y = \sum_{x \in \bZ} n_x - 2.
\end{gather*}
In other words, since we are assuming that $n_\infty = 2$, this yields
\begin{gather*}
\sum_{y \in \Pi \bzeta} m_y = \sum_{x \in \bZ'} n_x.
\end{gather*}
It follows that~\eqref{Cauchy mat} are the components of a \emph{square} confluent Cauchy matrix which is known to be invertible \cite[Corollary 10]{Vavrin}.

The inverse $\C^{-1} \colon \da \to \gzeta$ is then given explicitly, for any $\ms V = \big( \ms v^{x,p} \otimes \varepsilon^p_x \big)_{\up{x, p} \in \bsb{\bz'}} \in \da$, by~$\C^{-1} (\ms V) = \big( \C^{-1}(\ms V)^{\uz{y, q}} \big)_{\uz{y, q} \in \brb{\bzeta}}$, where
\begin{gather}
\C^{-1} (\ms V)^{\uz{y, q}} = \sum_{\up{x,p} \in \bsb{\bz_{\rm r}}} \big(C^{-1}\big)_{\quad\;\; \up{x, p}}^{\uz{y,q}} \ms v^{x,p} \notag
\\ \hphantom{\C^{-1} (\ms V)^{\uz{y, q}} =}
{} + \sum_{\up{x,p} \in \bsb{\bz_{\rm c}}} \Big( \big(C^{-1}\big)_{\quad\;\; \up{x, p}}^{\uz{y,q}} \ms v^{x,p} + \big(C^{-1}\big)_{\quad\;\; \up{\bar x, p}}^{\uz{y,q}} \tau (\ms v^{x,p}) \Big),
\label{map C inv explicit}
\end{gather}
in terms of the inverse $\big(C^{-1}\big)_{\quad\;\; \up{x, p}}^{\uz{y,q}}$ of the Cauchy matrix~\eqref{Cauchy mat}.
\end{proof}
\end{Lemma}

Recall from Section~\ref{sec: maps j pi} that the field content of the Lax connection $\L$ in~\eqref{admissible a inf} is encoded in the collection of coefficients $\bm \pi_{\bm \zeta} \L_\mu \in C^\infty\big(\Sigma, \gzeta\big)$. By Lemma~\ref{lem: Cauchy}, the field content of $\L$ is equivalently encoded in the jets $\jb_{\bm z'} \L_\mu \in C^\infty(\Sigma, \da)$ of its components at the set of finite poles $\bz'$ of $\omega$.

Applying the linear isomorphism $\C$ from Lemma~\ref{lem: Cauchy} to both sides of the relation~\eqref{admissible E def} we may obtain
\begin{gather}\label{Eq:Lz 2}
\jb_{\bm z'} \L_\tau = \Ec (\jb_{\bm z'} \L_\sigma),
\end{gather}
where we have defined the linear isomorphism
\begin{gather} \label{E def}
\Ec \coloneqq \C \circ \widetilde{\Ec} \circ \C^{-1} \colon\ \d \overset{\cong}\longrightarrow \d.
\end{gather}

Using the explicit forms~\eqref{Etilde def},~\eqref{map C explicit} and~\eqref{map C inv explicit} for the linear maps $\widetilde\Ec$, $\C$ and $\C^{-1}$, we may exp\-ress the linear isomorphism~\eqref{E def} in components as follows. For~every $\ms U = \big( \ms u^{y,q} \otimes \varepsilon^q_y \big)_{\up{y, q} \in \bsb{\bz'}} \in \da$ we have $\Ec (\ms U) = \big( \Ec (\ms U)^{x, p} \otimes \varepsilon^p_x \big)_{\up{x, p} \in \bsb{\bz'}}$, where
\begin{gather} \label{map E explicit}
\Ec (\ms U)^{x, p} = \sum_{\up{y,q} \in \bsb{\bz_{\rm r}}} E^{\up{x, p}}_{\quad\;\; \up{y,q}} \ms u^{y,q} + \sum_{\up{y,q} \in \bsb{\bz_{\rm c}}} \Big( E^{\up{x, p}}_{\quad\;\; \up{y,q}} \ms u^{y,q} + E^{\up{x, p}}_{\quad\;\; \up{\bar y,q}} \tau (\ms u^{y,q}) \Big),
\end{gather}
for some coefficients $E^{\up{x, p}}_{\quad\;\; \up{y,q}}$ expressible in terms of the Cauchy matrix~\eqref{Cauchy mat}, its inverse and the choice of $\epsilon_y$ for $y \in \bzeta$.

\subsection[Properties of E]{Properties of $\boldsymbol{\Ec}$}\label{SSec:EQ}

Recall that the Lie algebra $\da$ defined in Section~\ref{sec: removing inf} comes equipped with a non-degenerate symmetric invariant bilinear form $\langle\!\langle \cdot, \cdot \rangle\!\rangle_\da$. To study the symmetry property of $\Ec$ with respect to the latter, it is convenient to pull back this bilinear form to the vector space $\gzeta$ along the linear isomorphism from Lemma~\ref{lem: Cauchy} since the action of $\widetilde{\Ec}$ is much simpler. It~is useful to do this in two steps, by first pulling back $\langle\!\langle \cdot, \cdot \rangle\!\rangle_\da$ to $R_{\Pi\bzeta}\big(\g^\CC\big)^\Pi$ along the isomorphism~\eqref{j z def}, and then to~$\gzeta$ along the inverse of the isomorphism $\bm \pi_{\bm \zeta}$ given in~\eqref{pi zeta def}.

We define the non-degenerate symmetric bilinear form
\begin{subequations} \label{bilinear R zeta}
\begin{gather}
\langle\!\langle \cdot, \cdot \rangle\!\rangle_\omega \colon\ R_{\Pi\bzeta}\big(\g^\CC\big)^\Pi \times R_{\Pi\bzeta}\big(\g^\CC\big)^\Pi \longrightarrow \RR,
\end{gather}
defined for any $f, g \in R_{\Pi\bzeta}\big(\g^\CC\big)^\Pi$ by
\begin{gather}
\langle\!\langle f, g \rangle\!\rangle_\omega \coloneqq \sum_{x \in \bm z'} \frac{2}{|\Pi_x|} \Re \big( \!\res_x \langle f, g \rangle \omega \big) = - \sum_{y \in \bm \zeta} \frac{2}{|\Pi_y|} \Re\big( \!\res_y \langle f, g \rangle \omega \big).
\end{gather}
\end{subequations}
The equality here follows from the vanishing of the sum of residues, after observing that the poles of $\langle f, g \rangle \omega$ belong to the set $\bZ' \sqcup \Pi \bzeta = \Pi \bz' \sqcup \Pi \bzeta$. (Observe that infinity is not a pole of~$\langle f, g \rangle \omega$ because the double pole of $\omega$ at infinity is compensated by the simple zeroes of $f$ and~$g$ there.) Note in particular that
\begin{gather*}
\sum_{x \in \bz'} \frac{2}{|\Pi_x|} \Re \big( \!\res_x \langle f, g \rangle \omega \big) = \sum_{x \in \bZ'} \res_x \langle f, g \rangle \omega
\end{gather*}
and similarly for the sum of residues at the zeroes $\bzeta$.

\begin{Lemma} \label{lem: bilinear forms 1}
For any $f, g \in R_{\Pi\bzeta}\big(\g^\CC\big)^\Pi$, we have $\langle\!\langle \jb_{\bm z'} f, \jb_{\bm z'} g \rangle\!\rangle_\da = \langle\!\langle f, g \rangle\!\rangle_\omega$.
\begin{proof}
Let $f, g \in R_{\Pi\bzeta}\big(\g^\CC\big)^\Pi$. By definition we have
\begin{align*}
\langle\!\langle f, g \rangle\!\rangle_\omega &= \sum_{x \in \bm z'} \frac{2}{|\Pi_x|} \Re \big( \!\res_x \langle f, g \rangle \omega \big)
= \sum_{x \in \bm z'} \sum_{p=0}^{n_x-1} \frac{2}{|\Pi_x|} \Re \bigg( \!\res_x \langle f, g \rangle \frac{\ell^x_p {\rm d}z}{(z - x)^{p+1}} \bigg)\\
&= \sum_{x \in \bm z'} \sum_{p=0}^{n_x-1} \frac{2}{|\Pi_x|} \Re \bigg( \frac{\ell^x_p}{p!} \big( \partial_z^p \langle f, g \rangle \big)\big|_x \bigg)\\
&= \sum_{x \in \bm z'} \sum_{p=0}^{n_x-1} \sum_{q=0}^p \frac{2}{|\Pi_x|} \Re \bigg( \ell^x_p \bigg\langle \frac{1}{q!} (\partial_z^q f)|_x, \frac{1}{(p-q)!} (\partial_z^{p-q} g)|_x \bigg\rangle \bigg)\\
&= \sum_{x \in \bm z'} \sum_{q,r=0}^{n_x-1} \frac{2}{|\Pi_x|} \Re \bigg( \ell^x_{q+r} \bigg\langle \frac{1}{q!} (\partial_z^q f)|_x, \frac{1}{r!} (\partial_z^r g)|_x \bigg\rangle \bigg)
= \langle\!\langle \jb_{\bm z'} f, \jb_{\bm z'} g \rangle\!\rangle_\da,
\end{align*}
where in the second equality we used the explicit expression~\eqref{Eq:Omega} for $\omega$, dropping the term with~$\ell^\infty_1$ since it does not contribute to the residue at any of the finite poles $x \in \bm z'$. In~the second last step we changed variable from $p$ to $r = p-q$ and used the convention that $\ell^x_p = 0$ for~$p \geq n_x$. The last equality is by definition~\eqref{form on gz} of the induced bilinear form $\langle\!\langle \cdot, \cdot \rangle\!\rangle_\da$ on~$\da \subset \gz$ and of the map $\jb_{\bm z'}$ in~\eqref{j z def}.
\end{proof}
\end{Lemma}

We introduce the symmetric bilinear form
\begin{subequations} \label{bilinear g zeta}
\begin{gather}
\langle\!\langle \cdot, \cdot \rangle\!\rangle_{\gzeta} \colon\ \gzeta \times \gzeta \longrightarrow \RR,
\end{gather}
defined for any $\ms U = (\ms u^{\uz{x,p}})_{\uz{x,p}\in \brb{\bzeta}}$, $\ms V = (\ms v^{\uz{y, q}})_{\uz{y,q} \in \brb{\bzeta}} \in \gzeta$ by
\begin{gather}
\langle\!\langle \ms U, \ms V \rangle\!\rangle_{\gzeta} \coloneqq - \sum_{y \in \bm \zeta} \sum_{\substack{p, q = 0\\ p+q \geq m_y - 1}}^{m_y - 1} \frac{2}{|\Pi_y|} \Re \bigg( \frac{\big( \partial_z^{p+q+1-m_y} \psi_y\big)(y)}{(p+q+1-m_y)!} \big\langle \ms u^{\uz{y, p}}, \ms v^{\uz{y, q}} \big\rangle \bigg).
\end{gather}
\end{subequations}
Here we wrote the twist function $\varphi(z)$ defined in~\eqref{Eq:Omega} as $\varphi(z) = \psi_y(z) (z - y)^{m_y}$ with $\psi_y(y) \neq 0$ using the fact that it has a zero of order $m_y$ at $y \in \bm \zeta$. This definition is motivated by the following lemma.

\begin{Lemma} \label{lem: bilinear forms 2}
For any $f, g \in R_{\Pi\bzeta}\big(\g^\CC\big)^\Pi$, we have $\langle\!\langle \bm\pi_{\bzeta} f, \bm\pi_{\bzeta} g \rangle\!\rangle_{\gzeta} = \langle\!\langle f, g \rangle\!\rangle_\omega$.
In particular, for any $\ms U, \ms V \in \gzeta$ we have $\langle\!\langle \ms U, \ms V \rangle\!\rangle_{\gzeta} = \langle\!\langle \C\ms U, \C\ms V \rangle\!\rangle_\d$.
\end{Lemma}

\begin{proof}
Let $f, g \in R_{\Pi\bzeta}\big(\g^\CC\big)^\Pi$ which we can write out explicitly as
\begin{gather*}
f(z) = \sum_{\uz{y,p} \in \brb{\Pi \bzeta}} \frac{\ms u^{\uz{y,p}}}{(z-y)^{p+1}}, \qquad
g(z) = \sum_{\uz{x,q} \in \brb{\Pi \bzeta}} \frac{\ms v^{\uz{x,q}}}{(z-x)^{q+1}}.
\end{gather*}
Using the second expression for the bilinear form in~\eqref{bilinear R zeta} we then find
\begin{align*}
\langle\!\langle f, g \rangle\!\rangle_\omega &= - \sum_{y \in \bm \zeta} \frac{2}{|\Pi_y|} \Re\big( \!\res_y \langle f, g \rangle \omega \big) = - \sum_{y \in \bzeta} \sum_{p=0}^{m_y-1} \frac{2}{|\Pi_y|} \Re \bigg( \res_y \bigg\langle \frac{\ms u^{\uz{y,p}}}{(z-y)^{p+1}}, \varphi g \bigg\rangle {\rm d}z \bigg)
\\
&= - \sum_{y \in \bm \zeta} \sum_{p = 0}^{m_y-1} \frac{2}{|\Pi_y|} \Re \bigg\langle \ms u^{\uz{y,p}}, \frac{1}{p!} \big( \partial_z^p (\varphi g) \big) \big|_y \bigg\rangle\\
&= - \sum_{y \in \bm \zeta} \sum_{p,q = 0}^{m_y-1} \frac{2}{|\Pi_y|} \Re \bigg( \frac{1}{p!} \partial_z^p \bigg( \frac{\varphi(z)}{(z-y)^{q+1}} \bigg)\bigg|_y \langle \ms u^{\uz{y,p}}, \ms v^{\uz{y, q}} \rangle \bigg).
\end{align*}
In the second equality we wrote $\omega = \varphi {\rm d}z$. In~the third equality we used the fact that $\varphi g$ is regular at $\bm \zeta$, so that only the poles from $f$ contribute to the residue at each $y \in \bm \zeta$, and took the residue. In~the last equality we have used the fact that the terms in $\partial_z^q(\varphi g)$ coming from the poles of $g$ at $x \neq y$ all vanish at $y$.
Using $\varphi(z) = \psi_y(z) (z - y)^{m_y}$ we find
\begin{gather*}
\frac{1}{p!} \partial_z^p \bigg( \frac{\varphi(z)}{(z-y)^{q+1}} \bigg)\bigg|_y = \frac{1}{p!} \partial_z^p \big( \psi_y(z) (z-y)^{m_y-q-1} \big)\big|_y
= \frac{\big( \partial_z^{p+q+1-m_y} \psi_y\big)(y)}{(p+q+1-m_y)!},
\end{gather*}
where in the last step we have used the Leibniz rule and the fact that if any term still contains a factor of $(z-y)$ it will vanish upon setting $z=y$.

\looseness=1
The final statement follows form the above, Lemma~\ref{lem: bilinear forms 1} and the definition of $\C$ in Lem\-ma~\ref{lem: Cauchy}.
\end{proof}

\begin{Corollary}\label{CorE}
$\Ec$ is symmetric with respect to $\langle\!\langle \cdot, \cdot \rangle\!\rangle_\da$.
\end{Corollary}

\begin{proof}
Let $\ms U, \ms V \in \da$. Using~\ref{lem: bilinear forms 2} we have
\begin{gather*}
\langle\!\langle \ms U, \Ec \ms V \rangle\!\rangle_\da = \langle\!\langle \C^{-1} \ms U, \C^{-1} \Ec \ms V \rangle\!\rangle_{\gzeta} = \langle\!\langle \C^{-1} \ms U, \widetilde{\Ec} \C^{-1} \ms V \rangle\!\rangle_{\gzeta}.
\end{gather*}
Since $\widetilde{\Ec}$ is clearly symmetric with respect to~\eqref{bilinear g zeta} the claim follows.
\end{proof}

\begin{Remark} \label{rem: CorE simple pole}
When all the zeroes of $\omega$ are simple and real, i.e., $m_y = 1$ for every $y \in \bm \zeta$ and $\bzeta = \bzeta_{\rm r}$, a simple condition for $\Ec$ to be positive (i.e., $\langle\!\langle \cdot, \Ec \cdot \rangle\!\rangle_\da$ to be positive-definite) can be given in the case when $\g$ is compact, in which case we can choose the bilinear form $\langle \cdot, \cdot \rangle\colon \g \times \g \to \RR$ to be positive definite. Specifically, in this case the bilinear form~\eqref{bilinear g zeta} on $\gzeta$ reduces simply to
\begin{gather*}
\langle\!\langle \ms U, \ms V \rangle\!\rangle_{\gzeta} = - \sum_{y \in \bzeta_{\rm r}} \psi_y(y) \big\langle \ms u^{\uz{y, p}}, \ms v^{\uz{y, q}} \big\rangle.
\end{gather*}
It then follows directly from the proof of Corollary~\ref{CorE} that $\Ec$ is positive with respect to $\langle\!\langle \cdot, \cdot \rangle\!\rangle_\da$ if and only if $\widetilde \Ec$ is positive with respect to $\langle\!\langle \cdot, \cdot \rangle\!\rangle_{\gzeta}$, i.e., if and only if $- \epsilon_y \psi_y(y) > 0$ for every $y \in \bzeta$. Noting that $\psi_y(y) = \varphi'(y)$ this means $\varphi'(y)$ and $\epsilon_y$ should have opposite signs.
\end{Remark}

\subsection[Recovering the E-model]{Recovering the $\boldsymbol{\Ec}$-model}\label{sec: recover E-model}

In Section~\ref{sec: admissible for E-model} we described a very simple class of admissible $\g^\CC$-valued $1$-form $\L \in \Omega^1\big(X, \g^\CC\big)^\Pi$, namely ones satisfying the condition~\eqref{Eq:Lz}. We~showed that the latter could be rewritten in the form~\eqref{Eq:Lz 2} in terms of the linear isomorphism $\Ec \in \End \da$ defined in~\eqref{E def}. We~will now show that, assuming $\Ec$ and $\k \subset \da$ satisfy the condition~\eqref{Eq:Assump}, there exists a \emph{unique} solution to the constraint~\eqref{L l constraint 2} for $\L$ in terms of $l \in C^\infty(\Sigma, D)$ within this class of admissible $1$-forms. In~Section~\ref{SSec:EQ}, specifically Remark~\ref{rem: CorE simple pole}, we gave sufficient conditions for~\eqref{Eq:Assump} to hold in view of~Lemma~\ref{Lem:Pos}.

Let us define
$B_\sigma \coloneqq - \partial_\sigma l l^{-1} + \Ad_l (\jb_{\bm z'} \L_\sigma)$ and $B_\tau \coloneqq - \partial_\tau l l^{-1} + \Ad_l (\jb_{\bm z'} \L_\tau)$,
which both belong to $C^\infty(\Sigma, \k)$ by~\eqref{L l constraint 2}. Then using the relation~\eqref{Eq:Lz 2} we deduce
\begin{gather} \label{At As rel}
\Ad_l^{-1} B_\tau = \Ec \Ad_l^{-1} B_\sigma + \Ec\big(l^{-1} \partial_\sigma l\big) - l^{-1} \partial_\tau l.
\end{gather}
The left hand side takes value in $\Ad_l^{-1} \k = \ker\Pc_l$ while the first term on the right hand side is valued in $\Ec \Ad_l^{-1} \k = \im\Pc_l$. Here $\Pc_l$ is the projector with kernel and image~\eqref{Eq:KerIm} defined in~Section~\ref{sec: operators E Pl}, where now $l \in C^\infty(\Sigma, D)$ is the edge mode from Section~\ref{sec: removing inf}. Note that the existence of this projector is ensured by the condition~\eqref{Eq:Assump} which we are assuming holds.

Applying $\Pc_l$ to both sides of the equation~\eqref{At As rel} we then obtain
\begin{gather*}
0 = \Ec \big( \jb_{\bm z'} \L_\sigma - l^{-1} \partial_\sigma l \big) + \Pc_l \big( \Ec\big(l^{-1} \partial_\sigma l\big) - l^{-1} \partial_\tau l \big).
\end{gather*}
We can now solve this for $\jb_{\bm z'} \L_\sigma$ and then substitute the result into the relation~\eqref{Eq:Lz 2} to find $\jb_{\bm z'} \L_\tau$, yielding
\begin{gather*}
\jb_{\bm z'} \L_\sigma = \big(\id - \Ec^{-1} \Pc_l \Ec \big) \big(l^{-1} \partial_\sigma l\big) + \Ec^{-1} \Pc_l\big(l^{-1} \partial_\tau l\big),\\
\jb_{\bm z'} \L_\tau = (\Ec - \Pc_l \Ec) \big(l^{-1} \partial_\sigma l\big) + \Pc_l \big(l^{-1} \partial_\tau l\big).
\end{gather*}
Finally, note that using part $(ii)$ of Proposition~\ref{PropE} we can rewrite these as
\begin{subequations} \label{L for E-model}
\begin{gather}
\jb_{\bm z'} \L_\sigma = \overline{\Pc}_l\big(l^{-1} \partial_\sigma l\big) + \Ec^{-1} \Pc_l\big(l^{-1} \partial_\tau l\big), \\
\jb_{\bm z'} \L_\tau = \Ec \overline{\Pc}_l \big(l^{-1} \partial_\sigma l\big) + \Pc_l \big(l^{-1} \partial_\tau l\big).
\end{gather}
\end{subequations}
Since $\jb_{\bz'}$ is invertible by Lemma~\ref{lem: Cauchy}, this gives the desired unique solution $\L = \L(l)$ of the constraint~\eqref{L l constraint}. See Section~\ref{sec: inverse jz} below.

We observe that the expressions~\eqref{L for E-model} coincide with those in~\eqref{Eq:J} for the current $\J$ in the $\Ec$-model. It~is now clear, as advertised at the start of this section, that the action~\eqref{2d action} for the solution $\L = \L(l)$ to the constraint equation~\eqref{L l constraint} which we have obtained in this section coincides exactly with the action of the $\Ec$-model written in the form~\eqref{Eq:ActionJ}. Explicitly, we then have
\begin{gather}
S_{\rm 2d}(l) = \frac{1}{2} \int_\Sigma \big( \langle\!\langle l^{-1} \partial_\tau l, \Ec^{-1} \Pc_l \big(l^{-1} \partial_\tau l\big) \rangle\!\rangle_\da - \langle\!\langle l^{-1} \partial_\sigma l, \Ec \overline{\Pc}_l \big(l^{-1} \partial_\sigma l\big) \rangle\!\rangle_\da \notag
\\ \hphantom{S_{\rm 2d}(l) = \frac{1}{2} \int_\Sigma \big(}
{}+ \langle\!\langle l^{-1} \partial_\tau l, \big(\overline{\Pc}_l - \null^t \Pc_l\big) \big(l^{-1} \partial_\sigma l\big) \rangle\!\rangle_\da \big) {\rm d}\sigma \wedge {\rm d}\tau - \frac{1}{2} I^{\rm WZ}_{\da}[l],
\label{Eq:Action}
\end{gather}
which coincides with $S_{\Ec,\k}(l)$ defined in~\eqref{Eq:ActionE}. We~also note that, as required from Section~\ref{sec: removing inf}, the solution $\L(l)$ of the constraint~\eqref{L l constraint} satisfies~\eqref{constraint sol K-inv} for any $k \in C^\infty(\Sigma, K)$ since the expressions for $\J$ were noted in Section~\ref{sec: eom E-model} to have this property.

\subsection{The inverse of \texorpdfstring{$\bm j_{\bm z'}$}{jz}} \label{sec: inverse jz}

The admissibility condition $(b)$ from Section~\ref{sec: red to 2d} plays a central role in the passage from $4$d Chern--Simons theory to $2$d integrable field theories in the approach described in~\cite{Benini:2020skc}. In~particular, by \cite[Proposition 5.6]{Benini:2020skc} it allows one to lift the flatness equation for the $1$-form $\J = \J_\sigma {\rm d}\sigma + \J_\tau {\rm d}\tau$ with components $\J_\mu \coloneqq \jb_{\bz'} \L_\mu$, i.e.,
\begin{gather} \label{J flat}
{\rm d}\J + \frac 12 [\J, \J] = 0,
\end{gather}
which is essentially the boundary equation of motion for the extended action~\eqref{ext action}, to the flatness of the Lax connection $\L$ itself, namely
\begin{gather} \label{L flat}
{\rm d}\L + \frac 12 [\L, \L] = 0.
\end{gather}
In this section we will give a different perspective on the above passage from~\eqref{J flat} to~\eqref{L flat} in the case when the admissibility condition $(b)$ is ensured by the $\Ec$-model condition~\eqref{Eq:Lz 2}. It~follows from Section~\ref{sec: recover E-model} that in this case the flatness equation~\eqref{J flat} coincides with the equations of motion for the $\Ec$-model by Proposition~\ref{prop: E-model eom}.

Recall from Lemma~\ref{lem: Cauchy} that the map $\jb_{\bz'}$ defined in~\eqref{j z def} is an isomorphism. We~denote its inverse by
\begin{gather} \label{Severa map}
\bp \colon\ \da \longrightarrow R_{\Pi\bzeta}\big(\g^\CC\big)^\Pi.
\end{gather}
Applying this map to both equations in~\eqref{L for E-model} it follows that the components of the Lax connection $\L = \L_\sigma {\rm d}\sigma + \L_\tau {\rm d}\tau$ of the $\Ec$-model are given by
\begin{gather} \label{Lax for E-model}
\L_\sigma = \bp \J_\sigma, \qquad
\L_\tau = \bp \J_\tau.
\end{gather}

\begin{Lemma} \label{lem: p morph}
For any $\ms U \in \da$ we have $\bp([ \ms U, \Ec \ms U]) = [ \bp \ms U, \bp\Ec \ms U]$.
\begin{proof}
Let $\ms U \in \da$. Since $\jb_{\bz'}$ is an isomorphism we can write it as $\ms U = \jb_{\bz'} f$ for some $f \in R_{\Pi\bzeta}(\g)^\Pi$. Then we have
\begin{gather*}
[ \bp \ms U, \bp\Ec \ms U] = \big[ f, \bm \pi_{\bzeta}^{-1} \widetilde{\Ec} \bm \pi_{\bzeta} f \big]
= \bp \jb_{\bz'} \big( \big[ f, \bm \pi_{\bzeta}^{-1} \widetilde{\Ec} \bm \pi_{\bzeta} f \big] \big),
\end{gather*}
where in the first equality we substituted $\Ec = \jb_{\bz'} \bm \pi_{\bzeta}^{-1} \widetilde{\Ec} \bm \pi_{\bzeta} \jb_{\bz'}^{-1}$ and $\ms U = \jb_{\bz'} f$ and used the fact that $\bp \jb_{\bz'} = \id$. The second equality follows from noting that
\begin{gather} \label{admissibility f}
\big[ f, \bm \pi_{\bzeta}^{-1} \widetilde{\Ec} \bm \pi_{\bzeta} f \big] \in R_{\Pi\bzeta}\big(\g^\CC\big)^\Pi,
\end{gather}
i.e., that the order of the pole of $\big[ f, \bm \pi_{\bzeta}^{-1} \widetilde{\Ec} \bm \pi_{\bzeta} f \big]$ at each $y \in \Pi \bzeta$ is of order at most $m_y$, and inserting the identity on $R_{\Pi\bzeta}\big(\g^\CC\big)^\Pi$ in the form $\id = \bp \jb_{\bz'}$. Indeed, by definitions~\eqref{pi zeta def} and~\eqref{Etilde def} of the operators $\bm \pi_{\bzeta}$ and $\widetilde{\Ec}$, we have that
\begin{align*}
\bm \pi_{\bzeta}^{-1} \widetilde{\Ec} \bm \pi_{\bzeta} \colon\ R_{\Pi\bzeta}\big(\g^\CC\big)^\Pi &\longrightarrow R_{\Pi\bzeta}\big(\g^\CC\big)^\Pi,
\\
\sum_{\uz{y, q} \in \brb{\Pi \bzeta}} \frac{\ms u^{\uz{y,q}}}{(z-y)^{q+1}} &\longmapsto \sum_{\uz{y, q} \in \brb{\Pi \bzeta}} \frac{\epsilon_y \ms u^{\uz{y,q}}}{(z-y)^{q+1}}.
\end{align*}
We thus see that~\eqref{admissibility f} is true for precisely the same reason that the relation~\eqref{Eq:Lz} we imposed on the components of $\L$ in Section~\ref{sec: admissible for E-model} solves the admissibility condition $(b)$.

Noting that the operation of taking jets of $\g$-valued functions is a morphism of Lie algebras, we then obtain that
\begin{gather*}
[ \bp \ms U, \bp\Ec \ms U] = \bp \big( \big[ \jb_{\bz'} f, \jb_{\bz'} \bm \pi_{\bzeta}^{-1} \widetilde{\Ec} \bm \pi_{\bzeta} f \big] \big)
= \bp( [\ms U, \Ec \ms U]),
\end{gather*}
where the last step is by definition of $\Ec$ in~\eqref{E def} and the fact that $\ms U = \jb_{\bz'} f$.
\end{proof}
\end{Lemma}

\begin{Remark} 
In the relativistic case $\Ec^2 = \id$, we have the following direct comparison with the results of~\cite{Severa:2017kcs}.
Let $\da_\pm \coloneqq \ker(\Ec \mp \id) = \im(\Ec \pm \id) \subset \da$ denote the $\pm 1$ eigenspaces of $\Ec$ in $\da$, so~that we have a direct sum decomposition $\da = \da_+ \dotplus \da_-$.
The statement of Lemma~\ref{lem: p morph} can then be rephrased as follows: for any $\ms U_\pm \in \da_\pm$ we have
\begin{gather*}
\bp \big( [ \ms U_+, \ms U_-] \big) = \big[ \bp (\ms U_+), \bp(\ms U_-) \big].
\end{gather*}
This is exactly the property considered in~\cite{Severa:2017kcs}. Note, however, that our linear map $p$ in~\eqref{Severa map} already takes values in $\g^\CC$-valued rational functions, rather than just $\g^\CC$ itself as in~\cite{Severa:2017kcs}. Thus our map~\eqref{Severa map} plays the role of the spectral parameter dependent map $p_\lambda \colon \da \to \g^\CC$ from~\cite{Severa:2017kcs}, where $\lambda$ there is the spectral parameter.
\end{Remark}

We can now give an alternative derivation of~\eqref{L flat} from~\eqref{J flat} by using Lemma~\ref{lem: p morph}. Specifically, writing $\J \in \Omega^1(\Sigma, \da)$ in components as $\J = \J_\sigma {\rm d}\sigma + \J_\tau {\rm d}\tau$ we note that $\frac 12 [\J, \J] = [\J_\sigma, \J_\tau] {\rm d}\sigma \wedge {\rm d}\tau = [\J_\sigma, \Ec \J_\sigma] {\rm d}\sigma \wedge {\rm d}\tau$, where in the last step we used the condition~\eqref{Eq:Lz 2}. It~follows that
\begin{gather}
\frac 12 \bp ([\J, \J]) = \bp([\J_\sigma, \Ec \J_\sigma]) {\rm d}\sigma \wedge {\rm d}\tau = [\bp \J_\sigma, \bp \Ec \J_\sigma] {\rm d}\sigma \wedge {\rm d}\tau
= [\bp \J_\sigma, \bp \J_\tau] {\rm d}\sigma \wedge {\rm d}\tau \notag
\\ \hphantom{\frac 12 \bp ([\J, \J])}
{}= \frac 12 [\bp \J, \bp \J],\label{pJJ}
\end{gather}
where in the second equality we used Lemma~\ref{lem: p morph}. Applying the linear map~\eqref{Severa map} to~\eqref{J flat} we thus obtain
\begin{gather*}
{\rm d}(\bp\J) + \frac 12 [\bp \J, \bp \J] = 0,
\end{gather*}
where in the first term we used the linearity of $\bp$ and in the second term we used~\eqref{pJJ}. This is equivalent to~\eqref{L flat} by definition of $\J$.
The above derivation of the flatness equation for $\bp \J$ from that of $\J$ is analogous to~\cite[Proposition 1]{Severa:2017kcs}.

\subsection{Energy-momentum tensor} \label{sec: EM tensor}

In Section~\ref{sec: EM tensor E-model} we derived expressions for the components of the energy-momentum tensor of the $\Ec$-model in terms of the $\da$-valued field $\J_\sigma \in C^\infty(\Sigma, \da)$, the linear operator $\Ec\colon \da \to \da$ and the bilinear form $\langle\!\langle \cdot, \cdot \rangle\!\rangle_\da$ on $\da$, see Proposition~\ref{prop: EM tensor}. Having identified $\J_\sigma$ with the image under $\jb_{\bz'}$ of~the Lax matrix $\L_\sigma$ in Section~\ref{sec: recover E-model}, we may re-express the components of the energy-momentum tensor of the integrable $\Ec$-models we have constructed in terms of the Lax matrix $\L_\sigma$ itself.

Specifically, we may rewrite the expressions in Proposition~\ref{prop: EM tensor} as
\begin{subequations} \label{EM tensor E-model}
\begin{gather}
\label{T01}
T^\tau_{\;\;\, \sigma} = \frac 12 \langle\!\langle \jb_{\bz'} \L_\sigma, \jb_{\bz'} \L_\sigma \rangle\!\rangle_\da
= \frac 12 \langle\!\langle \bm \pi_{\bzeta} \L_\sigma, \bm \pi_{\bzeta} \L_\sigma \rangle\!\rangle_{\gzeta},\\
\label{Hamiltonian density}
T^\tau_{\;\;\, \tau} = - T^\sigma_{\;\;\, \sigma} = \frac 12 \langle\!\langle \jb_{\bz'} \L_\sigma, \Ec \jb_{\bz'} \L_\sigma \rangle\!\rangle_\da
= \frac 12 \langle\!\langle \bm \pi_{\bzeta} \L_\sigma, \widetilde{\Ec} \bm \pi_{\bzeta} \L_\sigma \rangle\!\rangle_{\gzeta},\\
\label{T10}
T^\sigma_{\;\;\, \tau} = - \frac 12 \langle\!\langle \jb_{\bz'} \L_\sigma, \Ec^2 \jb_{\bz'} \L_\sigma \rangle\!\rangle_\da
= - \frac 12 \langle\!\langle \bm \pi_{\bzeta} \L_\sigma, \widetilde{\Ec}^2 \bm \pi_{\bzeta} \L_\sigma \rangle\!\rangle_{\gzeta},
\end{gather}
\end{subequations}
where in each case we used Lemmas~\ref{lem: bilinear forms 1} and~\ref{lem: bilinear forms 2} in the last equality.

The expressions~\eqref{EM tensor E-model} can be directly compared with those in \cite[Proposition 2.4]{Delduc:2019bcl} for the energy-momentum tensor of an affine Gaudin model which were derived in the case when $\omega$ has only simple zeroes. Indeed, in the present notation, the expressions in \cite[Proposition 2.4]{Delduc:2019bcl} read
\begin{gather} \label{T00 T01 T10}
T^\tau_{\;\;\, \sigma} = \sum_{y \in \bm \zeta} q_y, \qquad
T^\tau_{\;\;\, \tau} = - T^\sigma_{\;\;\, \sigma} = \sum_{y \in \bm \zeta} \epsilon_y q_y, \qquad
T^\sigma_{\;\;\, \tau} = - \sum_{y \in \bm \zeta} \epsilon_y^2 q_y,
\end{gather}
where
$q_y \coloneqq - \frac 12 \varphi'(y) \big\langle \L_\sigma^{\uz{y,0}}, \L_\sigma^{\uz{y,0}} \big\rangle$ for each simple zero $y \in \bm \zeta$.
Using the definition~\eqref{bilinear g zeta} of the bilinear form $\langle\!\langle \cdot, \cdot \rangle\!\rangle_{\gzeta}$ on $\gzeta$ along with Remark~\ref{rem: CorE simple pole} about the simple zero case, and the definition of the operator $\widetilde{\Ec}$ in~\eqref{Etilde def}, we see that~\eqref{T00 T01 T10} coincides exactly with the expressions in~\eqref{EM tensor E-model}. In~particular, the relativistic invariance of the affine Gaudin model was shown in~\cite{Delduc:2019bcl} to be ensured by $\epsilon_y^2=1$ for all $y \in \bm \zeta$. We~see that this coincides with the condition $\Ec^2=\id$ for the relativistic invariance of the $\Ec$-model, see Remark~\ref{rem: rel inv E-model}.

In Remark~\ref{rem: CorE simple pole}, we also gave a simple condition for the operator $\Ec$ to be positive in the case when $\omega$ has simple real poles and $\g$ is compact, namely that $\varphi'(y)$ and $\epsilon_y$ should have opposite signs for every $y \in \bzeta$. This corresponds to the condition given in \cite[Section~2.2.3]{Delduc:2019bcl} for the Hamiltonian $\int_\RR {\rm d}\sigma \; T^\tau_{\;\; \tau}$ to be positive. See also Section~\ref{sec: EM tensor E-model}.

\subsection{Symmetries of the model}

\subsubsection[Global $G diag$-symmetry]{Global $\boldsymbol{G^{\diag}}$-symmetry} 

In this section we show that the $\Ec$-model action~\eqref{Eq:Action} for the edge mode $l \in C^\infty(\Sigma, D)$ has a~\emph{global} diagonal $G$-symmetry.

Let $\Delta \colon G \to G^{\times |\bz'|} \subset D$ denote the diagonal embedding of $G$ into $D$. For~any $g_0 \in G$ and $\ms U = (\ms u_{x, p} \otimes \varepsilon_x^p )_{\up{x, p} \in \bsb{\bz'}} \in \da$, the adjoint action of $\Delta(g_0) \in D$ on $\ms U$ reads
\begin{gather*}
\Ad_{\Delta(g_0)} \ms U = \big( (\Ad_{g_0} \ms u_{x, p}) \otimes \varepsilon_x^p \big)_{\up{x,p} \in \bsb{\bz'}}.
\end{gather*}
Since $g_0 \in G$ we have $\tau \Ad_{g_0} = \Ad_{g_0} \tau$ and so it follows from the explicit form~\eqref{map E explicit} of the linear operator $\Ec \colon \da \to \da$ defined in~\eqref{E def} that
\begin{gather} \label{E com Adg}
\Ec \Ad_{\Delta(g_0)} = \Ad_{\Delta(g_0)} \Ec.
\end{gather}

\begin{Proposition} \label{prop: Gdiag sym}
The action~\eqref{Eq:Action} is invariant under $l \mapsto l \Delta(g_0)$ for any $g_0 \in G$.
\begin{proof}
By construction, the kernel and image of the projector $\Pc_{l\,\Delta(g_0)}$ are given by
\begin{gather*}
\ker \Pc_{l\,\Delta(g_0)} = \Ad^{-1}_{l\,\Delta(g_0)} \k = \Ad_{\Delta(g_0)}^{-1} \Ad_l^{-1} \k = \Ad_{\Delta(g_0)}^{-1} \ker \Pc_l,\\
\im \Pc_{l\,\Delta(g_0)} = \Ec \Ad^{-1}_{l\,\Delta(g_0)} \k = \Ec \Ad_{\Delta(g_0)}^{-1} \Ad_l^{-1} \k = \Ad_{\Delta(g_0)}^{-1} \Ec \Ad_l^{-1} \k = \Ad_{\Delta(g_0)}^{-1} \im \Pc_l,
\end{gather*}
where in the second line we have used~\eqref{E com Adg} in the third step. It~is a standard result on~pro\-jec\-tors that $\Pc_{l\,\Delta(g_0)}$ is thus given by
\begin{gather*}
\Pc_{l\,\Delta(g_0)} = \Ad_{\Delta(g_0)}^{-1} \Pc_l \Ad_{\Delta(g_0)}.
\end{gather*}
Similar equalities also hold for $\Pb_{l\,\Delta(g_0)}$ and the transpose of $\Pc_{l\,\Delta(g_0)}$ and $\Pb_{l\,\Delta(g_0)}$.

Moreover, under $l \mapsto l \Delta(g_0)$, the Maurer--Cartan current $l^{-1}\partial_\mu l$ transforms as $l^{-1}\partial_\mu l \mapsto \Ad_{\Delta(g_0)}^{-1} l^{-1}\partial_\mu l$. Putting all of the above together it now follows that the first term in the action~\eqref{Eq:Action} is invariant under $l \mapsto l \Delta(g_0)$.

Finally, since the WZ-term in~\eqref{Eq:Action} is independent of the choice of extension $\widehat{l}^{-1} {\rm d}\widehat{l}$ of the $1$-form $l^{-1} {\rm d}l \in \Omega^1(\Sigma, \da)$ to the bulk $\Sigma \times I$, we can choose this extension for the transformed $1$-form $\Ad_{\Delta(g_0)}^{-1} l^{-1} {\rm d}l$ to be $\Ad_{\Delta(g_0)}^{-1} \widehat{l}^{-1} {\rm d}\widehat{l}$, from which it follows that the WZ-term is also invariant under the transformation $l \mapsto l \Delta(g_0)$.
\end{proof}
\end{Proposition}

\begin{Remark}
The only property of the element $\Delta(g_0) \in D$ which we used in the proof of~Pro\-position~\ref{prop: Gdiag sym} is~\eqref{E com Adg}. Therefore, the statement of the proposition would also hold for any other Lie subgroup of $D$ with the property that all its elements $d \in D$ are such that $\Ad_d \colon \da \to \da$ commutes with $\Ec$.
\end{Remark}

\subsubsection [Global symmetries for k an ideal]
{Global symmetries for $\boldsymbol{\k}$ an ideal}\label{sec: PCM sym}

We now consider the $\Ec$-model action~\eqref{Eq:Action} in a more specific setup.

\begin{Proposition} \label{prop: PCM sym}
If the Lagrangian subalgebra $\k \subset \da$ is an ideal, then the $\Ec$-model action~\eqref{Eq:Action} is invariant under $l \mapsto al$ for any $a \in D$.
\end{Proposition}
\begin{Remark}
In fact, since the $\Ec$-model action has a gauge invariance $l \mapsto kl$ for $k \in C^\infty(\Sigma, K)$, see Section~\ref{sec: gauge inv E-model}, the global symmetry in Proposition~\ref{prop: PCM sym} is really only an additional symmetry by the Lie group $K \backslash D$.
\end{Remark}
\begin{proof}{Proposition~\ref{prop: PCM sym}}
The currents $l^{-1}\partial_\mu l$ and the WZ-term in~\eqref{Eq:Action} are both invariant under the transformation $l \mapsto al$. It~remains to check that $\Pc_{al} = \Pc_l$ and $\Pb_{al} = \Pb_l$. Since $\k$ is an ideal of $\da$ we have $\Ad_l^{-1} \k = \k$ and thus
\begin{gather*}
\im \Pc_l = \Ec^{-1} \k, \qquad \ker \Pc_l = \k.
\end{gather*}
We see that the projector $\Pc_l$, and also $\Pb_l$, is in fact independent of $l$. It~is therefore invariant under $l \mapsto a l$, as required.
\end{proof}

We can construct simple examples of Lagrangian ideals $\k \subset \da$ in the case when all the multiplicities $n_x$ of the poles $x \in \bz$ of $\omega$ are even, i.e., $n_x = 2r_x$ for some $r_x \in \ZZ_{\geq 1}$. Recall the definition~\eqref{defect grp alg inf} of the Lie algebra $\da$, namely
\begin{gather*}
\da = \bigoplus_{x \in \bm z'_{\rm r}} \g \otimes_\RR \T^{2r_x}_x \oplus \bigoplus_{x \in \bm z_{\rm c}} \big(\g^\CC \otimes_\CC \T^{2r_x}_x\big)_\RR.
\end{gather*}
Consider the ideals $\T^{\geq r_x}_x \coloneqq \varepsilon_x^{r_x} \RR[\varepsilon_x]/ \big(\varepsilon_x^{2r_x}\big) \subset \T^{2 r_x}_x$ for real finite poles $x \in \bz'_{\rm r}$ and $\T^{\geq r_x}_x \coloneqq \varepsilon_x^{r_x} \CC[\varepsilon_x]/ \big(\varepsilon_x^{2r_x}\big) \subset \T^{2 r_x}_x$ for complex poles $x \in \bz_{\rm c}$. It~is easy to check that
\begin{gather*}
\k \coloneqq \bigoplus_{x \in \bm z'_{\rm r}} \g \otimes_\RR \T^{\geq r_x}_x \oplus \bigoplus_{x \in \bm z_{\rm c}} \big(\g^\CC \otimes_\CC \T^{\geq r_x}_x\big)_\RR
\end{gather*}
is a Lagrangian ideal of $\da$.

\section{Examples} \label{sec: examples}

In this section we give a few examples of the above general construction, including details of the one mentioned in the introduction. In~each case, we make a choice of meromorphic $1$-form~$\omega$, Lagrangian subalgebra $\k \subset \da$ and set of parameters $\epsilon_y$ associated with each zero $y \in \bzeta$ of $\omega$. This is then fed into the general construction to produce the action and Lax connection of the corresponding integrable $\sigma$-models.

\subsection[Principal chiral model and non-abelian T-dual]
{Principal chiral model and non-abelian $\boldsymbol T$-dual}\label{sec: PCM+Tdual}

We rederive the well-known actions of the principal chiral model and its non-abelian $T$-dual as an application of our general construction. The principal chiral model was already derived in \cite[Section~10.2]{Costello:2019tri} and then again in \cite[Section~5.1]{Delduc:2019whp} using the unifying $2$d action valid when $\omega$ has at most double poles. We~derive it again in Section~\ref{sec: PCM} below since it provides the simplest illustration of our construction.

However, the derivation of the non-abelian $T$-dual in Section~\ref{sec: Tdual} below from $4$d Chern--Simons theory is new. This was conjectured but could not be derived in~\cite{Delduc:2019whp} since the formula for the unifying $2$d action there was only applicable under a technical condition on the Lagrangian subalgebra $\k \subset \da$, see \cite[equation~(4.16)]{Delduc:2019whp}. The latter is not satisfied by the choice of Lagrangian subalgebra used in Section~\ref{sec: Tdual} to derive the non-abelian $T$-dual.

We begin by setting up the formalism to discuss both the principal chiral model and its non-abelian $T$-dual. We~let $a > 0$ and consider the $1$-form
\begin{gather} \label{PCM 1form}
\omega = \frac{a^2 {\rm d}z}{z^2} - {\rm d}z.
\end{gather}
In the notation of Section~\ref{sec: defects} we have the set of poles $\bz = \{ 0, \infty \}$ and the set of zeroes is $\bzeta = \{ a, -a \}$. Both poles are double poles so that $n_0 = n_\infty= 2$ and both zeroes are simple so~that $m_a = m_{-a} = 1$. Moreover, all the zeroes and poles are real so here $\bz = \bz_{\rm r}$ and $\bzeta = \bzeta_{\rm r}$. The~levels are read off from $\omega$ to be $\ell^0_0 = 0$, $\ell^0_1 = a^2$ and $\ell^\infty_1 = 1$.

We also choose the parameters associated to the set of zeroes in~\eqref{Eq:Lz} to be
\begin{gather} \label{PCM epsilon choice}
\epsilon_{\pm a} = \pm 1.
\end{gather}
Below we shall construct all the data associated with the choice of $1$-form~\eqref{PCM 1form} and para\-me\-ters~\eqref{PCM epsilon choice}, and then use this data to build the actions for the principal chiral model in Section~\ref{sec: PCM} and its non-abelian $T$-dual in Section~\ref{sec: Tdual}.

\subsubsection[Lie groups D, K and tilde K]{Lie groups $\boldsymbol D$, $\boldsymbol K$ and $\boldsymbol{\tilde K}$}

The defect Lie algebra~\eqref{defect grp alg inf} is given here by
\begin{gather*}
\da = \g \otimes_\RR \RR[\varepsilon_0] / \big(\varepsilon_0^2\big) = \g \ltimes \g_{\rm ab},
\end{gather*}
where $\g_{\rm ab} \coloneqq \g \otimes_\RR \varepsilon_0 \RR[\varepsilon_0] / \big(\varepsilon_0^2\big)$ is isomorphic to the vector space $\g$ equipped with the trivial Lie bracket and the adjoint action of $\g$. By using the abbreviated notation $\ms u^p \coloneqq \ms u^{[0, p]}= \ms u \otimes \varepsilon^p_0$ for any $\ms u \in \g$ and $p \in \{ 0,1 \}$, the Lie algebra relations in $\da$ read
\begin{gather} \label{relations Tdual}
\big[\ms u^0, \ms v^0\big] = [\ms u, \ms v]^0, \qquad
\big[\ms u^0, \ms v^1\big] = \big[\ms u^1, \ms v^0\big] = [\ms u, \ms v]^1, \qquad
\big[\ms u^1, \ms v^1\big] = 0,
\end{gather}
for any $\ms u, \ms v \in \g$.

The associated Lie group $D$ is the tangent bundle $TG$ of $G$, which in the right trivialisation is isomorphic to the Lie group $G \ltimes \g$ with product and inverse
\begin{gather*}
(g, \ms u) (h, \ms v) = \big( gh, \ms u + \Ad_g \ms v \big),\qquad
(g, \ms u)^{-1} = \big( g^{-1}, - \Ad_g^{-1} \ms u \big)
\end{gather*}
for every $g, h \in G$ and $\ms u, \ms v \in \g$.

We have the two obvious Lie subalgebras
\begin{gather} \label{k tilde k PCM}
\tilde \k \coloneqq \g \oplus \{ 0 \} = \big\{ \ms u^0 \mid \ms u \in \g \big\}, \qquad
\k \coloneqq \{ 0 \} \oplus \g_{\rm ab} = \big\{ \ms u^1 \mid \ms u \in \g \big\}
\end{gather}
of $\da$. These are complementary since we have the direct sum decomposition
\begin{gather} \label{da k tilde k}
\da = \k \dotplus \tilde{\k}.
\end{gather}
Let $\tilde K, K \subset D$ denote the corresponding connected Lie subgroups of $D$, which are isomorphic to $G \times \{ 0 \}$ and $\{ \id \} \times \g$, respectively. In~particular, it is clear that $K$ is normal in $D$. We~have the global decomposition $D = K \tilde K = \tilde K K$.

The bilinear form~\eqref{form on gz} on $\da$ is given explicitly by
\begin{gather} \label{bilinear Tdual}
\langle\!\langle \ms u^0, \ms v^0 \rangle\!\rangle_{\da} = \langle\!\langle \ms u^1, \ms v^1 \rangle\!\rangle_{\da} = 0, \qquad
\langle\!\langle \ms u^0, \ms v^1 \rangle\!\rangle_{\da} = \langle\!\langle \ms u^1, \ms v^0 \rangle\!\rangle_{\da} = a^2 \langle \ms u, \ms v \rangle,
\end{gather}
for any $\ms u, \ms v \in \g$ so that $\k$ and $\tilde{\k}$ are both Lagrangian subalgebras of $\da$.

\subsubsection[Linear operator E]{Linear operator $\boldsymbol{\Ec}$}

The real vector space~\eqref{gzeta def} associated with the two simple zeroes of $\omega$ is given here by $\gzeta = \g \oplus \g$. One checks by computing the Cauchy matrix~\eqref{Cauchy mat} that the isomorphism $\C \colon \gzeta \to \da$ from Lemma~\ref{lem: Cauchy} (see also~\eqref{map C explicit} and~\eqref{map C inv explicit}) and its inverse are given here by
\begin{gather*}
\C(\ms u, \ms v) = - \frac{\ms u^0}{a} + \frac{\ms v^0}{a} - \frac{\ms u^1}{a^2} - \frac{\ms v^1}{a^2}, \qquad
\C^{-1}\big( \ms u^0 + \ms v^1\big) = \bigg({-}\frac 12 a \ms u - \frac 12 a^2 \ms v, \frac 12 a \ms u - \frac 12 a^2 \ms v \bigg)
\end{gather*}
for every $\ms u, \ms v \in \g$. Given the choice of parameters $\epsilon_y$ for $y \in \bzeta$ we made in~\eqref{PCM epsilon choice}, the linear isomorphism~\eqref{Etilde def} then reads $\widetilde{\Ec}(\ms u, \ms v) = (\ms u, - \ms v)$ and the linear operator~\eqref{E def} is found to act in the present case as
\begin{gather} \label{E Tdual}
\Ec\big(\ms u^0 + \ms v^1\big) = a \ms v^0 + \frac{\ms u^1}a
\end{gather}
for every $\ms u, \ms v \in \g$. We~have $\Ec^2 = \id$, corresponding to the fact that the principal chiral model and its non-abelian $T$-dual are relativistic.

The Lie algebra relations~\eqref{relations Tdual}, bilinear form~\eqref{bilinear Tdual} and linear operator~\eqref{E Tdual} agree with those for the principal chiral model given in \cite[equations~(19)--(21)]{Klimcik:2015gba}.

\subsubsection{Principal chiral model} \label{sec: PCM}

Here we apply the general construction of Section~\ref{sec: E-model from 4d} with $\k$ defined in~\eqref{k tilde k PCM} as the Lagrangian subalgebra.
Since we have the global factorisation $D = K \tilde K$ we can factorise the field $l \in C^\infty(\Sigma, D)$ uniquely as $l = kg$ for some $k \in C^\infty(\Sigma, K)$ and $g \in C^\infty\big(\Sigma, \tilde K\big)$. Using the gauge invariance of the action~\eqref{Eq:ActionE} under $l \mapsto k^{-1} l$ from Section~\ref{sec: gauge inv E-model}, we can then fix $l = g$. In~particular, the action~\eqref{Eq:ActionE} now reads
\begin{gather}
S_{\rm 2d}(g) = \frac{1}{2} \int_\Sigma \big( \langle\!\langle g^{-1} \partial_\tau g, \Ec \Pc_g \big(g^{-1} \partial_\tau g\big) \rangle\!\rangle_\da - \langle\!\langle g^{-1} \partial_\sigma g, \Ec \Pc_g\big(g^{-1} \partial_\sigma g\big) \rangle\!\rangle_\da \notag
\\ \hphantom{S_{\rm 2d}(g) =\frac{1}{2} \int_\Sigma \big(}
{} + \langle\!\langle g^{-1} \partial_\tau g, \big(\Pc_g - \null^t \Pc_g\big)\big(g^{-1} \partial_\sigma g\big) \rangle\!\rangle_\da \big) {\rm d}\sigma \wedge {\rm d}\tau - \frac 12 I^{\rm WZ}_{\tilde \k}[g]
\label{action PCM}
\end{gather}
for the group valued field $g \in C^\infty(\Sigma, \tilde K)$, where we have used the fact that $\Ec^2 = \id$ together with part $(vi)$ of Proposition~\ref{PropE}. Here $\Pc_g$ denotes the projector defined by the relations~\eqref{Eq:KerIm}, namely
\begin{gather*}
\ker \Pc_g = \Ad_g^{-1} \k, \qquad \im \Pc_g = \Ec \Ad_g^{-1} \k.
\end{gather*}
Since $[\tilde \k, \k] \subset \k$ we have $\Ad_g^{-1} \k = \k$ and hence
\begin{gather*}
\ker \Pc_g = \k, \qquad \im \Pc_g = \Ec \k = \tilde \k,
\end{gather*}
where the last equality uses the explicit forms~\eqref{k tilde k PCM} and~\eqref{E Tdual} of the two subalgebras $\k, \tilde \k \subset \da$ and of~$\Ec$.
Thus $\Pc_g$ is simply the projection onto $\tilde \k$ along $\k$, relative to the direct sum decomposition~\eqref{da k tilde k}. In~particular, it acts as the identity on $g^{-1} \partial_\mu g \in C^\infty(\Sigma, \tilde \k)$. Moreover, since $\tilde \k$ is isotropic with respect to the bilinear form~\eqref{bilinear Tdual}, the WZ-term in the action~\eqref{action PCM} vanishes.

Putting all of the above together and noting the identity $\langle\!\langle \ms u^0, \Ec \ms v^0 \rangle\!\rangle_{\da} = a \langle \ms u, \ms v \rangle$ for any $\ms u, \ms v \in \g$, the action~\eqref{action PCM} reduces to the usual principal chiral model action
\begin{gather*}
S_{\rm 2d}(g) = \frac{1}{2} a \int_\Sigma \langle g^{-1} \partial_+ g, g^{-1} \partial_- g \rangle {\rm d}\sigma \wedge {\rm d}\tau,
\end{gather*}
where $\partial_\pm = \partial_\tau \pm \partial_\sigma$, for the Lie group valued field $g \in C^\infty\big(\Sigma, \tilde K\big) \cong C^\infty(\Sigma, G)$. The global $G^{\rm diag}$-symmetry from Proposition~\ref{prop: Gdiag sym} corresponds here to the right $G$-symmetry of the principal chiral model. Since the Lagrangian subalgebra $\k \subset \da$ is an ideal we also have the global left symmetry of~Proposition~\ref{prop: PCM sym} by the Lie group $K \backslash D \simeq \tilde K \simeq G$, corresponding to the left $G$-symmetry of the principal chiral model.

Let $j \coloneqq g^{-1} {\rm d}g$. Using $\ast {\rm d}\sigma = - {\rm d}\tau$ and $\ast {\rm d}\tau = - {\rm d}\sigma$, the Lax connection~\eqref{Lax for E-model} is given by~$\L = \bp \big( j^0 - \frac 1a \ast j^1 \big)$, where the map $\bp$ in~\eqref{Severa map} is the inverse of the map $\jb_{\bz'}$ defined in~\eqref{j z def} and given explicitly in the present case by
\begin{gather*}
\jb_{\bm z'} \colon\ \frac{\ms u_a}{z-a} + \frac{\ms u_{-a}}{z+a} \longmapsto \frac{1}{a} (\ms u_{-a} - \ms u_a)^0 + \frac{1}{a^2} (-\ms u_{-a} - \ms u_a)^1.
\end{gather*}
Its inverse is then given explicitly by
\begin{gather} \label{p for PCM-Tdual}
\bp \colon\ \ms v^0 + \ms w^1 \longmapsto \frac{a \ms v + a^2 \ms w}{2(a - z)} + \frac{a \ms v - a^2 \ms w}{2(a+z)}.
\end{gather}
We therefore obtain the Lax connection of the principal chiral model
\begin{gather} \label{Lax PCM}
\L = \bp \big( j^0 - \ast j^1 \big) = a \frac{a j - z \ast j}{a^2-z^2} = \frac{a j_+}{a - z} {\rm d}\sigma^+ + \frac{a j_-}{a + z} {\rm d}\sigma^-,
\end{gather}
where $j_\pm \coloneqq g^{-1} \partial_\pm g$.
This coincides with the usual Lax connection of the principal chiral model after rescaling the spectral parameter as $z \mapsto a z$.

\subsubsection[Non-abelian T-dual]{Non-abelian $\boldsymbol{T}$-dual} \label{sec: Tdual}

We will now use the reverse factorisation $D = \tilde K K$, treating $\tilde K \subset D$ as the Lie subgroup which we quotient by in Section~\ref{sec: Emodel action}. As in Section~\ref{sec: PCM}, we can factorise our field $l \in C^\infty(\Sigma, D)$ uniquely as $l = \tilde k p$ for some $\tilde k \in C^\infty\big(\Sigma, \tilde K\big)$ and $p \in C^\infty(\Sigma, K)$. We~may therefore use the gauge invariance by the subgroup $\tilde K$ from Section~\ref{sec: gauge inv E-model} to fix $l = p$, obtaining the action~\eqref{Eq:ActionE} for $p \in C^\infty(\Sigma, K)$, where $\tilde \k$ now plays the role of $\k$. As in Section~\ref{sec: PCM}, the WZ-term in this action vanishes since the Lie subalgebra $\k$ is isotropic. Furthermore, by definition we can write $p = \big(\id, \frac{1}{a} \ms m\big)$ for some $\ms m \in C^\infty(\Sigma, \g)$, where the factor of $\frac{1}{a}$ is introduced for later convenience. In~particular, we then have $p^{-1} \partial_\mu p = \frac{1}{a} \partial_\mu \ms m^1 \in C^\infty(\Sigma, \k)$. In~the present case, the action~\eqref{Eq:ActionE} therefore simplifies to
\begin{gather}
S_{\rm 2d}(\ms m) = \frac{1}{2} a^{-2} \int_\Sigma \big( \langle\!\langle \partial_\tau \ms m^1, \Ec \tilde \Pc_p \big(\partial_\tau \ms m^1\big) \rangle\!\rangle_\da - \langle\!\langle \partial_\sigma \ms m^1, \Ec \tilde \Pc_p\big(\partial_\sigma \ms m^1\big) \rangle\!\rangle_\da \notag
\\ \hphantom{S_{\rm 2d}(\ms m) = \frac{1}{2} a^{-2} \int_\Sigma \big(}
{}+ \langle\!\langle \partial_\tau \ms m^1, \big(\tilde \Pc_p - \null^t \tilde \Pc_p\big)\big(\partial_\sigma \ms m^1\big) \rangle\!\rangle_\da \big) {\rm d}\sigma \wedge {\rm d}\tau,
\label{action Tdual}
\end{gather}
where now $\tilde \Pc_p$ denotes the projector defined in the same way as in~\eqref{Eq:KerIm} but relative to the Lagrangian subalgebra $\tilde \k$, namely
\begin{gather*}
\ker \tilde \Pc_p = \Ad_p^{-1} \tilde \k =\bigg(\id - \frac{1}{a} \ad_{\ms m^1}\bigg) \tilde \k = \bigg\{ \ms u^0 - \frac{1}{a} [\ms m, \ms u]^1 \,\bigg|\, \ms u \in \g \bigg\},\\
\im \tilde \Pc_p = \Ec \Ad_p^{-1} \tilde \k = \Ec\bigg( \bigg(\id - \frac{1}{a} \ad_{\ms m^1}\bigg) \tilde \k \bigg) = \bigg\{ \frac{1}{a} \ms u^1 - [\ms m, \ms u]^0 \,\bigg|\, \ms u \in \g \bigg\}.
\end{gather*}
It is straightforward to check that the projector $\tilde \Pc_p \colon \da \to \da$ with the above kernel and image is given by
\begin{gather*}
\tilde \Pc_p \big(\ms u^0 + \ms v^1\big) = - \bigg( \frac{\ad^2_{\ms m}}{\id - \ad_{\ms m}^2} \ms u \bigg)^0 - \bigg( \frac{a \ad_{\ms m}}{\id - \ad_{\ms m}^2} \ms v \bigg)^0 + \bigg( \frac{a^{-1} \ad_{\ms m}}{\id - \ad_{\ms m}^2} \ms u \bigg)^1 + \bigg( \frac{1}{\id - \ad_{\ms m}^2} \ms v \bigg)^1
\end{gather*}
for any $\ms u, \ms v \in \g$. In~particular, applying this to $\partial_\mu \ms m^1$ we find
\begin{gather} \label{Pg for Tdual}
\tilde \Pc_p \big( \partial_\mu \ms m^1 \big) = \bigg({-} \frac{a \ad_{\ms m}}{\id - \ad^2_{\ms m}} \partial_\mu \ms m \bigg)^0 + \bigg( \frac{1}{\id - \ad^2_{\ms m}} \partial_\mu \ms m \bigg)^1.
\end{gather}
Substituting this into the action~\eqref{action Tdual} and using the expressions~\eqref{bilinear Tdual} and~\eqref{E Tdual} of the bilinear form $\langle\!\langle \cdot, \cdot \rangle\!\rangle_\da$ and the linear operator $\Ec$ we arrive at the standard action of the non-abelian $T$-dual of the principal chiral model~\cite{Fradkin:1984ai, Fridling:1983ha}, namely
\begin{gather*}
S_{\rm 2d}(\ms m) = \frac{1}{2} a \int_\Sigma \bigg\langle \partial_+ \ms m, \frac{1}{\id - \ad_{\ms m}} \partial_- \ms m \bigg\rangle {\rm d}\sigma \wedge {\rm d}\tau.
\end{gather*}

The Lax connection~\eqref{Lax for E-model} is now given by
\begin{gather*}
\L = \bp \bigg( \frac{1}{a} \tilde \Pc_p\big({\rm d}\ms m^1\big) - \frac{1}{a} \Ec \tilde \Pc_p\big({\ast} {\rm d}\ms m^1\big) \bigg).
\end{gather*}
Using the explicit form of the inverse $\bp$ in~\eqref{p for PCM-Tdual} and using~\eqref{Pg for Tdual} we find
\begin{gather*}
\L = \frac{a}{a-z} \frac{1}{\id + \ad_{\ms m}} \partial_+ \ms m \; {\rm d}\sigma^+ - \frac{a}{a+z} \frac{1}{\id - \ad_{\ms m}} \partial_- \ms m \; {\rm d}\sigma^-.
\end{gather*}
This becomes the usual Lax connection of the non-abelian $T$-dual of the principal chiral model after the rescaling $z \mapsto a z$.

\subsection{Fourth order pole}

In this section we give an example of our construction in the case when $\omega$ has a pole of order $4$. We~will let $a > b > 0$ and take
\begin{gather} \label{Fourth 1form}
\omega = \frac{\big(z^2 - a^2\big) \big(b^2 - z^2\big)}{z^4} {\rm d}z.
\end{gather}
The set of poles is $\bz = \bz_{\rm r} = \{ 0, \infty \}$ with orders $n_0 = 4$ and $n_\infty= 2$. The associated levels are $\ell^0_3 = - a^2 b^2$, $\ell^0_2 = 0$, $\ell^0_1 = a^2 + b^2$, $\ell^0_0 = 0$ and $\ell^\infty_1 = 1$. The set of zeroes is $\bzeta = \bzeta_{\rm r} = \{ a, -a, b, -b \}$ with all zeroes being simple.

We let the parameters in~\eqref{Eq:Lz} associated with the set $\bzeta$ of zeroes of $\omega$ be
\begin{gather} \label{Fourth epsilon choice}
\epsilon_{\pm a} = \pm 1, \qquad \epsilon_{\pm b} = \mp 1.
\end{gather}
This choice ensures that $\Ec$ will be positive (for compact $\g$) and such that $\Ec^2 = \id$. Indeed, the latter condition follows since $\epsilon_y^2 = 1$ for all $y \in \bzeta$ and positivity follows from Remark~\ref{rem: CorE simple pole} after noting that $- \epsilon_y \varphi'(y) > 0$ for each $y \in \bzeta$. In~what follows we shall construct all the necessary data associated with the choice of $1$-form~\eqref{Fourth 1form} and parameters~\eqref{Fourth epsilon choice}. We~then extract from this data the action and Lax connection of a new $2$d integrable field theory.

\subsubsection[Lie groups D and K]{Lie groups $\boldsymbol{D}$ and $\boldsymbol{K}$}

Since the pole $0$ is real, the defect Lie algebra~\eqref{defect grp alg inf} is
\begin{gather*}
\da = \g \otimes_\RR \RR[\varepsilon_0] / \big(\varepsilon_0^4\big).
\end{gather*}
As in Section~\ref{sec: PCM+Tdual} we use the abbreviated notation $\ms u^p \coloneqq \ms u^{[0, p]}= \ms u \otimes \varepsilon^p_0$ for any $\ms u \in \g$ and $p \in \{ 0,1,2,3 \}$. The Lie algebra relations in $\da$ read{\samepage
\begin{gather*} 
\big[\ms u^p, \ms v^q\big] = [\ms u, \ms v]^{p+q},
\end{gather*}
for any $\ms u, \ms v \in \g$ and $p, q \in \{ 0,1,2,3 \}$. In~particular, recall this vanishes for $p+q \geq 4$.}

The Lie group $D$ is given by the $3^{\rm rd}$ order jet bundle $J^3 G$ of the Lie group $G$. In~the right trivialisation it is isomorphic to $G \times \g \times \g \times \g$ equipped with the Lie group product and inverse~\cite{Vizman}
\begin{gather*}
(g, \ms u, \ms v, \ms w) (h, \ms x, \ms y, \ms z) = \big( gh, \ms u + \Ad_g \ms x, \ms v + \Ad_g \ms y + [\ms u, \Ad_g \ms x],
\\ \hphantom{(g, \ms u, \ms v, \ms w) (h, \ms x, \ms y, \ms z) = \big(}
{} \ms w + \Ad_g \ms z + 2 [\ms u, \Ad_g \ms y] + [\ms v, \Ad_g \ms x] + [\ms u, [\ms u, \Ad_g \ms x] \big),
\\
(g, \ms u, \ms v, \ms w)^{-1} = \big( g^{-1}, - \Ad_g^{-1} \ms u, - \Ad_g^{-1} \ms v, - \Ad_g^{-1} \ms w + \Ad_g^{-1} [\ms u, \ms v] \big)
\end{gather*}
for every $g, h \in G$ and $\ms u, \ms v, \ms w, \ms x, \ms y, \ms z \in \g$.

We consider the ideal $\k \subset \da$ defined in Section~\ref{sec: PCM sym}, which is given here by
\begin{gather*}
\k = \g \otimes_\RR \varepsilon_0^2 \RR[\varepsilon_0] / \big(\varepsilon_0^4\big) = \big\{ \ms u^2 + \ms v^3 \mid \ms u, \ms v \in \g \big\}.
\end{gather*}
Let $K \subset D$ denote the corresponding connected Lie subgroup of $D$. It~is isomorphic to $\{ \id \} \times \{ 0 \} \times \g \times \g$ which is normal in $G \times \g \times \g \times \g$ since
\begin{gather*}
(g, \ms u, \ms v, \ms w) (\id, 0, \ms y, \ms z) (g, \ms u, \ms v, \ms w)^{-1}
= \big( \id, 0, \Ad_g \ms y, \Ad_g \ms z + 3 [\ms u, \Ad_g \ms y] \big),
\end{gather*}
for any $g \in G$ and $\ms u, \ms v, \ms w, \ms y, \ms z \in \g$. In~particular, the left coset $K \backslash D$ is naturally a Lie group which, as a manifold, is diffeomorphic to $G \times \g \times \{ 0 \} \times \{ 0 \}$.

The bilinear form~\eqref{form on gz} on $\da$ is given explicitly by
\begin{gather*} 
\langle\!\langle \ms u^p, \ms v^q \rangle\!\rangle_{\da} = \left\{
\begin{array}{ll}
- a^2 b^2 \langle \ms u, \ms v \rangle, & \text{if}\quad p+q = 3,\\
\big(a^2 +b^2\big) \langle \ms u, \ms v \rangle, & \text{if}\quad p+q = 1,\\
0, & \textup{otherwise}
\end{array}
\right.
\end{gather*}
for any $\ms u, \ms v \in \g$. In~particular, we see that $\k$ is indeed a Lagrangian subalgebra of $\da$.

\subsubsection[Linear operator E]{Linear operator $\boldsymbol{\Ec}$}

The real vector space~\eqref{gzeta def} associated with the zeroes of $\omega$ is given here by $\gzeta = \g^{\oplus 4}$.
With the choice of parameters $\epsilon_y$ for $y \in \bzeta$ in~\eqref{Fourth epsilon choice} we find by explicitly computing the Cauchy matrix~\eqref{Cauchy mat} that the linear operator $\Ec \colon \da \to \da$ in~\eqref{E def} is given by
\begin{gather}
\Ec\big(\ms u^0 + \ms v^1 + \ms w^2 + \ms x^3\big) = \frac{a^2-ab+b^2}{a-b} \ms v^0 - \frac{a^2 b^2}{a-b} \ms x^0 + \frac{1}{a-b} \ms u^1 - \frac{a b}{a-b} \ms w^1 \notag
\\ \hphantom{\Ec\big(\ms u^0 + \ms v^1 + \ms w^2 + \ms x^3\big) =}
{} + \frac{1}{a-b} \ms v^2 - \frac{ab}{a-b} \ms x^2 + \frac{1}{a b(a - b)} \ms u^3 - \frac{a^2-ab+b^2}{ab(a-b)} \ms w^3,
\label{E order 4}
\end{gather}
for every $\ms u, \ms v, \ms w, \ms x \in \g$. As previously noted after~\eqref{Fourth epsilon choice} we have $\Ec^2 = \id$.
Since $\k \subset \da$ is an ideal we have $\Ad_l^{-1} \k = \k$ for any $l \in D$ and therefore the projector $\Pc_l$ defined as in~\eqref{Eq:KerIm} has kernel and image
\begin{gather*}
\ker \Pc_l = \k = \big\{ \ms u^2 + \ms v^3 \mid \ms u, \ms v \in \g \big\},
\\
\im \Pc_l = \Ec \k = \big\{ a^3 b^3 \ms u^0 + a^2 b^2 \ms v^1 + a^2 b^2 \ms u^2 + \big(a^2-ab+b^2\big) \ms v^3 \mid \ms u, \ms v \in \g \big\},
\end{gather*}
where we used the explicit form of $\Ec$ in~\eqref{E order 4}. Note that $\Pc_l$ is therefore independent of $l$.
Explicitly, we have
\begin{gather}
\Pc_l\big(\ms u^0 + \ms v^1 + \ms w^2 + \ms x^3\big) = \ms u^0 + \ms v^1 + \frac 1{ab} \ms u^2 + \frac{a^2-ab+b^2}{a^2b^2} \ms v^3, \nonumber
\\
\Ec \Pc_l\big(\ms u^0 + \ms v^1 + \ms w^2 + \ms x^3\big) = \frac{b-a}{ab} \ms v^2 + \frac{b-a}{a^2b^2} \ms u^3,
\label{Fourth order Pl EPl}
\end{gather}
for any $\ms u, \ms v, \ms w, \ms x \in \g$.

\subsubsection{Action}

Let $p = (g, \ms u, 0,0) \in C^\infty(\Sigma, D)$ be a representative of a class in $K \backslash D$ in $D$. We~would like to explicitly determine the corresponding action~\eqref{Eq:ActionE}, namely
\begin{gather} \label{Fourth order E-model}
S_{\rm 2d}(p) = \frac{1}{2} \int_\Sigma \big( \langle\!\langle p^{-1} {\rm d}p, \Ec \Pc_p \big({\ast} p^{-1} {\rm d}p\big) \rangle\!\rangle_\da - \langle\!\langle p^{-1} {\rm d}p, \Pc_p \big(p^{-1} {\rm d}p\big) \rangle\!\rangle_\da \big) - \frac 12 I^{\rm WZ}_{\da}[p].
\end{gather}
First, we note that~\cite{Vizman}
\begin{gather*}
p^{-1} {\rm d}p = \big( g^{-1} {\rm d}g \big)^0 + \big(\Ad_g^{-1} {\rm d}\ms u\big)^1 - \frac 12 \big(\Ad_g^{-1} [\ms u, {\rm d}\ms u]\big)^2 + \frac 16 \big( \! \Ad_g^{-1} [\ms u, [\ms u,{\rm d}\ms u] ] \big)^3.
\end{gather*}
Applying the operators~\eqref{Fourth order Pl EPl} we then find
\begin{gather*}
\Pc_p \big(p^{-1} {\rm d}p\big) = \big(g^{-1} {\rm d}g\big)^0 + \big(\Ad_g^{-1} {\rm d}\ms u\big)^1 + \frac 1{ab} \big(g^{-1} {\rm d}g\big)^2 + \frac{a^2-ab+b^2}{a^2b^2} (\Ad_g^{-1} {\rm d}\ms u)^3,
\\
\Ec \Pc_p \big({\ast} p^{-1} {\rm d}p\big) = \frac{b-a}{ab} \big(\Ad_g^{-1} \ast {\rm d}\ms u\big)^2 + \frac{b-a}{a^2b^2} \big({\ast} g^{-1} {\rm d}g\big)^3.
\end{gather*}
The first two terms in the action~\eqref{Fourth order E-model} then take the form
\begin{gather*}
\langle\!\langle p^{-1} {\rm d}p, \Ec \Pc_p \big({\ast} p^{-1} {\rm d}p\big) \rangle\!\rangle_\da
= (a-b) \langle {\rm d}g g^{-1}, \ast {\rm d}g g^{-1} \rangle + ab(a-b) \langle {\rm d}\ms u, \ast {\rm d}\ms u \rangle,\\
\langle\!\langle p^{-1} {\rm d}p, \Pc_p \big(p^{-1} {\rm d}p\big) \rangle\!\rangle_\da
= - (a-b)^2 \langle {\rm d}gg^{-1}, {\rm d}\ms u \rangle + \frac{a^2b^2}2 \bigg\langle [\ms u, {\rm d}\ms u], {\rm d}\ms u + \frac 13 \big[\ms u, {\rm d}gg^{-1}\big] \bigg\rangle.
\end{gather*}
On the other hand, we find that the Wess--Zumino $3$-form is exact since
\begin{gather*}
\langle\!\langle \widehat p^{-1} {\rm d}\widehat p, \big[\widehat p^{-1} {\rm d}\widehat p, \widehat p^{-1} {\rm d}\widehat p\big] \rangle\!\rangle_\da
= {\rm d} \big( \big\langle {\rm d} \widehat g \widehat g^{-1}, 6\big(a^2+b^2\big) {\rm d} \widehat{\ms u} - a^2b^2 [\widehat{\ms u}, [\widehat{\ms u}, {\rm d}\widehat{\ms u}]] \big\rangle
\\ \hphantom{\langle\!\langle \widehat p^{-1} {\rm d}\widehat p, [\widehat p^{-1} {\rm d}\widehat p, \widehat p^{-1} {\rm d}\widehat p] \rangle\!\rangle_\da
= {\rm d} \big(}
{}- a^2b^2 \big\langle [\widehat{\ms u}, [{\rm d} \widehat{\ms u}, {\rm d}\widehat{\ms u}]] \big\rangle \big).
\end{gather*}
Putting all of the above together we then arrive at the action
\begin{gather}
S_{\rm 2d}(g, \ms u) = \int_\Sigma \bigg( \frac 12 (a-b) \langle {\rm d}g g^{-1}, \ast {\rm d}g g^{-1} \rangle + \frac 12 ab (a - b) \langle {\rm d}\ms u, \ast {\rm d}\ms u \rangle - a b \langle {\rm d}g g^{-1}, {\rm d}\ms u \rangle\notag
\\ \hphantom{S_{\rm 2d}(g, \ms u) = \int_\Sigma \bigg(}
{} - \frac 16 a^2 b^2 \langle \ms u, [{\rm d}\ms u, {\rm d}\ms u] \rangle \bigg).
\label{Fourth order action}
\end{gather}
It is interesting to note that in the limit $b \to 0$ we recover the principal chiral model action. In~particular, the model with action~\eqref{Fourth order action} can be seen as a deformation of the principal chiral model to which a new $\g$-valued field $\ms u$ is added. In~fact, removing all the terms involving the field $g$ from the above action we are left with the action of the pseudo-dual of the principal chiral model for the field $\ms u$~\cite{Curtright, Nappi, Zakharov}. One may therefore view the action~\eqref{Fourth order action} as coupling together a principal chiral model field $g$ and a pseudo-dual principal chiral model field $\ms u$ in an~integ\-rable~way.

Note that the pseudo-dual of the principal chiral model was derived very recently in~\cite{Bittleston:2020hfv} starting from 6d holomorphic Chern--Simons theory. It~was argued there that such an action could also be derived directly from 4d Chern--Simons theory where $\omega$ is taken to have a fourth order pole but is regular at infinity. By contrast, in the present work we explicitly required~$\omega$ to have a double pole at infinity in~\eqref{Eq:Omega} and then used the right diagonal gauge invariance by~$G$ in~\eqref{gauge tr} to fix the corresponding edge modes at infinity in Section~\ref{sec: removing inf}. We~expect that by~starting instead from a meromorphic $1$-form $\omega$ with a fourth order pole at the origin and which is regular at infinity we would obtain a gauged version of the action~\eqref{Fourth order action}. Moreover, after fixing the gauge invariance by setting $g = \id$ this action should reduce to that of the pseudo-dual of~the principal chiral model field, as in~\cite{Bittleston:2020hfv}.

\subsubsection{Lax connection}

The Lax connection~\eqref{Lax for E-model} takes the form
\begin{align*}
\L &= \bp \big( \Pc_p\big(p^{-1} {\rm d}p\big) - \Ec \Pc_p\big({\ast} p^{-1} {\rm d}p\big) \big)\\
&= \bp \bigg( j^0 + \Ad_g^{-1} {\rm d}\ms u^1 + \frac 1{ab} j^2 + \frac{a-b}{ab} \Ad_g^{-1} \ast {\rm d}\ms u^2 + \frac{a^2 - ab + b^2}{a^2 b^2} \Ad_g^{-1} {\rm d}\ms u^3 + \frac{a-b}{a^2 b^2} \ast j^3 \bigg),
\end{align*}
where we have introduced the shorthand $j \coloneqq g^{-1} {\rm d}g$ and $\bp$ is the inverse~\eqref{Severa map} of $\jb_{\bz'}$ defined in~\eqref{j z def} given explicitly here by
\begin{align*}
\bp \colon\ \ms u^0 + \ms v^1 + \ms w^2 + \ms x^3 &\longmapsto \frac{a^3(\ms u + a \ms v - b^2(\ms w + a \ms x))}{2(b^2-a^2)(z-a)} + \frac{a^3(-\ms u + a \ms v + b^2(\ms w - a \ms x))}{2(b^2-a^2)(z+a)}\\
&\qquad - \frac{b^3 (\ms u + b \ms v - a^2(\ms w + b \ms x))}{2(b^2-a^2)(z-b)}
- \frac{b^3(-\ms u + b \ms v + a^2(\ms w - b \ms x))}{2(b^2-a^2)(z+b)}.
\end{align*}
The above Lax connection therefore explicitly reads
\begin{gather*}
\L = \bigg( a^2 \frac{j_+ + b \Ad_g^{-1} \partial_+ \ms u}{(a+b) (a-z)} +
b^2 \frac{j_+ - a \Ad_g^{-1} \partial_+ \ms u}{(a+b) (b+z)} \bigg) {\rm d}\sigma^+
\\ \hphantom{\L = \bigg(}
{} + \bigg( a^2 \frac{j_- - b \Ad_g^{-1} \partial_- \ms u}{(a+b) (a+z)} + b^2 \frac{j_- + a \Ad_g^{-1} \partial_- \ms u}{(a+b) (b-z)} \bigg) {\rm d}\sigma^-.
\end{gather*}

Note that in the limit $b \to 0$ we recover the Lax connection of the principal chiral model, in the form given in~\eqref{Lax PCM}. This is in agreement with the observation made above about the action~\eqref{Fourth order action}.

The flatness of $\L$ is equivalent to
\begin{subequations} 
\begin{gather}
\label{n eq 4 eom a}
{\rm d}j + \frac 12 [j, j] = 0,\\
\label{n eq 4 eom b}
{\rm d}\big(\Ad_g^{-1} {\rm d}\ms u\big) + \Ad_g^{-1} \big[{\rm d}\ms u, {\rm d}g g^{-1}\big] = 0,\\
\label{n eq 4 eom c}
{\rm d} \ast {\rm d}\ms u + \frac{a b [{\rm d}\ms u, {\rm d}\ms u]}{2(a-b)} + \frac{\big[{\rm d}gg^{-1}, {\rm d}gg^{-1}\big]}{2(a-b)} = 0,\\
\label{n eq 4 eom d}
{\rm d}\big({\ast} {\rm d}gg^{-1}\big) + \frac{ab}{a-b} \big[{\rm d}\ms u, {\rm d}g g^{-1}\big] = 0.
\end{gather}
\end{subequations}
The equations~\eqref{n eq 4 eom a} and~\eqref{n eq 4 eom b} are both identically true off-shell. The first is the Maurer--Cartan equation for $j$ and the second holds because
\begin{align*}
{\rm d}\big(\Ad_g^{-1} {\rm d}\ms u\big) &= {\rm d}\big(g^{-1} {\rm d}\ms u g\big) = - g^{-1} {\rm d}g g^{-1} \wedge {\rm d}\ms u g - g^{-1} {\rm d}\ms u \wedge {\rm d}g\\
&= - \frac 12 \Ad_g^{-1} \big( \big[{\rm d}\ms u, {\rm d}g g^{-1}\big] + \big[dg g^{-1}, {\rm d}\ms u\big] \big) = - \Ad_g^{-1} \big[{\rm d}\ms u, {\rm d}g g^{-1}\big].
\end{align*}
The equations~\eqref{n eq 4 eom c} and~\eqref{n eq 4 eom d} coincide with the equations of motion obtained from the action~\eqref{Fourth order action}, as expected.

\subsection{Real simple zeroes and poles} \label{sec: intro example}

In this final section we discuss the example mentioned in the introduction. Let $\omega$ be given as in~\eqref{omega intro}, namely
\begin{gather*}
\omega = - \ell^\infty_1 \frac{\prod_{i=1}^N (z - \zeta_i)}{\prod_{i=1}^N (z - z_i)} {\rm d}z,
\end{gather*}
where the poles and zeroes are all real and distinct. In~the notation of Section~\ref{sec: defects} we then have $\bz = \bz_{\rm r} = \{ z_i \}_{i=1}^N$ and $\bzeta = \bzeta_{\rm r} = \{ \zeta_i \}_{i=1}^N$. And since all the poles and zeroes are simple we have $n_x = 1$ for all $x \in \bz$ and $m_y = 1$ for all $y \in \bzeta$.

As in the introduction, we shall use the shorthand notation $\epsilon_i \coloneqq \epsilon_{\zeta_i}$ for every $i = 1, \dots, N$ but leave these real non-zero parameters arbitrary.

Since all the poles of $\omega$ are real and simple, the defect Lie algebra~\eqref{defect grp alg inf} is simply given by the direct sum of Lie algebras $\da = \g^{\oplus N}$. The corresponding Lie group is $D = G^{\times N}$. As in the introduction, we shall leave the choice of Lagrangian subalgebra $\k \subset \da$ unspecified, and only assume that it satisfies the technical condition~\eqref{Eq:Direct}. We~denote the corresponding connected Lie subgroup by $K \subset D$.

Since all the zeroes of $\omega$ are real and simple, the real vector space~\eqref{gzeta def} associated with these zeroes is given here by $\gzeta = \g^{\oplus N}$.

The linear isomorphism $\C \colon \gzeta = \g^{\oplus N} \to \da = \g^{\oplus N}$ from Lemma~\ref{lem: Cauchy} is given in components by~\eqref{Cauchy mat}, namely
\begin{gather} \label{simple Cauchy mat}
C^i_{\; j} \coloneqq C^{\up{z_i, 0}}_{\quad\;\; \uz{\zeta_j, 0}} = \frac{1}{z_i - \zeta_j}
\end{gather}
for $i, j = 1, \dots, N$. These are simply the components of the usual Cauchy matrix. It~is well know that the inverse of the Cauchy matrix~\eqref{simple Cauchy mat} has components
\begin{gather} \label{inv Cauchy intro ex}
\big(C^{-1}\big)^i_{\; j} = \frac{\prod_{r \neq i} (\zeta_r - z_j) \prod_r (z_r - \zeta_i)}{\prod_{r \neq j} (z_r - z_j) \prod_{r \neq i} (\zeta_r - \zeta_i)}
\end{gather}
for $i, j = 1, \dots, N$. In~particular, since $\widetilde \Ec \colon \gzeta = \g^{\oplus N} \to \gzeta = \g^{\oplus N}$ defined in~\eqref{Etilde def} is given in components by the diagonal matrix $\diag(\epsilon_1, \dots, \epsilon_N)$, the components of the linear operator $\Ec \colon \da = \g^{\oplus N} \to \da = \g^{\oplus N}$ defined in~\eqref{E def} are given by
\begin{gather*}
\sum_{j=1}^N C^i_{\; j} \epsilon_j \big(C^{-1}\big)^j_{\; k} = \sum_{j=1}^N \epsilon_j \frac{\prod_{r \neq j} (\zeta_r - z_k) \prod_{r \neq i} (z_r - \zeta_j)}{\prod_{r \neq k} (z_r - z_k) \prod_{r \neq j} (\zeta_r - \zeta_j)}.
\end{gather*}
These coincide with the components given in~\eqref{E intro}.

Finally, to compute the Lax connection using the general formula~\eqref{Lax for E-model} we need to compute the inverse $\bp$ in~\eqref{Severa map} of the map $\jb_{\bz'}$ defined in~\eqref{j z def}. The latter reads
\begin{gather*}
\jb_{\bm z'} \colon\ \sum_{j=1}^N \frac{\ms u_j}{z-\zeta_j} \longmapsto \bigg( \sum_{j=1}^N \frac{\ms u_j}{z_i - \zeta_j} \bigg)_{i=1}^N = \bigg( \sum_{j=1}^N C^i_{\; j} \ms u_j \bigg)_{i=1}^N,
\end{gather*}
where the equality is by definition~\eqref{simple Cauchy mat} of the Cauchy matrix. Its inverse is then clearly gi\-ven~by
\begin{gather} \label{p for intro}
\bp \colon\ (\ms v_i)_{i=1}^N \longmapsto \sum_{i,j=1}^N \frac{\big(C^{-1}\big)^i_{\;\; j} \ms v_j}{z-\zeta_i}.
\end{gather}
According to~\eqref{Lax for E-model}, the Lax connection is now given by $\L = \bp(\J_\sigma {\rm d}\sigma + \J_\tau {\rm d}\tau)$, where the $\d$-valued fields $\J_\sigma, \J_\tau \in C^\infty(\Sigma, \d)$ are given in components by
\begin{gather*}
\J_\sigma = (\J^i_\sigma)_{i=1}^N = l^{-1} \partial_\sigma l - \Ec^{-1} \Pc_l \big( \Ec \big(l^{-1} \partial_\sigma l\big) - l^{-1} \partial_\tau l \big),
\\
\J_\tau = (\J^i_\tau)_{i=1}^N = \Ec\big(l^{-1} \partial_\sigma l\big) - \Pc_l \big( \Ec \big(l^{-1} \partial_\sigma l\big) - l^{-1} \partial_\tau l \big).
\end{gather*}
In other words, using the explicit form~\eqref{p for intro} of the linear map $\bp$ we deduce that the Lax connection reads
\begin{gather*}
\L = \sum_{i,j=1}^N \frac{\big(C^{-1}\big)^i_{\;\; j} \big( \J^j_\sigma {\rm d}\sigma + \J^j_\tau {\rm d}\tau \big)}{z-\zeta_i},
\end{gather*}
which corresponds to the expression~\eqref{Lax intro} from the introduction using the explicit inverse of the Cauchy matrix in~\eqref{inv Cauchy intro ex}.

\section{Outlook} \label{sec: outlook}

In this work we constructed a very broad family of integrable $\Ec$-models using the formalism of Costello--Yamazaki~\cite{Costello:2019tri} by starting from the general $2$d action obtained in~\cite{Benini:2020skc}. There are a~number of interesting open problems.

\subsection{Hamiltonian formalism} In this work we focused entirely on constructing the actions of the new family of $2$d integrable field theories. Indeed, the formalism of Costello--Yamazaki~\cite{Costello:2019tri} is most convenient for describing integrable field theories in the Lagrangian formalism.

By contrast, $2$d integrable field theories can be best described in the Hamiltonian formalism using the framework of classical dihedral affine Gaudin models proposed in~\cite{Vicedo:2017cge}, and further developed in~\cite{Delduc:2019bcl, Lacroix:2019xeh}. The formalisms of~\cite{Costello:2019tri} and~\cite{Vicedo:2017cge} were shown to be intimately related in~\cite{Vicedo:2019dej} by performing a Hamiltonian analysis of $4$d Chern--Simons theory. It~would therefore be interesting to perform the Hamiltonian analysis of the family of integrable $\Ec$-model actions described in the present work. In~particular, one should show that the Poisson bracket of~the Lax matrix is of the Maillet $r/s$-form~\cite{Maillet:1985fn, Maillet:1985ek} with twist function, which is equivalent to describing these models as classical dihedral affine Gaudin models. We~will come back to this in a forthcoming paper~\cite{InPrep}.

\subsection[Degenerate E-model]{Degenerate $\boldsymbol{\Ec}$-model}

An important restriction we imposed on the general setting of~\cite{Benini:2020skc} is that the $1$-form $\omega$ had a~double pole at infinity. This allowed us in Section~\ref{sec: removing inf} to partially fix the gauge invariance of the $2$d action of~\cite{Benini:2020skc}, by bringing the component of the edge mode at infinity to the identity.

It would be natural to try to extend our construction to the general setting of~\cite{Benini:2020skc} by allowing arbitrary orders at all the poles of $\omega$. The resulting $2$d integrable field theory would have an~additional gauge invariance and so it is natural to expect that this generalisation would lead to an~integrable family of the class of degenerate $\Ec$-models introduced in~\cite{Klimcik:1996np}, see also~\cite{Klimcik:2018vhl, Klimcik:2019kkf, Sfetsos:1999zm, Squellari:2011dg}, which, in particular, would include by~\cite{Klimcik:2019kkf} the bi-Yang--Baxter model with WZ-term~\cite{Delduc:2017fib}.

\subsection[Integrable E-model hierarchy]{Integrable $\boldsymbol{\Ec}$-model hierarchy}

A crucial step in our analysis was imposing the condition~\eqref{Eq:Lz} on the coefficients in the partial fraction decomposition of the components of the Lax connection $\L$. Indeed, this condition gave a~particular way of satisfying the admissibility condition $(b)$ from Section~\ref{sec: red to 2d} and we showed that within this class of admissible $1$-forms $\L$ there was a unique solution $\L = \L(l)$ to the boundary equation~\eqref{L l constraint} relating $\L$ to the edge mode $l \in C^\infty(\Sigma, D)$. Moreover, the condition~\eqref{Eq:Lz} is at the origin of the introduction of the operator $\Ec$ in our construction.

However,~\eqref{Eq:Lz} is by no means the only way to solve the admissibility condition $(b)$, and it would be very interesting to explore other classes of admissible $1$-forms $\L$. In~the case when the zeroes of $\omega$ are all simple, an obvious alternative way to solve the admissibility condition $(b)$ is to work in a representation of the complex Lie algebra $\g^\CC$, pick $n \in \ZZ_{\geq 1}$ and set
\begin{gather} \label{E-model hierarchy}
\L_{t_n}^{\uz{y, 0}} = \epsilon_{n, y} \big( \L_\sigma^{\uz{y, 0}} \big)^n,
\end{gather}
for every $y \in \bzeta$ and some choice of $\epsilon_{n, y}$. On the left hand side we used the notation $t_n$ instead of~$\tau$ for the time coordinate since we expect the corresponding model to be related to a different flow in the same hierarchy. Indeed, the above solution~\eqref{E-model hierarchy} of the admissibility condition $(b)$ is motivated by the expressions for the Lax matrices inducing higher flows in the integrable hierarchies of affine Gaudin models~\cite{Lacroix:2017isl}. Explicitly, when $\g^\CC$ is of type $B$, for instance, to~each simple zero $y \in \bzeta$ and every odd positive integer $n$ is associated a higher flow $\partial_{t_n}$ with corresponding Lax matrix
\begin{gather*}
\L_{t_n} = \frac{\big( \L_\sigma^{\uz{y, 0}} \big)^n}{z-y}.
\end{gather*}
We therefore expect from~\cite{Lacroix:2017isl} that the Lax connection defined by imposing~\eqref{E-model hierarchy} corresponds to a~higher flow of the same integrable $\Ec$-model hierarchy. It~would be very interesting to investigate this further. In~particular, we expect from~\cite{Lacroix:2017isl}, see also~\cite{Evans:2001sz, Evans:1999mj, Evans:2000qx}, that when $\g^\CC$ is of type~$B$, $C$ or $D$ the condition~\eqref{E-model hierarchy} should give non-trivial commuting flows for all odd $n \in \ZZ_{\geq 1}$, corresponding to the set of exponents of the (untwisted) affine Kac--Moody algebra associated with $\g^\CC$. In~type $A$ we also expect that one should have to modify the ansatz~\eqref{E-model hierarchy} accordingly to produce commuting flows~\cite{Evans:2001sz, Evans:1999mj, Evans:2000qx, Lacroix:2017isl}.

\subsection{3d Chern--Simons theory}

\def\Ann{\raisebox{-.2mm}{\begin{tikzpicture}[scale=0.065]
 \path [draw=none,fill=gray, fill opacity = 0.1] (0,-1) circle (2);
 \path [draw=none,fill=white, fill opacity = 1] (0,-1) circle (.7);
 \draw[solid] (0,-1) circle (.7);
 \draw[solid] (0,-1) circle (2);
\end{tikzpicture}}}
\def\Disc{\raisebox{-.2mm}{\begin{tikzpicture}[scale=0.065]
 \path [draw=none,fill=gray, fill opacity = 0.1] (0,-1) circle (2);
 \draw[solid] (0,-1) circle (2);
\end{tikzpicture}}}

The $\Ec$-model on the infinite cyclinder $S^1 \times \RR$ was shown in~\cite{Severa:2016prq} to arise from $3$d Chern--Simons theory for the Lie group $D$ on the solid cylinder $\Disc \times \RR$, with $\Disc$ a disc, by imposing a suitable boundary condition on the gauge field at the boundary $\partial \, \Disc \simeq S^1$. Moreover, the $\sigma$-model on~$K \backslash D$ was also obtained from $3$d Chern--Simons theory on a hollowed out cylinder $\Ann \times \RR$, with $\Ann$ an annulus, by imposing the same boundary condition as before on the gauge field at the outer boundary and another topological boundary condition depending on the choice of~Lagrangian subalgebra $\k \subset \da$ at the inner boundary.

It would be interesting to understand if, in the integrable case, there is a relation between the above description of the $\sigma$-model on $K \backslash D$ from $3$d Chern--Simons derived in~\cite{Severa:2016prq} and the description from $4$d Chern--Simons theory obtained here.

Another possible connection to $3$d Chern--Simons theory is suggested by the results of~\cite{Schmidtt:2017ngw, Schmidtt:2018hop} where the action of the $\lambda$-deformation~\cite{Sfetsos:2013wia} of the principal chiral model, in the form of the universal $2$d action~\eqref{2d action}, was obtained from a certain ``doubled'' version of $3$d Chern--Simons theory on $\Disc \times \RR$.

\appendix

\section{Proof of Proposition~\ref{PropE}} \label{sec: PropE}

As noted in Section~\ref{sec: operators E Pl}, the direct sum~\eqref{Eq:Direct} is orthogonal with respect to $\langle\!\langle \cdot, \cdot \rangle\!\rangle_{\da, \Ec}$. The corresponding projector $\Pc_l$ is then symmetric with respect to $\langle\!\langle \cdot, \cdot \rangle\!\rangle_{\da, \Ec}$. Hence, for any $\ms U, \ms V \in \da$, we~have
\begin{gather*}
\langle\!\langle \ms U, \tp\Pc_l \ms V \rangle\!\rangle_\da = \langle\!\langle \Pc_l \ms U, \ms V \rangle\!\rangle_\da = \langle\!\langle \Pc_l \ms U, \Ec \ms V \rangle\!\rangle_{\da, \Ec} = \langle\!\langle \ms U, \Pc_l \Ec \ms V \rangle\!\rangle_{\da, \Ec} = \langle\!\langle \ms U, \Ec^{-1} \Pc_l \Ec \ms V \rangle\!\rangle_\da,
\end{gather*}
where we used the symmetry of $\Pc_l$ with respect to $\langle\!\langle \cdot, \cdot \rangle\!\rangle_{\da, \Ec}$ in the third step. Part $(i)$ now follows from the non-degeneracy of $\langle\!\langle \cdot, \cdot \rangle\!\rangle_\da$.

Using this result and the definition~\eqref{Eq:Pb} of $\Pb_l$, we then get
\begin{gather*}
\Pc_l \Ec + \Ec \Pb_l = \Pc_l \Ec - \Ec\, \tp\Pc_l + \Ec = \Ec.
\end{gather*}
This is the first equation in $(ii)$. The second one is simply obtained by multiplying it on both sides by $\Ec^{-1}$.

Part $(iii)$ is easily proved by observing that $\tp\Pb_l \,\Pc_l = (\id - \Pc_l)\Pc_l =0$ since $\Pc_l$ is a projector, and similarly for $\tp \Pc_l \Pb_l = 0$.

Let us now prove $(iv)$. Let $\ms U, \ms V\in \da$. From the definition~\eqref{Eq:Pb} of $\Pb_l$, we have
\begin{align*}
\langle\!\langle \big(\Pc_l - \Pb_l\big) \ms U, \ms V \rangle\!\rangle_\da
&= \langle\!\langle \Pc_l \ms U, \ms V \rangle\!\rangle_\da - \langle\!\langle \ms U, (\id-\Pc_l) \ms V \rangle\!\rangle_\da\\
&= \langle\!\langle \Pc_l \ms U, \Pc_l \ms V + (\id - \Pc_l) \ms V \rangle\!\rangle_\da - \langle\!\langle \Pc_l \ms U + (\id - \Pc_l) \ms U, (\id-\Pc_l) \ms V \rangle\!\rangle_\da\\
&= \langle\!\langle \Pc_l \ms U, \Pc_l \ms V \rangle\!\rangle_\da - \langle\!\langle (\id - \Pc_l) \ms U, (\id-\Pc_l) \ms V \rangle\!\rangle_\da.
\end{align*}
Clearly $(\id-\Pc_l) \ms U$ and $(\id-\Pc_l) \ms V$ belong to $\ker\Pc_l=\Ad_l^{-1}\k$. By the ad-invariance of $\langle\!\langle \cdot, \cdot \rangle\!\rangle_\da$ and the isotropy of $\k$, we thus get $\langle\!\langle (\id - \Pc_l) \ms U, (\id-\Pc_l) \ms V \rangle\!\rangle_\da = 0$, leaving only the first term in the above equation. Moreover, a similar computation can be performed with $\Pc_l$ and $\Pb_l$ exchanged. In~the end, we get
\begin{gather}\label{Eq:PP}
\langle\!\langle (\Pc_l - \Pb_l) \ms U, \ms V \rangle\!\rangle_\da = \langle\!\langle \Pc_l \ms U, \Pc_l \ms V \rangle\!\rangle_\da = - \langle\!\langle \Pb_l \ms U, \Pb_l \ms V \rangle\!\rangle_\da.
\end{gather}
Bringing all operators on the left-hand side of the bilinear form $\langle\!\langle \cdot, \cdot \rangle\!\rangle_\da$ and using its non-degeneracy, we get $\Pc_l - \Pb_l = \tp \Pc_l \Pc_l = - \tp\Pb_l \Pb_l$, proving part $(iv)$.

Part $(v)$ then follows from
\begin{gather}
\Pc_l \tp\Pc_l = \big(1-\tp\Pb_l\big)\big(1-\Pb_l\big) = 1 - \tp\Pb_l - \Pb_l + \tp \Pb_l\Pb_l = \Pc_l - \Pb_l + \Pb_l - \Pc_l = 0,
\end{gather}
and from a similar computation for $\Pb_l\tp\Pb_l$.

Finally, we note that $\Pc_l \ms U$ and $\Pc_l \ms V$ in~\eqref{Eq:PP} belong to $\im\Pc_l = \Ec\Ad_l^{-1}\k$. If $\Ec^2=\id$, the latter is an isotropic subspace and hence $\langle\!\langle \big(\Pc_l - \Pb_l\big) \ms U, \ms V \rangle\!\rangle_\da=\langle\!\langle \Pc_l \ms U, \Pc_l \ms V \rangle\!\rangle_\da =0$. So $(vi)$ follows from the non-degeneracy of $\langle\!\langle \cdot, \cdot \rangle\!\rangle_\da$.

\section{Proof of Proposition~\ref{prop: EM tensor}}\label{App:EMT}

It is convenient to first consider a general model describing a field $l\in C^\infty(\Sigma,D)$, with an action of the form
\begin{gather}
 S(l) = \int_\Sigma \bigg( \frac{1}{2} \langle\!\langle l^{-1} \partial_\tau l, \Oc^{\tau\tau} l^{-1}\partial_\tau l \rangle\!\rangle_\da - \frac{1}{2} \langle\!\langle l^{-1} \partial_\sigma l, \Oc^{\sigma\sigma} l^{-1}\partial_\sigma l \rangle\!\rangle_\da \notag \\ \hphantom{ S(l) = \int_\Sigma \bigg(}
{} + \langle\!\langle l^{-1} \partial_\tau l, \Oc^{\tau\sigma} l^{-1}\partial_\sigma l \rangle\!\rangle_\da \bigg) {\rm d}\sigma \wedge {\rm d}\tau - \frac{1}{2} I^{\rm WZ}_{\da}[l],
\label{Eq:ActionApp}
\end{gather}
where $\Oc^{\tau\tau}$, $\Oc^{\sigma\sigma}$ and $\Oc^{\tau\sigma}$ are linear operators on $\da$, which can depend on the field $l$ but not on its derivatives. Without loss of generality, we can suppose $\Oc^{\tau\tau}$ and $\Oc^{\sigma\sigma}$ symmetric with respect to $\langle\!\langle \cdot, \cdot \rangle\!\rangle_\da$. Note that the model~\eqref{Eq:ActionApp} is relativistic if and only if $\Oc^{\tau\tau}=\Oc^{\sigma\sigma}$ and $\Oc^{\tau\sigma}$ is skew-symmetric with respect to $\langle\!\langle \cdot, \cdot \rangle\!\rangle_\da$. The goal of this appendix is to prove the following result.

\begin{Lemma}\label{PropEMT}
The components of the energy-momentum tensor of the model~\eqref{Eq:ActionApp} are given by
\begin{gather*}
T^\tau_{\;\;\,\tau} = -T^\sigma_{\;\;\,\sigma} = \frac{1}{2} \langle\!\langle l^{-1} \partial_\tau l, \Oc^{\tau\tau} l^{-1}\partial_\tau l \rangle\!\rangle_\da + \frac{1}{2} \langle\!\langle l^{-1} \partial_\sigma l, \Oc^{\sigma\sigma} l^{-1}\partial_\sigma l \rangle\!\rangle_\da,
\\
T^\tau_{\;\;\,\sigma}= \langle\!\langle l^{-1} \partial_\tau l, \Oc^{\tau\tau} l^{-1}\partial_\sigma l \rangle\!\rangle_\da + \langle\!\langle l^{-1} \partial_\sigma l, \Oc^{\tau\sigma} l^{-1}\partial_\sigma l \rangle\!\rangle_\da,
\\
T^\sigma_{\;\;\,\tau}= -\langle\!\langle l^{-1} \partial_\tau l, \Oc^{\sigma\sigma} l^{-1}\partial_\sigma l \rangle\!\rangle_\da + \langle\!\langle l^{-1} \partial_\tau l, \Oc^{\tau\sigma} l^{-1}\partial_\tau l \rangle\!\rangle_\da.
\end{gather*}
\end{Lemma}

\begin{proof}
Let us fix a basis $\lbrace	I_A \rbrace_{A=1,\dots,\dim\da}$ of the Lie algebra $\da$ and a choice of local coordinates $\big\lbrace y^M \big\rbrace_{M=1,\dots,\dim\da}$ on the group manifold $D$, parametrising the element $l$ in $D$. We~define the vielbeins $\e AM$ through the decomposition
\begin{gather*}
l^{-1} \frac{\partial l}{\partial y^M} = \e AM I_A.
\end{gather*}
We may then express the components $l^{-1} \partial_\mu l$ for $\mu=\tau,\sigma$ of the Maurer--Cartan current as
\begin{gather}\label{Eq:MaurerVielbeins}
l^{-1} \partial_\mu l = \e AM \,\partial_\mu y^M \, I_A.
\end{gather}
For any linear operator $\Oc$ on $\da$, let us also define
\begin{gather*}
\Oc_{AB} \coloneqq \langle\!\langle I_A, \Oc I_B \rangle\!\rangle_\da.
\end{gather*}
The symmetry of $\Oc^{\tau\tau}$ and $\Oc^{\sigma\sigma}$ then translates to the fact that $\Oc^{\tau\tau}_{AB}=\Oc^{\tau\tau}_{BA}$ and $\Oc^{\sigma\sigma}_{AB}=\Oc^{\sigma\sigma}_{BA}$ for every $A,B = 1,\dots,\dim\da$. Using the above definitions, we then rewrite the action~\eqref{Eq:ActionApp} as the integral over $\Sigma$ of the Lagrangian density
\begin{gather*}
L = \e AM \e BN \bigg( \frac{1}{2} \Oc^{\tau\tau}_{AB}\, \partial_\tau y^M \partial_\tau y^N - \frac{1}{2} \Oc^{\sigma\sigma}_{AB}\, \partial_\sigma y^M \partial_\sigma y^N + (\Oc^{\tau\sigma}_{AB}+W_{AB})\, \partial_\tau y^M \partial_\sigma y^N \bigg),
\end{gather*}
where $W_{AB}=-W_{BA}$ describes the contribution of the Wess--Zumino term to the action (such a rewriting of the Wess--Zumino term as a two-dimensional integral is always possible, at least locally~--- as we shall see, an explicit expression for $W_{AB}$ will not be needed in what follows).

The energy-momentum tensor of the model can be computed explicitly in terms of the Lag\-rangian density $L$ as
\begin{gather}\label{Eq:TGen}
T^\mu_{\;\;\,\nu} = \frac{\partial L}{\partial(\partial_\mu y^M)} \partial_\nu y^M - \delta^\mu_{\;\,\nu} \,L.
\end{gather}
From the above expression of $L$, the symmetry of $\Oc^{\tau\tau}_{AB}$ and $\Oc^{\sigma\sigma}_{AB}$ and the skew-symmetry of~$W_{AB}$, we get
\begin{subequations}
\begin{gather}
\frac{\partial L}{\partial(\partial_\tau y^M)} = \e AM \e BN \big( \Oc^{\tau\tau}_{AB} \partial_\tau y^N + (\Oc^{\tau\sigma}_{AB}+W_{AB}) \partial_\sigma y^N \big), \label{Eq:DerL}
\\
\frac{\partial L}{\partial(\partial_\sigma y^M)} = \e AM \e BN \big({-}\Oc^{\sigma\sigma}_{AB} \partial_\sigma y^N + (\Oc^{\tau\sigma}_{BA}-W_{AB}) \partial_\tau y^N \big).
\end{gather}
\end{subequations}
We then deduce that
\begin{gather*}
T^\tau_{\;\;\,\tau} = \frac{1}{2} \e AM \e BN \big( \Oc^{\tau\tau}_{AB} \, \partial_\tau y^M \partial_\tau y^N + \Oc^{\sigma\sigma}_{AB} \, \partial_\sigma y^M \partial_\sigma y^N \big).
\end{gather*}
Note in particular that the parts proportional to $\Oc^{\tau\sigma}_{AB}+W_{AB}$ coming from the two terms in the right-hand side of~\eqref{Eq:TGen} cancel for $\mu=\nu=\tau$. A similar computation yields $T^\sigma_{\;\;\,\sigma} = -T^\tau_{\;\;\,\tau}$. In~terms of the Maurer--Cartan currents~\eqref{Eq:MaurerVielbeins}, we find
\begin{gather*}
T^\tau_{\;\;\,\tau} = -T^\sigma_{\;\;\,\sigma} = \frac{1}{2} \langle\!\langle l^{-1} \partial_\tau l, \Oc^{\tau\tau} l^{-1}\partial_\tau l \rangle\!\rangle_\da + \frac{1}{2} \langle\!\langle l^{-1} \partial_\sigma l, \Oc^{\sigma\sigma} l^{-1}\partial_\sigma l \rangle\!\rangle_\da.
\end{gather*}

Let us now compute the crossed-term $T^\tau_{\;\;\,\sigma}$. Substituting the expression~\eqref{Eq:DerL} in~\eqref{Eq:TGen}, we obtain
\begin{gather*}
T^\tau_{\;\;\,\sigma} = \e AM \e BN \big( \Oc^{\tau\tau}_{AB} \, \partial_\sigma y^M \partial_\tau y^N + (\Oc^{\tau\sigma}_{AB}+W_{AB})\, \partial_\sigma y^M \partial_\sigma y^N \big).
\end{gather*}
The term containing $W_{AB}$ vanishes, as it is given by the contraction of the symmetric tensor $\e AM \e BN \partial_\sigma y^M \partial_\sigma y^N $ with the skew-symmetric tensor $W_{AB}$. One then rewrites the resulting expression in terms of the Maurer--Cartan currents~\eqref{Eq:MaurerVielbeins} as
\begin{gather*}
T^\tau_{\;\;\,\sigma} = \langle\!\langle l^{-1} \partial_\tau l, \Oc^{\tau\tau} l^{-1}\partial_\sigma l \rangle\!\rangle_\da + \langle\!\langle l^{-1} \partial_\sigma l, \Oc^{\tau\sigma} l^{-1}\partial_\sigma l \rangle\!\rangle_\da.
\end{gather*}
A similar computation leads to the announced expression of $T^\tau_{\;\;\,\sigma}$.
\end{proof}

Let us now turn to the proof of Proposition~\ref{prop: EM tensor}. In~the above notations, the action~\eqref{Eq:ActionE} corresponds to the choice of operators
\begin{gather}\label{Eq:OE}
\Oc^{\tau\tau} = \Ec^{-1} \Pc_l, \qquad
\Oc^{\sigma\sigma} = \Ec \overline{\Pc}_l \qquad \text{and} \qquad
\Oc^{\tau\sigma} = \overline{\Pc}_l - \null^t \Pc_l.
\end{gather}
Applying Lemma~\ref{PropEMT}, we then get an explicit expression of the components $T^\mu_{\;\;\,\nu}$ of the energy-momentum tensor. For~instance, we have
\begin{gather*}
T^\tau_{\;\;\,\tau} = -T^\sigma_{\;\;\,\sigma} = \frac{1}{2} \langle\!\langle l^{-1} \partial_\tau l, \Ec^{-1} \Pc_l\, l^{-1}\partial_\tau l \rangle\!\rangle_\da + \frac{1}{2} \langle\!\langle l^{-1} \partial_\sigma l, \Ec \overline{\Pc}_l\, l^{-1}\partial_\sigma l \rangle\!\rangle_\da.
\end{gather*}
On the other hand, from the definition~\eqref{Eq:J} of $\J_\sigma$ and the symmetry of $\Ec$, we get
\begin{gather*}
\frac{1}{2} \langle\!\langle \J_\sigma, \Ec \J_\sigma \rangle\!\rangle_\da = \frac{1}{2} \langle\!\langle \Pc_l\, l^{-1} \partial_\tau l, \Ec^{-1} \Pc_l\, l^{-1}\partial_\tau l \rangle\!\rangle_\da + \frac{1}{2} \langle\!\langle \overline{\Pc}_l\, l^{-1} \partial_\sigma l,\Ec \overline{\Pc}_l \, l^{-1} \partial_\sigma l \rangle\!\rangle_\da
\\ \hphantom{\frac{1}{2} \langle\!\langle \J_\sigma, \Ec \J_\sigma \rangle\!\rangle_\da =}
{} + \langle\!\langle \Pc_l\, l^{-1} \partial_\tau l, \Pb_l \, l^{-1} \partial_\sigma l \rangle\!\rangle_\da.
\end{gather*}
The last line vanishes due to part $(iii)$ of Proposition~\ref{PropE}. Moreover, from part $(i)$, we have
\begin{gather*}
\tp\Pc_l \Ec^{-1} \Pc_l = \Ec^{-1} \Pc_l^2 = \Ec^{-1} \Pc_l,
\end{gather*}
where we have used the fact that $\Pc_l$ is a projector and thus that $\Pc_l^2=\Pc_l$. A similar computation yields $\tp \overline{\Pc}_l\Ec \overline{\Pc}_l=\Ec \overline{\Pc}_l$. Thus, we get
\begin{gather*}
\frac{1}{2} \langle\!\langle \J_\sigma, \Ec \J_\sigma \rangle\!\rangle_\da = T^\tau_{\;\;\,\tau} = -T^\sigma_{\;\;\,\sigma}.
\end{gather*}
Let us now turn our attention to $T^\tau_{\;\;\,\sigma}$. Applying Lemma~\ref{PropEMT} with the operators~\eqref{Eq:OE}, we get
\begin{gather*}
T^\tau_{\;\;\,\sigma} = \langle\!\langle l^{-1} \partial_\tau l, \Ec^{-1} \Pc_l\, l^{-1}\partial_\sigma l \rangle\!\rangle_\da + \langle\!\langle l^{-1} \partial_\sigma l, \big(\overline{\Pc}_l - \Pc_l\big) l^{-1}\partial_\sigma l \rangle\!\rangle_\da,
\end{gather*}
where in particular we transposed the operator $\tp\Pc_l$ in the last term. On the other hand, it follows from~\eqref{Eq:J} that
\begin{gather*}
\frac{1}{2} \langle\!\langle \J_\sigma, \J_\sigma \rangle\!\rangle_\da = \frac{1}{2} \langle\!\langle \Ec^{-1} \Pc_l\, l^{-1} \partial_\tau l, \Ec^{-1} \Pc_l\, l^{-1}\partial_\tau l \rangle\!\rangle_\da + \frac{1}{2} \langle\!\langle \overline{\Pc}_l\, l^{-1} \partial_\sigma l, \overline{\Pc}_l \, l^{-1} \partial_\sigma l \rangle\!\rangle_\da
\\ \hphantom{\frac{1}{2} \langle\!\langle \J_\sigma, \J_\sigma \rangle\!\rangle_\da =}
{} + \langle\!\langle \Ec^{-1} \Pc_l \, l^{-1} \partial_\tau l, \overline{\Pc}_l\, l^{-1} \partial_\sigma l \rangle\!\rangle_\da.
\end{gather*}
Note from part $(i)$ of Proposition~\ref{PropE} that $\Ec^{-1} \Pc_l$ is symmetric. Using also part $(v)$ we get
\begin{gather*}
\tp\bigl( \Ec^{-1} \Pc_l \bigr) \Ec^{-1} \Pc_l = \Ec^{-1} \Pc_l \,\tp\bigl( \Ec^{-1} \Pc_l \bigr) = \Ec^{-1} \Pc_l \tp\Pc_l \Ec^{-1} = 0,
\end{gather*}
so that the first term in $\frac{1}{2} \langle\!\langle \J_\sigma, \J_\sigma \rangle\!\rangle_\da$ vanishes. Using the same properties, we also get
\begin{gather*}
\tp\bigl( \Ec^{-1} \Pc_l \bigr) \Pb_l = \Ec^{-1} \Pc_l \Pb_l = \Ec^{-1} \big(\Pc_l - \Pc_l \tp\Pc_l\big) = \Ec^{-1} \Pc_l.
\end{gather*}
Finally, using $\tp\Pb_l\Pb_l = \Pb_l-\Pc_l$ (see part $(iv)$ of Proposition~\ref{PropE}), we find that
\begin{gather*}
\frac{1}{2} \langle\!\langle \J_\sigma, \J_\sigma \rangle\!\rangle_\da = T^\tau_{\;\;\,\sigma}.
\end{gather*}
A similar computation yields $- \frac{1}{2} \langle\!\langle \J_\sigma, \Ec^2 \J_\sigma \rangle\!\rangle_\da = T^\sigma_{\;\;\,\tau}$.

\subsubsection*{Acknowledgements}

S.L.~would like to thank B.~Hoare for useful discussions.
The work of S.L.~is funded by the Deutsche Forschungsgemeinschaft (DFG, German Research Foundation) under Germany's Excel\-lence Strategy~-- EXC 2121 ``Quantum Universe''~-- 390833306.

\pdfbookmark[1]{References}{ref}
\LastPageEnding

\end{document}